\documentclass[12pt]{article}
\usepackage{eurosym}
\usepackage[colorlinks,hyperfootnotes=false,pageanchor=false]{hyperref}
\usepackage{amsmath, amsthm}
\usepackage{algorithm}
\usepackage[noend]{algpseudocode}
\usepackage[nameinlink,noabbrev]{cleveref}
\usepackage[doublespacing]{setspace}
\usepackage{placeins}
\usepackage{graphicx}
\usepackage{pgfplots}
\usepackage{tikz}
\usepackage{relsize}
\usepackage{pdflscape}
\usepackage{amssymb}
\usepackage[round]{natbib}
\usepackage{amsfonts}
\usepackage{amsmath}
\usepackage{bm}
\usepackage{latexsym}
\usepackage{epsfig}
\usepackage{rotating}
\usepackage{array}
\usepackage{animate}

\hypersetup{colorlinks=true, linkcolor=black, citecolor=black}
\pgfplotsset{compat=newest}
\pgfplotsset{plot coordinates/math parser=false}
\newlength\figureheight
\newlength\figurewidth

\def\1{1\!{\rm l}}

\newtheorem{theorem}{Theorem}
\newtheorem{lemma}[theorem]{Lemma}
\newtheorem{proposition}{Proposition}
\newtheorem{assumption}{Assumption}

\newtheorem{remark}{Remark}

\newcommand{\ignore}[1]{}

\newcommand{\y}{\mathbf{y}}

\newcommand{\dx}{\text{d}}

\DeclareMathOperator*{\plim}{plim}
\DeclareMathOperator*{\argmax}{arg\,max}

\def\spacingset#1{\renewcommand{\baselinestretch}{#1}\small}
\spacingset{1}
\newcolumntype{C}[1]{>{\centering\let\newline\\\arraybackslash\hspace{0pt}}m{#1}}
\newcolumntype{L}[1]{>{\raggedright\let\newline\\\arraybackslash\hspace{0pt}}m{#1}}

\begin{document}

\title{Focused Bayesian Prediction\thanks{{\footnotesize We would like to
thank a co-editor and two anonymous referees for very constructive comments
on an earlier draft of the paper. We would also like to thank various
participants at the following workshops and conferences for very helpful
comments on earlier versions of the paper: the Casa Matem\'{a}tica Oaxaca
Workshop on `Computational Statistics and Molecular Simulation: A Practical
Cross-Fertilization', Oaxaca, November, 2018; the Australian Centre for
Excellence in Mathematics and Statistics Workshop on `Advances and
Challenges in Monte Carlo Methods', Brisbane, November, 2018; the
International Symposium of Forecasting, Thessaloniki, June, 2019; the 13th
RCEA Bayesian Econometrics Workshop, Cyprus, June 2019; the 12th
International Conference on `Monte Carlo Methods and Applications', Sydney,
July, 2019; the Joint Statistical Meetings, Colorado, July, 2019; the
`Frontiers in Research for Statistics' Conference, Brisbane, October, 2019;
and the `Statistical Methods in Data Science Workshop', Mathematical
Research Institute, Creswick, December, 2019. This research has been
supported by Australian Research Council (ARC) Discovery Grants DP170100729
and DP200101414. Frazier was also supported by ARC Early Career Researcher
Award DE200101070.}}}
\author{Ruben Loaiza-Maya, Gael M. Martin\thanks{
{\footnotesize Corresponding author: gael.martin@monash.edu.}} \ and David
T. Frazier \medskip  \and \textit{Department of Econometrics and Business
Statistics, Monash University\smallskip } \and \textit{and Australian Centre
of Excellence in Mathematics and Statistics}}
\maketitle
\title{}

\begin{abstract}
{We} propose a new method for conducting Bayesian prediction that delivers
accurate predictions without correctly specifying the unknown true data
generating process. {A prior is defined over a class} of plausible
predictive models. {After observing data, we update} {the prior to} a
posterior over these models, via a criterion that captures a user-specified
measure of predictive accuracy. Under regularity, this{\ update} yields
posterior concentration onto the element of the predictive class that
maximizes the expectation of the accuracy measure. In a series of simulation
experiments and empirical examples we find notable gains in predictive
accuracy relative to conventional likelihood-based{\ prediction.}

\bigskip


\bigskip

\emph{MSC2010 Subject Classification}: 62F15, 60G25, 62M20\smallskip

\end{abstract}

\newpage

\section{Introduction}

\baselineskip18pt

{Bayesian prediction quantifies uncertainty about the future value of a
random variable using the rules and language of probability. A probability
distribution for a future value is produced, conditioned only on past
observations; all uncertainty about the parameters of the prediction model,
plus any uncertainty about the model itself, having been }integrated{{, or
averaged out} via these simple rules. Inherent to this natural and coherent
approach to prediction, however, is the assumption that the process that has
generated the observed data is either equivalent to the particular model on
which we condition, or contained in the set of models over which we average.
Such a heroic assumption is clearly at odds with reality}; in particular in
the realm of the{\ social and {economic} sciences where statistical data
arises through complex} processes {that we can only \textit{ever} intend to
approximate. }

{In response to this limitation of the conventional paradigm, we propose {an
alternative }approach to Bayesian prediction. }{A prior is placed over a
class of \textit{plausible} predictive models. The prior is then updated to
a posterior over these models, via a criterion function that represents a {%
user-specified measure of predictive accuracy. This criterion replaces the
likelihood function in the conventional Bayesian update and, hence, obviates
the explicit need for correct model specification. }Summarization of the
posterior so produced - via its mean, for example - yields a single,
representative predictive distribution that is expressly designed to yield
accurate forecasts according to the given measure. Alternatively, the full
extent of the posterior variation that obtains can be quantified and
visualized. }Given this {deliberate} \textit{focus} on a particular aspect
of predictive performance in the building of {predictions}, we refer to the
principle as focused Bayesian prediction, or simply FBP.

{To quantify predictive accuracy }{we use the concept of a scoring rule.
(See \citealp{gneiting2007probabilistic}, and \citealp{gneiting2007strictly}%
, for early expositions.) In short, a scoring rule rewards a probabilistic
forecast for assigning a high density ordinate (or high probability mass) to
the observed value (`calibration'), subject to some criterion of
`sharpness', or some reward for accuracy in a particular part of the
predictive support (e.g. the tails; \citealp{diks2011likelihood}, %
\citealp{opschoor2017combining}). }Under appropriate regularity, we
establish that this approach ensures, asymptotically, accurate performance
according to the {specified} measure of predictive accuracy, and without
dependence on correct model choice. {Extensive numerical results support the
theoretical results: }focus{\ on predictive accuracy, rather than correct
model specification \textit{per se}, leads to improved predictive
performance.}

{This }{approach to Bayesian prediction{\ has elements in common with the
`probably approximately correct' (PAC)-Bayes approach to prediction adopted
in the machine learning literature; see \cite{guedj2019} for a recent review 
}{and extensive referencing}. {The} use of Bayesian updating \textit{per se}
without reference to a likelihood function {also echoes} the \textit{%
generalized inferential} methods proposed by, {for example}},{\ \cite%
{bissiri2016general},}\textbf{\ }\cite{giummole2017objective}, {\cite%
{GVI2019} }{and \cite{Syring2019},} in which uncertainty about unknown
parameters (in a given model) is {updated} via a general loss, or score,
function. A major challenge {in these generalizations of the standard
Bayesian paradigm} is the calibration of the scale of the loss (or score),
which has a direct impact on the resultant variance of posterior of the
parameters. Several methods for specifying this scale have been proposed (%
\citeauthor{bissiri2016general}; \citeauthor{giummole2017objective}; %
\citealp{holmes2017assigning}; \citealp{lyddon2019general}; {%
\citeauthor{Syring2019}}).\footnote{%
We note that \cite{CH03} also propose the use of quasi-posteriors based on
updating a general loss function; however, the loss function is not
calibrated as it is in the other literature referenced here.} Whilst we draw
some insights from this literature, we propose new approaches that are
informed specifically by the prediction context in which we are working, and
which ensure posterior concentration around the predictive that is optimal
under the given accuracy measure.

The predictive distributions within the plausible class {may characterize a
single dynamic structure} depending on a single vector of unknown parameters.%
\footnote{{A related frequentist literature exists in which distributional
forecasts produced via a particular model }are `optimized'{\ according to} a
given form of predictive accuracy. Work in which this idea is implemented,
or at least discussed, {includes} \cite{gneiting2005calibrated}, \cite%
{gneiting2007strictly}, \cite{elliott2008economic} {and} \cite%
{patton2019comparing}.} {However, they may also be} weighted combinations of
predictives from distinct {models}. {As such, our approach represents a
coherent Bayesian }method for{\ estimating weighted combinations of
predictives via forecast accuracy criteria, and without the need to assume
that the true model lies within the set of constituent predictives. Whilst
an established literature on estimating mixtures of predictives exists (see }%
\citealp{aastveit2018evolution}, for an extensive review) - including work
that invokes Bayesian principles - our {paper provides an alternative way of
updating predictive combinations via non-likelihood-based Bayesian principles%
}. {We comment further on the connection of our work with the literature on
predictive combinations, plus provide detailed referencing to this literature%
}, in Section \ref{comb}.

{After establishing the theoretical validity of {the new} method, its}{\
efficacy and usefulness }is demonstrated through a set of simulation
exercises, based on alternative predictive classes for a stochastic
volatility model for financial returns. These classes are deliberately
chosen to represent, at one end, a very misspecified representation of the
(known) true data generating process (DGP) and, at the other end, a less
misspecified version. The comparator in all cases is the standard, and
misspecified, likelihood-based Bayesian {update} of the given parametric
class. The results are clear: within-sample updating based on a specific
measure of predictive accuracy {almost always }leads to the best
out-of-sample performance according to that measure. The \textit{degree} of
superiority depends on the interplay between the model class, {including}
the manner in which the model is misspecified, and the desired measure of
accuracy - with animated graphics used to illustrate this point. The
differential impact of {update} choice on posterior variation is also
highlighted, via an animated display of posterior distributions for the
expected shortfall (ES) of both `long' and `short' portfolios in the
financial asset.

{Two }empirical illustrations complete the analysis. In the first, we
predict two{\ different series of daily }financial returns using predictive
classes based on the Gaussian (generalized) autoregressive conditional
heteroscedastic ((G)ARCH) class of volatility model, known to be
misspecified for the more complex process driving returns. The series
considered are returns on: the U.S. dollar currency index{, and the }S\&P500
stock index. The empirical results mimic those produced by simulation, with {%
predictive accuracy improved by using the focused }update - rather than the
conventional (likelihood-based) update - {in virtually all cases. The
increase in predictive accuracy translates into more accurate value at risk
(VaR) forecasts: }use of an {update} that focuses on tail accuracy leads to
a better match of empirical to nominal VaR coverage (than does the
likelihood {update}) in \textit{all} cases, and more frequent support of the
joint null of correct coverage and independent violations. {Improved
forecasts of ES also result in many cases.}

{In the} {second} empirical example we pit FBP against the best performers
in the Makridakis 4 (M4) forecasting competition. We perform the exercise
using the {23,000} annual time series from the set of {100,000 }series (of
varying frequencies) used in the competition. We select as the predictive
class, the exponential smoothing model of \cite{hyndman2002state} (referred
to as ETS hereafter), {which had }ranked highly {amongst} all twenty-five
competitors. Adopting the same preliminary model selection procedure as the
authors to specify the components of the ETS class for each of the {23,000
series}, we {update} the chosen class using the mean scaled interval score
(MSIS). This measure of predictive accuracy penalizes a prediction interval
if the observed value falls outside the interval (appropriately weighted by
its nominal coverage), and rewards a narrow interval, and was one of the
measures used to rank methods in the competition. As measured by MSIS, FBP
not only almost always outperforms maximum likelihood-based implementation
of ETS, but it outperforms all four predictive methods that were previously
ranked best in the competition, in a large number of cases.

The rest of the article is organized as follows. In Section~\ref%
{sec:FocusedBayesianPrediction} we propose our new Bayesian predictive
paradigm{, and briefly illustrate its ability to produce more accurate
predictions using a toy example}. {In Section \ref{theory}} asymptotic
validation of the method is provided, under the required regularity. More
extensive illustration of the {new approach} via simulation, and
visualization of the results, is the content of Section~\ref%
{section:simulaexcer}, whilst Section~\ref{empapp} illustrates the power of
the method in empirical settings. In Section~\ref{section:discuss} we
discuss the implications of our results and future lines of research. The
proofs of all theoretical results, and certain computational details, are
provided in appendices to the paper.

\section{Focused Bayesian Prediction\label{sec:FocusedBayesianPrediction}}

\subsection{Preliminaries and notation\label{prelim}}

Consider a stochastic process $\{y_{t}:\Omega \rightarrow \mathbb{R},t\in 
\mathbb{N}\}$ defined on the complete probability space $(\Omega ,\mathcal{F}%
,G)$. Let $\mathcal{F}_{t}:=\sigma (y_{1},\dots ,y_{t})$ denote the natural
sigma-field, and let $G$ denote the infinite-dimensional distribution of the
sequence $y_{1},y_{2},\dots $.

Throughout, we focus {on one-step-ahead} predictions and let $\mathcal{P}%
^{t}:=\left\{ P_{\boldsymbol{\theta }}^{t}:\boldsymbol{\theta }\in \Theta
\right\} $ denote a generic class of one-step-ahead predictive models {for }$%
y_{t+1}$, which are {conditioned} on time $t$ information $\mathcal{F}_{t}$,
and where the generic elements of $\mathcal{P}^{t}$\ are represented by $P_{%
\boldsymbol{\theta }}^{t}(\cdot ):=P(\cdot |\boldsymbol{\theta },\mathcal{F}%
_{t}).$ The {\ parameter} $\boldsymbol{\theta }$ indexes values in the
predictive class, with $\boldsymbol{\theta }$ defined on $(\Theta ,\mathcal{T%
},\Pi )$, {and where }$\Pi ${\ measures our {beliefs about} }$\boldsymbol{%
\theta }${. }{Our beliefs }$\Pi ${\ over }$\Theta ${\ - both prior and
posterior - generate corresponding beliefs over the elements in the
predictive class $\mathcal{P}^{t}${, in the usual manner} and, therefore,
throughout the remainder we abuse notation and refer to $\Pi $ as {indexing}
beliefs over the class $\mathcal{P}^{t}$. }

Our goal is to construct a sequence of probability measures over $\mathcal{P}%
^{t}$, starting from our prior beliefs $\Pi $, such that {hypotheses} in $%
\mathcal{P}^{t}$ that have `higher predictive accuracy', are given higher
posterior probability, after observing realizations from $\{y_{t}:t\geq 1\}$%
. Given a user-defined measure of accuracy, we demonstrate that such a
probability measure can be constructed using a Bayesian updating framework.

Importantly however, we deviate from the standard approach to the production
of Bayesian predictives in that the class $\mathcal{P}^{t}$ only represents 
\textit{plausible} predictive models for $y_{t+1}$. At no point in what
follows do we make the unrealistic assumption that the true one-step-ahead
predictive is contained in $\mathcal{P}^{t}$.\footnote{{The treatment of
scalar $y_{t}$ and one-step-ahead prediction is for the purpose of
illustration only, and all the methodology that follows can easily be
extended to multivariate $y_{t}$ and {multi-step-ahead} prediction in the
usual manner.}}

\subsection{Bayesian updating based on scoring rules}

Using generic notation for the moment,\ for $\mathcal{P}$ a convex class of
predictive distributions on $(\Omega ,\mathcal{F})$, we measure the
predictive accuracy of $P\in \mathcal{P}$ using the scoring rule $S:\mathcal{%
P}\times \Omega \rightarrow \mathbb{R}$, whereby if the predictive
distribution $P$ is quoted and the value $y$ eventuates, then the reward, or
positively-oriented `score', is $S(P,y).$ The expected score under the{\
true predictive} $G$ {is }defined as 
\begin{equation}
\mathbb{S}(\cdot ,G):=\int_{y\in \Omega }S(\cdot ,y)dG(y).  \label{exp_score}
\end{equation}%
We say that a scoring rule is \textit{proper} relative to $\mathcal{P}$ if,
for all $P,G\in \mathcal{P}$, $\mathbb{S}(G,G)\geq \mathbb{S}(P,G),$ and is
strictly proper, relative to $\mathcal{P}$, if $\mathbb{S}(G,G)=\mathbb{S}%
(P,G)\iff P=G.$ That is, a proper scoring rule is one whereby if the
forecaster's best judgment is indeed the true measure $G$ there is no
incentive to quote anything other than $P=G$ {\citep{gneiting2007strictly}}.

{Under the assumption that a given predictive class contains the truth},{\
i.e. that} $G\in \mathcal{P}$, {the expectation of any\textit{\ }proper
score $S(\cdot ,y)$, with respect to the truth $(G)$, will be maximized at
the truth, $G$. Hence, {maximization over }$\mathcal{P}$ {of the expected
scoring rule} $\mathbb{S}(\cdot ,G)$, will reveal the true predictive
mechanism when it is contained in $\mathcal{P}$. In practice of course, the
expected score $\mathbb{S}(\cdot ,G)$ is unattainable, and a sample estimate
based on observed data is used to define a sample score-based criterion.
Maximization of the sample criterion, which implicitly depends on the true
predictive process through the observed data,\textbf{\ }will yield the
member of the predictive class that maximizes the relevant sample criterion. 
}However, asymptotically the true predictive{\ distribution will be
recovered via \textit{any} proper score criterion {(again, on the assumption
that the true predictive lies in the class $\mathcal{P}$).}}

The very premise of this paper is that, in reality, any choice of predictive
class$\,$is such that the truth is not contained therein,\ at which point
there is no reason to presume that the expectation of any particular scoring
rule will be maximized at the truth or, indeed, maximized by the \textit{same%
} predictive distribution that maximizes a different (expected) score. This
does not, however, preclude the meaningfulness of a score as a measure of
predictive accuracy, or invalidate the goal of seeking accuracy according to
this particular measure. Indeed, it renders the distinctiveness of different
scoring rules, and what form of forecast accuracy they do and do not reward,
even more critical, and provides strong justification for driving predictive
decisions by the very score that matters for the problem at hand.

With these insights in mind, we proceed as follows, reverting now to the
specific notation that characterizes our problem, as defined in Section \ref%
{prelim}. Given observed data $\mathbf{y}_{n}=(y_{1},y_{2},...,y_{n})^{%
\prime }$, our object of interest is $P_{\boldsymbol{\theta }}^{n}$, that
is, the predictive distribution for $y_{n+1}$, conditional on information
known at time $n$, $\mathcal{F}_{n}$. Given our prior beliefs $\Pi (\cdot )$%
, over $\mathcal{P}^{n}$, we update these beliefs using the following
coherent posterior measure: for $A\subset \mathcal{P}^{n}$, 
\begin{equation}
\Pi _{w}(A|\mathbf{y}_{n})=\frac{\int_{A}\exp \left[ w_{n}S_{n}\left( P_{%
\boldsymbol{\theta }}^{n}\right) \right] \text{d}\Pi \left( P_{\boldsymbol{%
\theta }}^{n}\right) }{\int_{\Theta }\exp \left[ w_{n}S_{n}\left( P_{%
\boldsymbol{\theta }}^{n}\right) \right] \text{d}\Pi \left( P_{\boldsymbol{%
\theta }}^{n}\right) },  \label{post}
\end{equation}%
where 
\begin{equation}
S_{n}(P_{\boldsymbol{\theta }}^{n})=\sum\limits_{t=0}^{n-1}S(P_{\boldsymbol{%
\theta }}^{t},y_{t+1}),  \label{score}
\end{equation}%
and where the scale factor\textbf{\ }$w_{n}$ ({indexed by }$n$){, used to
define and index the posterior (via the notation }$\Pi _{w}(\cdot |\mathbf{y}%
_{n})$){, is to be discussed in detail below.\footnote{%
The nature of the conditioning set $\mathcal{F}_{n}$ differs according to
the dynamic structure of $P_{\boldsymbol{\theta }}^{n}.$ For example, in a
Markov model of order 1, $\mathcal{F}_{n}$ comprises the observed $y_{n}$
only. In contrast, a predictive for a long memory model conditions on all
available past observations. The conditioning set, $\mathcal{F}_{n}$, may
also, of course, include observed values of covariates. Hence, we keep the
notation for this conditioning set, $\mathcal{F}_{n}$, {which is implicit in
the definition of }$P_{\boldsymbol{\theta }}^{n}$, distinct from that of the
observed data, $\mathbf{y}_{n}$, that is used to build the posterior over
the elements of {$\mathcal{P}^{n}.$}}} {Two more comments regarding notation
are useful at this point. First, consistent with our earlier comment, we
abuse notation by defining a prior directly over a predictive, }$P_{%
\boldsymbol{\theta }}^{n}$ {here. In fact, the prior is placed over }$%
\boldsymbol{\theta }$, {and the prior over }$P_{\boldsymbol{\theta }}^{n}$ {%
merely implied. Hence, the incongruous appearance of {conditioning data in
the prior (through the definition of }}$P_{\boldsymbol{\theta }}^{n}$){\ is
of no concern. {It is }simply used to {define} the particular function of }$%
\boldsymbol{\theta }${\ (}$P_{\boldsymbol{\theta }}^{n}$) {that is our
ultimate object of interest, {and that function conditions on (past)
observed data, as it is a predictive {distribution}.} Second, the criterion
function that defines the update in (\ref{post}) is, of course, comprised of
the sequential one-step-ahead predictives, }$P_{\boldsymbol{\theta }}^{t},${%
\ for }$t=0,1,...,n-1.$

{The use of the non-likelihood-based update in (\ref{post}) mimics various
generalizations of the standard Bayesian \textit{inferential }paradigm that }%
have been proposed.{\ Such generalizations replace the likelihood with the
exponential of a problem-specific loss function; the goal being to produce
useful inference in the realistic setting in which the true DGP is unknown,
and the correct likelihood function thus unavailable. This literature has
its roots in the `Gibbs posteriors' of {\cite{Zhang2006a}, \cite{Zhang2006b}}
and {\cite{jiang2008}}, in which the exponential of a general loss function
replaces the likelihood function in the Bayesian update. However, it is
arguably \cite{bissiri2016general}, }{{and the subsequent related work in (%
\textit{inter alia}) }}{{\cite{holmes2017assigning}}}, {{\cite%
{lyddon2019general}}} {and \cite{Syring2019}, that have given the method its
recent prominence} in the statistics and econometrics literature.

{The PAC-Bayes algorithms used in machine learning are also characterized,
in part, by exponential functions of general losses. The focus therein is on
loss defined with respect to }\textit{predictors},\textit{\ }{rather than
parameters}; {hence the particular connection with our approach. Our work
is, however, quite distinct from PAC-Bayes. Most notably, the updating
mechanism in (\ref{post}) is expressed in terms of a class of plausible
conditional predictive \textit{distributions}, rather than point predictors,
and the `loss function' defined explicitly in terms of a proper scoring
rule. We also entertain predictive models that feature in the statistics
and/or econometrics literature, and provide asymptotic validation of the
method in this context.\footnote{{The PAC-Bayes method also encompassess
up-dates based on so-called `tempered', or `power' likelihoods, in which
robustness to model misspecification is sought by raising the likelihood
function associated with an assumed model to a particular power. See {\cite%
{grunwald2017}}, {\cite{holmes2017assigning}} and {\cite{miller2019robust}}
for recent examples in which such modified likelihoods feature. }We refer to 
{{{\cite{guedj2019}}} for a thorough review of PAC-Bayes, including the
methods and terminology used in that setting.}}}

{The update in (\ref{post}) is \textit{coherent}} in the sense that the
posterior that results {from updating the prior using two sets of
observations in one step, is the same as that produced by two sequential
updates. Proof of this property follows that of \cite{bissiri2016general}
and exploits the exponential form of the first term on the right-hand-side
of (\ref{post}), in addition to certain conditions on }$w_{n}$ {to be made
explicit below. Indeed, the appearance of }$w_{n}$ {serves to distinguish (%
\ref{post}) from what would be an extension (to the predictive setting) of
the loss-based inference approach adopted by \cite{CH03}.\footnote{%
We note that a negatively-oriented score can be viewed as a relevant measure
of `loss' in a predictive setting. Moreover, it is also possible to define
the loss associated with predictive inaccuracy using functions that are not
formally defined as scoring rules (see, for example, %
\citealp{pesaran2004decision}). However, we give emphasis to scoring rules
in this paper, making brief note only of the applicability of our method to
more general loss functions in the Discussion.}}

{In the case where }$w_{n}=1$ and\textbf{\ }$S(P_{\boldsymbol{\theta }%
}^{t},y_{t+1})=\ln p\left( y_{t+1}|\mathcal{F}_{t},\boldsymbol{\theta }%
\right) ,${\ with }$p\left( y_{t+1}|\mathcal{F}_{t},\boldsymbol{\theta }%
\right) ${\ denoting the predictive density (or mass) function associated
with the class} $\mathcal{P}^{t}$, {the update in (\ref{post}) }obviously
defaults {to the conventional likelihood-based update of the prior defined
over }$\boldsymbol{\theta }$. {We refer hereafter to this special case as
`exact Bayes', }and acknowledge that, given the presumption of
misspecification, there is no sense in which exact Bayes remains the `gold
standard'. {{This case remains, however, a }critical benchmark in the
numerical work{, }}{{\ in which the degree of misspecification of}%
{\normalsize \textbf{\ }$\mathcal{P}^{t}$ }{will be seen to influence the
relative out-of-sample performance of the conventional {Bayesian update}. }}

{We can summarize the posterior in (\ref{post}) by producing} a
simulation-based estimate of the mean predictive: 
\begin{equation}
\mathbb{E}_{w}\left[ P_{\boldsymbol{\theta }}^{n}|\mathbf{y}_{n}\right]
=\int P_{\boldsymbol{\theta }}^{n}\text{d}\Pi _{w}(P_{\boldsymbol{\theta }%
}^{n}|\mathbf{y}_{n}).  \label{Eq:pred_density_BHW_general}
\end{equation}%
However, it is equally feasible to construct measures that capture the
variability of the posterior, such as quantiles or the posterior variance. {%
Moreover, we can use various} {graphical techniques to visualize} the
variation of the {predictives themselves} and understand the way in which
posterior variation over the class $\mathcal{P}^{t}$ impacts on predictive
accuracy \textit{per se}.

Before moving on, we quickly demonstrate the usefulness of this new approach
to Bayesian prediction, {and the predictive gains that it can reap}, using a
simple toy example.

\subsection{A toy example: ARCH(1)\label{toy_eg}}

{We produce predictive distributions for a financial return generated from }%
a latent stochastic volatility model with {a }skewed {marginal distribution}%
, {with precise details of this true DGP to be given in Section \ref%
{section:simulaexcer}. The predictive class, }{\normalsize $\mathcal{P}^{t}$%
, }{is defined by an }ARCH {model }of order 1 (ARCH(1)) with Gaussian
errors, $y_{t}=\theta _{1}+\sigma _{t}{\epsilon _{t},}$ ${\epsilon _{t}\sim
i.i.d.N\left( 0,1\right) ,}$ ${\sigma _{t}^{2}}{=\theta _{2}+\theta
_{3}\left( y_{t-1}-\theta _{1}\right) ^{2},}$ {with }{$\boldsymbol{\theta }%
=\left( \theta _{1},\theta _{2},\theta _{3}\right) ^{\prime }$, }{and $\pi
\left( \boldsymbol{\theta }\right) \propto \frac{1}{\theta _{2}}\times I%
\left[ \theta _{2}>0,\theta _{3}\in \lbrack {0,1)}\right] ]$ }(with $I$\ the
indicator function){\ }{defining a prior density over }$\boldsymbol{\theta }$%
.

For $p\left( y_{t+1}|\mathcal{F}_{t},\boldsymbol{\theta }\right) $ denoting
the predictive density function associated with {the }Gaussian ARCH(1)
class, we implement FBP using the following two scoring rules:{%
\begin{align}
S_{\text{LS}}(P_{\boldsymbol{\theta }}^{t},y_{t+1})& =\ln p\left( y_{t+1}|%
\mathcal{F}_{t},\boldsymbol{\theta }\right)  \label{ls_prelim} \\
S_{\text{CS}}(P_{\boldsymbol{\theta }}^{t},y_{t+1})& =\ln p\left( y_{t+1}|%
\mathcal{F}_{t},\boldsymbol{\theta }\right) I\left( y_{t+1}\in A\right) +%
\left[ \ln \int_{A^{c}}p\left( y|\mathcal{F}_{t},\boldsymbol{\theta }\right)
dy\right] I\left( y_{t+1}\in A^{c}\right) .  \label{csr_prelim}
\end{align}%
}As noted{\ }already, use {of the log score }({LS}){, (\ref{ls_prelim}), in (%
\ref{post}) (with }$w_{n}=1$) {yields the conventional likelihood-based
Bayesian update, and we label the results based on this score as exact Bayes}
as a consequence. The score in (\ref{csr_prelim}) is the {censored
likelihood score} ({CS}){\ introduced by \cite{diks2011likelihood}, and
applied by }\cite{opschoor2017combining}{{\ to the prediction of financial
returns.\ This score rewards predictive accuracy over the region of interest}
$A$ ({with} $A^{c}${\ indicating the complement of this region). Here we
report results solely for }$A$ {defining} the lower {and upper tail of the
predictive distribution, as determined respectively by the 10\% and 90\%
quantile of the empirical distribution of }$y_{t}${. We label the results
based on the use of (\ref{csr_prelim}) in (\ref{post}) (also using }$w_{n}=1$%
) {as }FBP-{CS}$_{<10\%}$ and FBP-{CS}$_{>90\%}.$}

{Postponing discussion of the full design details until Section \ref%
{section:simulaexcer}, we record in Table }\ref{toy}{\ out-of-sample results
based on repeated computation of (\ref{Eq:pred_density_BHW_general}) using
expanding windows to produce }(via{\ Markov chain Monte Carlo) draws from (%
\ref{post}). Using a total of 2,000 out-of-sample values, the average {LS}
(computed across the 2,000 mean predictives) and the average }{CS for the
lower and upper 10\% tail (denoted by CS}$_{<10\%}${\ and }{CS}$_{>90\%}$ {%
respectively}) {are computed for each of the three different updating
methods. }Recalling that we use positively{-oriented scores, the largest
average score, according to each out-of-sample evaluation method, is
highlighted in bold.}

\begin{table}[tbp]
\centering%
\begin{tabular}{lllll}
\hline\hline
&  & \multicolumn{3}{c}{\textbf{Out-of-sample score}} \\ \cline{3-5}
&  & \multicolumn{1}{c}{LS} & \multicolumn{1}{c}{\ CS$_{<10\%}$} & 
\multicolumn{1}{c}{CS$_{>90\%}$} \\ 
\textbf{Updating method} &  & \multicolumn{1}{c}{} & \multicolumn{1}{c}{} & 
\multicolumn{1}{c}{} \\ \cline{1-1}
&  & \multicolumn{1}{c}{} & \multicolumn{1}{c}{} & \multicolumn{1}{c}{} \\ 
Exact Bayes &  & \multicolumn{1}{c}{\textbf{-1.3605}} & -0.4089 & -0.2745 \\ 
FBP-CS$_{<10\%}$ &  & \multicolumn{1}{c}{-1.4420} & \textbf{-0.3943} & 
-0.3833 \\ 
FBP-CS$_{>90\%}$ &  & \multicolumn{1}{c}{-3.0067} & -1.4157 & \textbf{%
-0.2397 } \\ \hline\hline
\end{tabular}%
\caption{{\protect\footnotesize Predictive accuracy of FBP, using the
ARCH(1) predictive class. The rows in each panel refer to the update method
used. The columns refer to the out-of-sample measure used to compute the
average scores. The figures in bold are the largest average scores according
to a given out-of-sample measure.}}
\label{toy}
\end{table}

{We see that use of the {CS} {rule} in the posterior update yields {better
out-of-sample }performance, \textit{as measured by that score}, in both the
upper and lower tails. In absolute terms, the gain of `focusing' is more
substantial in the upper tail than the lower tail, and in Section \ref%
{section:simulaexcer} we shall see why this is so. {The} average LS produced
by the exact Bayes ({LS-based) }update is {also} larger than the average LS
produced by both }FBP-{CS}$_{<10\%}$ and FBP-{CS}$_{>90\%}$.

{In summary, focusing works, and the following theoretical results {give
some insight into} why.}

\section{Bayesian and Frequentist Agreement\label{theory}}

Whilst the elements of $\mathcal{P}^{t}$ may, in principle, be either
parametric or nonparametric conditional distributions, in the remainder we
focus on the parametric case to {simplify} the analysis, leaving rigorous
analysis of nonparametric conditionals for later study. However, we remind
the reader that this reduction to parametric conditionals {covers both }the
canonical case where the elements of $\mathcal{P}^{t}$ are indexed by a
finite-dimensional parameter, in which case $\Theta $ is a Euclidean space,
as well as the case where the elements in $\mathcal{P}^{t} $ are (a finite
collection of) mixtures of predictives, in which case $\Theta $ denotes
either the weights of the mixture, or the combination of the weights and the
unknown parameters of the constituent predictives.

\subsection{Choosing $w_{n}$ \label{w}}

With reference to the conventional Bayesian approach to inference on the
unknown parameters, $\boldsymbol{\theta }$, which characterize an assumed
DGP, the posterior density, 
\begin{equation}
\pi (\boldsymbol{\theta }|\mathbf{y}_{n})\propto \ell (\mathbf{y}_{n}|%
\boldsymbol{\theta })\pi (\boldsymbol{\theta }),  \label{exact_Bayes}
\end{equation}%
{where }$\ell (\mathbf{y}_{n}|\boldsymbol{\theta })$ {denotes the likelihood
function, }arises via a decomposition of the joint probability distribution
for $\boldsymbol{\theta }$ and the random vector $\mathbf{y}_{n}.$ As such,
the representation of $\pi (\boldsymbol{\theta }|\mathbf{y}_{n})$ as
proportional to the product of a density (or mass) function for $\mathbf{y}%
_{n}$, and the prior for $\boldsymbol{\theta }$, reflects the usual calculus
of probability distributions, and provides a natural `weighting' between the
likelihood and the prior.

Once one moves away from this conventional framework, and replaces the {%
likelihood} with an alternative mechanism through which the data provides
information about $\boldsymbol{\theta }$, {this natural weighting is lost}. {%
Instead, a} subjective choice {must} be made regarding the relative weight
given to prior and data-based information in the production of the
posterior, with the scale factor $w_{n}$ {in (\ref{post}) }denoting this
subjective choice of weighting. \cite{bissiri2016general} propose several
methods for choosing $w_{n}$, including annealing methods,
hyper-parametrization of $w_{n}$, and setting $w_{n}$ to ensure the
equivalence of the expected `loss' of the prior and data-based components of
(\ref{post}). The authors also suggest choosing $w_{n}$ to ensure correct
frequentist coverage of posterior credible intervals, plus the use of priors
that are conjugate to the weighted data-based criterion.\footnote{%
Further proposals on the choice of $w_{n}$ can be found in \cite%
{holmes2017assigning}, \cite{lyddon2019general} and {\cite{Syring2019}.}}

In contrast, our interest is not {in inference on} $\boldsymbol{\theta }$ 
\textit{per se}, but in forecast accuracy. Given this goal, from a
theoretical standpoint, our only concern is that, for $\{w_{n}:n\geq 1\}$ a
chosen scaling sequence, the FBP {posterior measure concentrates} onto the
element of $\mathcal{P}^{n}$ that is most accurate in the chosen scoring
rule, which is defined by the following value in $\Theta ${:} 
\begin{equation}
\boldsymbol{\theta }_{\ast }=\argmax_{\boldsymbol{\theta }\in \Theta
}\lim_{n\rightarrow \infty }\mathbb{E}\left[ {S}_{n}(P_{\boldsymbol{\theta }%
}^{n})/n\right] .  \label{eq:opt}
\end{equation}%
{As the following }result demonstrates, this concentration occurs for any
reasonable choice of $w_{n}$.

\begin{lemma}
\label{lem:one} \label{Theorem1} Assume Assumptions \ref{ass:post}-\ref%
{ass:expand} in Appendix \ref{app:A} are satisfied, {and denote the }%
{\normalsize FBP }{posterior density function }{by $\pi _{w}\left( 
\boldsymbol{\theta }|\mathbf{y}_{n}\right) $.} If the sequence $\{w_{n}\}$
satisfies $\lim_{n}w_{n}=C$, $0<C<\infty $, the posterior density $\pi
_{w}(\cdot |\mathbf{y}_{n})$ converges to $P_{\boldsymbol{\theta }_{\ast
}}:=\lim_{n\rightarrow \infty }P_{\boldsymbol{\theta }_{\ast }}^{n},$\textbf{%
\ }the limit of the predictive defined by $\boldsymbol{\theta }_{\ast }$, at
rate $1/\sqrt{n}$.
\end{lemma}

\begin{remark}
The above result demonstrates that, if we restrict our analysis to a class
of parametric predictives, \textit{FBP asymptotically concentrates,} at rate 
$1/\sqrt{n}$, onto the predictive that is most accurate according to the
scoring rule $P_{\boldsymbol{\theta }}^{t}\mapsto \lim_n\mathbb{E}\left[S(P_{%
\boldsymbol{\theta }}^{t},\cdot )/n\right]$.
\end{remark}

\begin{remark}
The conditions for the above result are discussed in Appendix \ref{app:A}
and are similar to the standard regularity conditions for parametric $M$%
-estimators, along with some uniform control on the tail of the prior $\Pi $%
. These conditions are similar to those used elsewhere in the literature,
e.g., {\cite{CH03}}. Interestingly, Lemma \ref{Theorem1} is valid for a wide
variety of\textbf{\ }$\{w_{n}\}$. In Sections \ref{section:simulaexcer} and %
\ref{empapp}\textit{,} we detail the particular values of $w_{n}$ {that we
use to produce our numerical predictions.}
\end{remark}

\subsection{Merging}

In the previous section, we have seen that, for a reasonable choice of $%
w_{n} $, the FBP posterior concentrates on the element of $\mathcal{P}^{n}$
that is most accurate for prediction under the chosen scoring rule. In this
section, we compare the behavior of predictions obtained from the FBP
posterior with those {that would be }obtained {using direct optimization of
an expected score criterion to produce a frequentist point estimate of} $%
\boldsymbol{\theta }$, {and the associated predictive that conditions on
this point estimate.}

Define the following predictive measures{\ 
\begin{align}
P_{w}^{n}(\cdot )& =\int_{\Theta }P_{\boldsymbol{\theta }}^{n}(\cdot )\text{d%
}\Pi _{w}(P_{\boldsymbol{\theta }}^{n}|\mathbf{y}_{n}),  \label{P_fore} \\
P_{\ast }^{n}(\cdot )& =\int_{\Theta }P_{\boldsymbol{\theta }}^{n}(\cdot )%
\text{d}\delta _{\boldsymbol{\theta }_{\ast }},  \label{P_true}
\end{align}%
where $\delta _{\boldsymbol{\theta }_{\ast }}$ denotes the Dirac measure at
the point $\boldsymbol{\theta }=\boldsymbol{\theta }_{\ast }$, for $%
\boldsymbol{\theta }_{\ast }$ defined in \eqref{eq:opt}. The }mean{\
predictive $P_{w}^{n}(\cdot )$ defines a distribution for the random variable%
} $y_{n+1}$, conditional {on observed} data $\mathbf{y}_{n}$, and where our
uncertainty about the members of the predictive class, $\mathcal{P}^{n}$, is
integrated out using the {posterior} $\Pi _{w}(\cdot |\mathbf{y}_{n})$, for
some choice of tuning sequence $w_{n}$.\footnote{%
The expression in (\ref{P_fore}) is just a more formal representation of (%
\ref{Eq:pred_density_BHW_general}).} In contrast, the predictive $P_{\ast
}^{n}(\cdot )$ {in (\ref{P_true}) }denotes the optimal predictive obtained
by {maximizing the expected score}. {Clearly, obtaining $P_{\ast }^{n}$ is
infeasible in practice. {Instead,}} the following estimated value of {{$%
\boldsymbol{\theta }_{\ast }$, }}$\widehat{\boldsymbol{\theta }}:=\arg
\min_{P_{\boldsymbol{\theta }}^{n}\in \mathcal{P}^{n}}{S}_{n}(P_{\boldsymbol{%
\theta }}^{n})/n,${\ }is generally used in place of {$\boldsymbol{\theta }%
_{\ast }$. }Under the same regularity conditions as in Lemma \ref{Theorem1},
we can derive the asymptotic behavior of $\widehat{\boldsymbol{\theta }}$.

\begin{lemma}
\label{lem:two}Under Assumptions \ref{ass:post}-\ref{ass:expand} in Appendix %
\ref{app:A}, if $\widehat{\boldsymbol{\theta }}$ is consistent for $%
\boldsymbol{\theta }_{\ast }$, and if $S_{n}(P_{\widehat{\boldsymbol{\theta }%
}_{n}}^{n})\geq S_{n}(P_{{\boldsymbol{\theta }}_{\ast }}^{n})+o_{p}(1/\sqrt{n%
})$, then $\sqrt{n}(\widehat{\boldsymbol{\theta }}-{{\boldsymbol{\theta }%
_{\ast }}})\Rightarrow \mathcal{N}(0,W),$ where $W=H^{-1}VH^{-1}$ and 
\begin{equation*}
V:=\lim_{n\rightarrow \infty }\text{Var}\left[ \sqrt{n}\left\{ \frac{%
\partial }{\partial \boldsymbol{\theta }}{S}_{n}(P_{\boldsymbol{\theta }%
_{\ast }}^{n})-\mathbb{E}\left[ \frac{\partial }{\partial \boldsymbol{\theta 
}}{S}_{n}(P_{\boldsymbol{\theta }_{\ast }}^{n})\right] \right\} \right] ;%
\text{ }H:=\plim_{n\rightarrow \infty }\mathbb{E}\left[ \frac{\partial }{%
\partial \boldsymbol{\theta }\partial \boldsymbol{\theta }^{\prime }}{S}%
_{n}(P_{\boldsymbol{\theta }}^{n})|_{\boldsymbol{\theta }={{\boldsymbol{%
\theta }_{\ast }}}}\right] .
\end{equation*}
\end{lemma}

Using the estimator $\widehat{\boldsymbol{\theta }}$, we can define the
following frequentist equivalent to the FBP predictive:%
\begin{equation}
P_{\widehat{\boldsymbol{\theta }}}^{n}(\cdot )=\int_{\Theta }P_{\boldsymbol{%
\ \theta }}^{n}(\cdot )\text{d}\mathcal{N}(\widehat{\boldsymbol{\theta }}%
,W/n),  \label{eq:freq}
\end{equation}%
where $\mathcal{N}(\widehat{\boldsymbol{\theta }},W/n)$ denotes the normal
distribution function with mean $\widehat{\boldsymbol{\theta }}$ and
variance-covariance matrix $W/n$. Using Lemmas \ref{Theorem1} and \ref%
{lem:two}, we can deduce the following relationship between the frequentist
predictive in equation \eqref{eq:freq} and the FBP predictive in (\ref%
{P_fore}).{\setcounter{theorem}{0}}

\begin{theorem}
{\ \label{thm2} Under Assumptions \ref{ass:post}-\ref{ass:expand} in
Appendix \ref{app:A}, for $\lim_{n}w_{n}=C>0$, the predictive distributions $%
P_{w}^{n}(\cdot )$ }and {$P{_{\widehat{\boldsymbol{\theta }}}^{n}}(\cdot )$
satisfy: 
\begin{equation*}
\sup_{B\in \mathcal{F}}|P_{w}^{n}(B)-P_{\widehat{\boldsymbol{\theta }}%
}^{n}(B)|=o_{p}(1).
\end{equation*}%
}
\end{theorem}

\begin{remark}
Theorem \ref{thm2} states that, for any sequence $\lim_{n}w_{n}=C$, the FBP
predictive $P_{w}^{n}(\cdot )$\textbf{\ }and the {(feasible) }optimal
frequentist predictive $P_{\widehat{\boldsymbol{\theta }}}^{n}(\cdot )$ will
agree asymptotically. The above result is colloquially referred to as
`merging' (\citealp{blackwell1962}). This result states that, in terms of
the total variation distance, the predictions obtained by FBP and those
obtained by a frequentist making predictions according to an optimal score
estimator $\widehat{\boldsymbol{\theta }}$ will asymptotically agree.
\end{remark}

\section{Simulation Study: Financial Returns{\label{section:simulaexcer}}}

\subsection{Overview of the simulation design\label{overview}}

We first\ illustrate our approach with a simulation exercise that nests the
toy example in Section \ref{toy_eg}. Adopting a simulation approach allows
us to choose predictive classes that misspecify the (known) DGP to varying
degrees, and to thereby measure the relative performance of FBP in different
misspecification settings. We use both numerical summaries and animated
graphics to illustrate the predictive accuracy of FBP, using a range of
scores to define the {update}. With a slight abuse of terminology, in what
follows we refer to FBP-LS solely as `exact Bayes', reserving the
abbreviation FBP for all other instances of the focused method.

We address three questions. \textit{First}, what sample size is required in
practice for the asymptotic results to be on display? That is, how large
does $n$ have to be for FBP based on a particular scoring rule to provide
the best out-of-sample performance according to that same rule? \textit{%
Second}, does the degree of misspecification affect the dominance of FBP
over exact Bayes? \textit{Third}, does misspecification have a differential
impact on FBP implemented via different scoring rules?

With the aim of replicating the stylized features of financial returns data,
we generate a logarithmic return, $y_{t}$, from

{\ 
\begin{align}
h_{t}& =\bar{h}+a(h_{t-1}-\bar{h})+\sigma _{h}\eta _{t}  \label{h} \\
z_{t}& =e^{0.5h_{t}}\varepsilon _{t}  \label{z} \\
y_{t}& =D^{-1}\left( F_{z}\left( z_{t}\right) \right) ,  \label{skew}
\end{align}%
where $\eta _{t}\sim i.i.d.N\left( 0,1\right) $ and $\varepsilon _{t}\sim
i.i.d.N\left( 0,1\right) $ are independent processes, $\{z_{t}\}_{t=1}^{n}$
is a }latent process with stochastic (logarithmic) variance, $h_{t}$, and $%
F_{z}$ is the implied marginal distribution of $z_{t}$ (evaluated via
simulation). The `observed' return, $y_{t}$, is then generated as in (\ref%
{skew}), via the (inverse) distribution function associated with a
standardized skewed-normal distribution, $D$.\footnote{%
For more details of the specific skewed-normal specification that we adopt
see \cite{azzalini1985class}.} This process of inversion imposes on $%
\{y_{t}\}_{t=1}^{n}$ the dynamics of the stochastic volatility model
represented by Equations (\ref{h}) and (\ref{z}) (via $F_{z}$) in addition
to the negative skewness that is characteristic of the empirical
distribution of a financial return.\footnote{%
See \cite{smith2018inversion} for discussion of this type of implied copula\
model.}

We adopt three alternative parametric predictive classes, $\mathcal{P}^{t}:$%
\ i) Gaussian ARCH(1) (reproduced here for convenience and numbered for
future reference): \textbf{\ }{\ 
\begin{equation}
y_{t}=\theta _{1}+\sigma _{t}{\epsilon _{t};}\text{ }{\epsilon _{t}\sim
i.i.d.N\left( 0,1\right) ;}\text{ }{\sigma _{t}^{2}}{=\theta _{2}+\theta
_{3}\left( y_{t-1}-\theta _{1}\right) ^{2};}  \label{y1}
\end{equation}%
}ii) Gaussian GARCH(1,1):{\ 
\begin{equation}
y_{t}=\theta _{1}+\sigma _{t}{\epsilon _{t};}\text{ }{\epsilon _{t}\sim
i.i.d.N\left( 0,1\right) ;}\text{ }{\sigma _{t}^{2}}{=\theta _{2}+\theta
_{3}\left( y_{t-1}-\theta _{1}\right) ^{2}}+\theta _{4}\sigma _{t-1}^{2};
\label{y2}
\end{equation}%
and} iii) a mixture of the predictives of two models: an ARCH(1) model with
a skewed-normal innovation, and a GARCH(1,1) model with a Gaussian
innovation. We represent the elements of class iii) as:{\ 
\begin{equation}
p\left( y_{t+1}|\mathcal{F}_{t},\theta _{1}\right) =\theta _{1}p_{1}\left(
y_{t+1}|\mathcal{F}_{t},{\boldsymbol{\psi }_{1}}\right) +\left( 1-\theta
_{1}\right) p_{2}\left( y_{t+1}|\mathcal{F}_{t},{\boldsymbol{\psi }_{2}}%
\right) .  \label{y3}
\end{equation}%
}In (\ref{y1}) and (\ref{y2}) the respective parameter vectors,\textbf{\ }{$%
\boldsymbol{\theta }=\left( \theta _{1},\theta _{2},\theta _{3}\right)
^{\prime }$ }and {$\boldsymbol{\theta }=\left( \theta _{1},\theta
_{2},\theta _{3},\theta _{4}\right) ^{\prime }$}, characterize the specific
predictive model, with the GARCH(1,1) model being a more flexible (and, in
this sense, `less misspecified') representation of the true DGP than is the
ARCH(1) model. In (\ref{y3}) the {parameter vectors $\boldsymbol{\psi }_{1}$
and $\boldsymbol{\psi }_{2}$ }that characterize the constituent predictives
in the mixture are taken as known, and the scalar weight parameter{\textbf{\ 
}}$\theta _{1}$ is the only unknown. The component $p_{2}\left( y_{t+1}|%
\mathcal{F}_{t},{\boldsymbol{\psi }_{2}}\right) $ is specified according to
the model in (\ref{y2}), whilst\textbf{\ }$p_{1}\left( y_{t+1}|\mathcal{F}%
_{t},{\boldsymbol{\psi }_{1}}\right) $ represents the predictive associated
with the model in (\ref{y1}), but with $\epsilon _{t}$ distributed as a
standardized\textbf{\ }skewed-normal variable with asymmetry parameter $%
\gamma $.\footnote{{The values imposed for }${\boldsymbol{\psi }}_{1}$ {and }%
${\boldsymbol{\psi }_{2}}$ {are the maximum likelihood estimators of these
parameters.}} With the proviso made that the parameters of the constituent
models are taken as given, the {linear pool} is arguably the most flexible
form of predictive considered here, and defines the least misspecified
predictive class in this sense. The prior over each of the three predictive
classes is determined by the prior over the relevant parameter (vector) $%
\boldsymbol{\theta }$, respectively:\textbf{\ }i)\textbf{\ }{$\pi \left( 
\boldsymbol{\theta }\right) \propto \frac{1}{\theta _{2}}\times I\left[
\theta _{2}>0,\theta _{3}\in \lbrack {0,1)}\right] ]$ }(as defined earlier){%
, }ii){\ $\pi \left( \boldsymbol{\theta }\right) \propto \frac{1}{\theta _{2}%
}\times I\left[ \theta _{2}>0,\theta _{3}\in \lbrack 0,1),\theta _{4}\in
\lbrack {0,1)}\right] ]\times I({\theta _{3}+}\theta _{4}<1)$}, and iii)%
\textbf{\ }{$\pi $}$\left( \theta _{1}\right) \propto u$ (with $u$ uniform
on $(0,1)$).

We now implement FBP using the two scoring rules in (\ref{ls_prelim}) and (%
\ref{csr_prelim}), plus the continuously ranked probability score (CRPS),{%
\begin{equation}
S_{\text{CRPS}}(P_{\boldsymbol{\theta }}^{t},y_{t+1})=-\int_{-\infty
}^{\infty }\left[ P\left( y|\mathcal{F}_{t},\boldsymbol{\theta }\right)
-I(y\geq y_{t+1})\right] ^{2}dy,  \label{Eq:CRPSloss}
\end{equation}%
}where $P\left( y|\mathcal{F}_{t},\boldsymbol{\theta }\right) $ denotes the
predictive cumulative distribution function (cdf) associated with $p\left( y|%
\mathcal{F}_{t},\boldsymbol{\theta }\right) $. Proposed by \cite%
{gneiting2007strictly}, CRPS is sensitive to distance,\ rewarding the
assignment of high predictive mass near to the realized value of $y_{t+1}$.
It can be evaluated in closed form for the (conditionally) Gaussian
predictive classes i) and ii), using the third equation {provided in %
\citet[p.~367]{gneiting2007strictly}}.\ For predictive class iii),
evaluation is performed numerically using expression (17) in %
\citet[p.~367]{gneiting2011comparing}. In the case of the CS in (\ref%
{csr_prelim}), all components, including the integral $\int_{A^{c}}p\left( y|%
\mathcal{F}_{t},\boldsymbol{\ \theta }\right) dy$, have closed-form
representations for predictive classes i) and ii). For the third predictive
class, CS is computed as 
\begin{equation*}
S_{\text{CS}}(P_{\boldsymbol{\theta }}^{t},y_{t+1})=S_{\text{CS}}\left( P_{%
\boldsymbol{\psi _{1}}}^{t},y_{t+1}\right) +\log \left\{ \theta _{1}+\left(
1-\theta _{1}\right) \exp \left[ S_{\text{CS}}\left( P_{\boldsymbol{\psi _{2}%
}}^{t},y_{t+1}\right) -S_{\text{CS}}\left( P_{\boldsymbol{\psi _{1}}%
}^{t},y_{t+1}\right) \right] \right\} ,
\end{equation*}%
where\textbf{\ }$S_{\text{CS}}\left( P_{\boldsymbol{\psi _{1}}%
}^{t},y_{t+1}\right) $ and $S_{\text{CS}}\left( P_{\boldsymbol{\psi _{2}}%
}^{t},y_{t+1}\right) $ correspond to the censored scores for the two
constituent models, both having closed-form solutions.

As noted in Section \ref{toy_eg}, when either (\ref{ls_prelim}) or (\ref%
{csr_prelim}) is used in (\ref{post}) a scale of $w_{n}=1$ is adopted. This
is a natural choice, given that use of (\ref{ls_prelim}) defines the
(misspecified) likelihood function induced by the predictive class, and that
use of (\ref{csr_prelim}) is comparable to the specification of the
likelihood function for a censored random variable (%
\citealp{diks2011likelihood}). When (\ref{Eq:CRPSloss}) is used to define
the posterior {update} however, the interpretation of $\exp \left[
w_{n}S_{n}\left( P_{\boldsymbol{\theta }}^{n}\right) \right] $ as the
(unnormalized) probability distribution of a random variable is {lost, and} $%
w_{n}$ must be chosen with reference to some criterion for weighting $\exp %
\left[ w_{n}S_{n}\left( P_{\boldsymbol{\theta }}^{n}\right) \right] $ and $%
\Pi \left( P_{\boldsymbol{\theta }}^{n}\right) $. We choose to target a
value for $w_{n}$ that ensures a rate of posterior {update} - when using
CRPS - that is \textit{similar} to that of the {update} based on LS, by
defining 
\begin{equation}
w_{n}=\frac{\mathbb{E}_{\pi \left( \boldsymbol{\theta }|\mathbf{y}%
_{n}\right) }\left[ \sum_{t=0}^{n-1}S_{\text{LS}}(P_{\boldsymbol{\theta }%
}^{t},y_{t+1})\right] }{\mathbb{E}_{\pi \left( \boldsymbol{\theta }|\mathbf{y%
}_{n}\right) }\left[ \sum_{t=0}^{n-1}S_{\text{CRPS}}(P_{\boldsymbol{\theta }%
}^{t},y_{t+1})\right] }.  \label{Eq:set_w_specific_doc}
\end{equation}%
The subscript $\pi \left( \boldsymbol{\theta }|\mathbf{y}_{n}\right) $
indicates that the expectation is with respect to the exact posterior
distribution for $\boldsymbol{\theta }$. In practice,\textbf{\ }$w_{n}$\ is
estimated as%
\begin{equation}
\widehat{w}_{n}=\frac{\sum_{j=1}^{J}\left[ \sum_{t=0}^{n-1}S_{\text{LS}}(P_{%
\boldsymbol{\theta }^{(j)}}^{t},y_{t+1})\right] }{\sum_{j=1}^{J}\left[
\sum_{t=0}^{n-1}S_{\text{CRPS}}(P_{\boldsymbol{\theta }^{(j)}}^{t},y_{t+1})%
\right] },  \label{w_hat}
\end{equation}%
using $J$ draws of $\boldsymbol{\theta }$ from the $\pi \left( \boldsymbol{%
\theta }|\mathbf{y}_{n}\right) $, $\boldsymbol{\theta }^{(j)}$, $%
j=1,2,...,J. $ The link between specifying $w_{n}$ as in (\ref%
{Eq:set_w_specific_doc}) and achieving a rate of posterior {update} that
approximates that of exact Bayes, is detailed in Appendix~\ref{Appen:w_crps}%
. All details of the Markov chain Monte Carlo (MCMC) scheme used to perform
the posterior sampling are provided in Appendix \ref{comp}.

\subsection{A comment on the role of predictive combinations\label{comb}}

{Before we proceed to document and discuss the simulation results in the
following section, we comment briefly on the role played by predictive
combinations in the simulation exercise and, subsequently, in certain of our
empirical illustrations. }

{As noted} in the Introduction,{\ there is a well-established literature -
invoking both frequentist and Bayesian principles - in which the weights in
weighted combinations of predictives are either optimized (in the first
case) or up-dated (in the second) according to predictive criteria. The
frequentist literature includes work on linear combinations (or linear
pools), in which various measures of predictive accuracy, including scoring
rules, are used to define the criterion that is optimized to estimate the
weights. Relevant references here include {\cite{Hall2007}}, {\cite%
{Kascha2010}}, {\cite{Geweke2011}}, {\cite{Ganics2017}}, {\cite%
{opschoor2017combining}} and {Pauwels \textit{et al.} (2019)}.\footnote{{%
Closely related work appears in {\cite{Jore2010}.}}} Non-linear weighting
schemes (or non-linear transformations of linear schemes) - again estimated
via optimization of prediction-based criteria - are explored in {\cite%
{ranjan2010}}, {\cite{Clements2011}}, {\cite{gneiting2013}} and {\cite%
{Kap2015}}. The Bayesian} literature,{\ having access as it does to
posterior simulation schemes, has entertained more sophisticated (including
time-varying) weighting schemes, in which predictive performance influences
the posterior up-dates of the weights in one way or another. Key work here
includes }{\cite{Billio2013}}, {{\cite{casarin2015jss}, }\cite{casarin2015}, 
{\cite{casarin2016}}, {\cite{Pett2016}}, {\cite{Bassetti2018}, {\cite%
{BASTURK2019}}} and {\cite{casarin2020}}.\footnote{%
Conventional Bayesian model averaging (BMA) as applied to predictive
distributions is effectively driven by a log score criterion, given the
intrinsic connection between the predictive likelihood and the marginal
likelihood that underlies each BMA weight. (See, for example, %
\citealp{geweke2005}, Chapter 2.)}}

\textit{In principle}{, }any of the above combination schemes could be used
to construct the predictive class{\ }$\mathcal{P}^{t}${, with the }chosen{\
measure of predictive accuracy used to define the update in (\ref{post}).} {%
It is certainly the case that }more sophisticated combination schemes for
the constituent weights would require an alternative, and more
computationally intensive, posterior sampling scheme{\ than the
straightforward MCMC algorithm we have adopted for the simple linear pool; {%
nevertheless}, beyond {this} {issue}, the }principles{\ that underpin our
methodology would remain the same.\footnote{%
We conjecture that the melding of the nonparametric Bayesian approach in 
\cite{Bassetti2018} and our generalized Bayesian updating would be a
particularly fruitful{\ avenue for future exploration; as would be a merging
of the Bayesian predictive synthesis of Mike West and co-authors (e.g. %
\citealp{Johnson2018}) and our focused prediction method.}}}

{However, the key point is that} {the goal of yielding a more accurate
representation of the true DGP by }employing more sophisticated weighting
schemes is ancillary to the predictive philosophy underpinning our approach:
we are not concerned with trying to find a model that more accurately
represents the true DGP, but {with} ensuring that our chosen predictive
delivers accuracy in terms of our chosen loss measure. {Hence, we include
predictive combinations in certain numerical illustrations \textit{not} for
the purpose of building a better representation of the DGP \textit{per se},
but: first,} to highlight the fact that such forms of distribution \textit{%
can} be accommodated within our procedure; {and second}, as a way of
illustrating the effect on {the \textit{relative} performance of FBP and
exact Bayes }of using a predictive class that more accurately captures the
features of the true DGP than does any single model.

\subsection{Simulation results\label{sim}}

\subsubsection{Summary results based on mean predictives\label{sim_results}}

We generate{\ $2,500${\ observations} }of $y_{t}$ from the DGP in (\ref{h})-(%
\ref{skew}), using parameter values:{\ $a=0.9$, $\bar{h}=-0.4581$ }and{\ $%
\sigma _{h}=0.4173$, while }$D${\ }defines the standardized skewed-normal
distribution with shape parameter $\gamma =-5$, which produces an
empirically plausible degree of negative skewness. For each predictive
class, and for each score {update}, the exercise begins by using the
relevant computational scheme, as described in Appendix \ref{comp}, to
produce (after thinning) $M=4,000$ posterior draws of{\textbf{\ }}$%
\boldsymbol{\theta },$ $\boldsymbol{\theta }^{(j)}$, $j=1,2,..,M$, and,
hence, $M$ posterior draws of\textbf{\ }$p\left( \cdot |\mathcal{F}_{n},%
\boldsymbol{\theta }\right) $ (at any point in the support of\textbf{\ }$%
y_{n+1})$, which we denote simply by\textbf{\ }$p^{(j)}$, as indexed by the $%
j^{th}$ draw of $\boldsymbol{\theta }.${\ }This first set of posterior draws
is produced using the first{\textbf{\ }$n=500$ }values of {$y_{t}$ }in the
update in (\ref{post}). For each predictive class, six score updates are
employed, corresponding to (\ref{ls_prelim}) and (\ref{Eq:CRPSloss}), plus (%
\ref{csr_prelim}) with the region A defining (approximately) four tails of
the predictive distribution: lower 10\%, lower 20\%, upper 10\% and upper
20\%.\footnote{%
As noted in Section 2.3, the set $A,$ for any required tail probability, is
determined via reference to the \textit{empirical} distribution of $y_{t}$;
hence the use of the word `approximately'.} Hence, for each predictive
class, draws from six different{\textbf{\ }}$\Pi _{w}(P_{\boldsymbol{\theta }%
}^{n}|\mathbf{y}_{n})$ are produced.

Referencing the draws, $p^{(j)}$,\textbf{\ }$j=1,2,..,M$, from any one of
the six distinct posteriors,\textbf{\ }we first estimate the mean predictive
in (\ref{Eq:pred_density_BHW_general}) as: $\widehat{\mathbb{E}}_{w}\left[
P_{\boldsymbol{\theta }}^{n}|\mathbf{y}_{n}\right] =1/M\sum%
\nolimits_{j=1}^{M}p^{(j)},$ and compute the out-of-sample score of this
single predictive, based on the observed value of $y_{n+1}$, for period $%
n+1=501.$ The same six scores used in the in-sample updates are used to
produce these out-of-sample scores. The sample is then extended to $n=501$,
and the same exercise is repeated, with the out-of-sample scores computed
using the observed value of $y_{n+1}$, for time period $n+1=502.$ This
exercise is repeated 2,000 times, with the final set of out-of-sample scores
computed using the observed value of $y_{n+1}$ for time period $n+1=2,500.$
The average of the 2,000 scores is recorded in Table \ref{Table 1}:\ for
each update method, each out-of-sample evaluation method, and each
predictive class.

Expanding on the results reported in Section \ref{toy_eg}, Panels A, B and C
in Table \ref{Table 1} correspond respectively to the three predictive
classes: ARCH(1), GARCH(1,1) and the mixture. The rows in each panel refer
to the six distinct update methods, denoted in turn by: exact Bayes ( $%
\equiv $\ FBP-LS), FBP-CRPS, FBP-CS$_{<10\%}$, FBP-CS$_{<20\%}$, FBP-CS$%
_{>80\%}$ and FBP-CS$_{>90\%}$. The columns refer to the out-of-sample
measure used to compute the average scores: LS, CRPS, CS$_{<10\%}$, CS$%
_{<20\%}$, CS$_{>80\%}$ and CS$_{>90\%}.$ Numerical validation of the
asymptotic results occurs if the largest average scores (bolded) appear in
the diagonal positions in the table; that is, if using FBP with a particular 
\textit{focus} yields the best out-of-sample performance according to that
same measure of predictive accuracy.

\begin{table}[tbp]
\begin{center}
{\normalsize \ 
\scalebox{0.82}{
				\begin{tabular}{lcccccccc}
					\hline\hline
					                          &                  &                  &  &                  &                  &  &                  &                  \\
					                          &                            \multicolumn{8}{c}{\textbf{Panel A: ARCH(1) predictive class}}                             \\
					                          &                  &                  &  &                  &                  &  &                  &                  \\
					                          &                                   \multicolumn{8}{c}{\textbf{Out-of-sample score}}                                    \\ \cline{2-9}
					                          &                  &                  &  &                  &                  &  &                  &                  \\
					                          & \multicolumn{2}{c}{Center Focused}  &  &  \multicolumn{2}{c}{Left Focused}   &  &  \multicolumn{2}{c}{Right Focused}  \\
					                          &        LS        &       CRPS       &  &   CS$_{<10\%}$   &   CS$_{<20\%}$   &  &   CS$_{>80\%}$   &   CS$_{>90\%}$   \\
					\textbf{Updating method} &                  &                  &  &                  &                  &  &                  &                  \\ \cline{1-1}
					                          &                  &                  &  &                  &                  &  &                  &                  \\
					Exact Bayes               & \textbf{-1.3605} &     -0.5299      &  &     -0.4089      &     -0.6687      &  &     -0.4716      &     -0.2745      \\
					FBP-CRPS                  &     -1.3663      & \textbf{-0.5290} &  &     -0.4206      &     -0.6774      &  &     -0.4723      &     -0.2775      \\
					FBP-CS$_{<10\%}$          &     -1.4442      &     -0.5558      &  &     -0.3943      &     -0.6506      &  &     -0.5755      &     -0.3833      \\
					FBP-CS$_{<20\%}$          &     -1.4660      &     -0.5652      &  & \textbf{-0.3933} & \textbf{-0.6484} &  &     -0.6002      &     -0.4072      \\
					FBP-CS$_{>80\%}$          &     -2.0422      &     -0.5902      &  &     -0.9655      &     -1.3470      &  & \textbf{-0.4365} &     -0.2430      \\
					FBP-CS$_{>90\%}$          &     -3.0067      &     -0.6747      &  &     -1.4157      &     -2.0858      &  &     -0.4592      & \textbf{-0.2397} \\ \hline\hline
					                          &                  &                  &  &                  &                  &  &                  &                  \\
					                          &                           \multicolumn{8}{c}{\textbf{Panel B: GARCH(1,1) predictive class}}                           \\
					                          &                  &                  &  &                  &                  &  &                  &                  \\
					                          &                                   \multicolumn{8}{c}{\textbf{Out-of-sample score}}                                    \\ \cline{2-9}
					                          &                  &                  &  &                  &                  &  &                  &                  \\
					                          & \multicolumn{2}{c}{Center Focused}  &  &  \multicolumn{2}{c}{Left Focused}   &  &  \multicolumn{2}{c}{Right Focused}  \\
					                          &        LS        &       CRPS       &  &   CS$_{<10\%}$   &   CS$_{<20\%}$   &  &   CS$_{>80\%}$   &   CS$_{>90\%}$   \\
					\textbf{Updating method} &                  &                  &  &                  &                  &  &                  &                  \\ \cline{1-1}
					                          &                  &                  &  &                  &                  &  &                  &                  \\
					Exact Bayes               & \textbf{-1.3355} &     -0.5259      &  &     -0.3941      &      -0.6500       &  &      -0.4710      &     -0.2747      \\
					FBP-CRPS                  &     -1.3381      & \textbf{-0.5258} &  &     -0.3992      &     -0.6532      &  &     -0.4734      &     -0.2788      \\
					FBP-CS$_{<10\%}$          &     -1.3801      &     -0.5341      &  & \textbf{-0.3838} &     -0.6387      &  &     -0.5282      &     -0.3317      \\
					FBP-CS$_{<20\%}$          &     -1.4126      &      -0.5480      &  &     -0.3840      & \textbf{-0.6375} &  &      -0.5650      &      -0.3710      \\
					FBP-CS$_{>80\%}$          &     -2.0535      &     -0.5918      &  &     -0.9612      &      -1.3530      &  & \textbf{-0.4318} &     -0.2387      \\
					FBP-CS$_{>90\%}$          &     -3.1207      &     -0.6818      &  &     -1.4544      &     -2.1502      &  &     -0.4572      & \textbf{-0.2347} \\ \hline\hline
					                          &                  &                  &  &                  &                  &  &                  &                  \\
					                          &                            \multicolumn{8}{c}{\textbf{Panel C: Mixture predictive class}}                             \\
					                          &                  &                  &  &                  &                  &  &                  &                  \\
					                          &                                   \multicolumn{8}{c}{\textbf{Out-of-sample score}}                                    \\ \cline{2-9}
					                          &                  &                  &  &                  &                  &  &                  &                  \\
					                          & \multicolumn{2}{c}{Center Focused}  &  &  \multicolumn{2}{c}{Left Focused}   &  &  \multicolumn{2}{c}{Right Focused}  \\
					                          &        LS        &       CRPS       &  &   CS$_{<10\%}$   &   CS$_{<20\%}$   &  &   CS$_{>80\%}$   &   CS$_{>90\%}$   \\
					\textbf{Updating method} &                  &                  &  &                  &                  &  &                  &                  \\ \cline{1-1}
					                          &                  &                  &  &                  &                  &  &                  &                  \\
					Exact Bayes               & \textbf{-1.2901} &     -0.5241      &  &     -0.3898      &     -0.6448      &  &     -0.4363      &     -0.2447      \\
					FBP-CRPS                  &     -1.2975      & \textbf{-0.5234} &  & \textbf{-0.3868} & \textbf{-0.6418} &  &     -0.4476      &     -0.2557      \\
					FBP-CS$_{<10\%}$          &     -1.3048      &     -0.5236      &  &     -0.3871      &     -0.6422      &  &     -0.4536      &     -0.2610      \\
					FBP-CS$_{<20\%}$          &     -1.3029      &     -0.5235      &  &     -0.3871      &     -0.6421      &  &     -0.4523      &     -0.2599      \\
					FBP-CS$_{>80\%}$          &     -1.2902      &     -0.5250      &  &     -0.3921      &     -0.6472      &  &     -0.4325      &     -0.2407      \\
					FBP-CS$_{>90\%}$          &     -1.2902      &     -0.5250      &  &     -0.3922      &     -0.6472      &  & \textbf{-0.4324} & \textbf{-0.2406} \\ \hline\hline
				\end{tabular}
		} }
\end{center}
\caption{{\protect\footnotesize Predictive accuracy based on the six
different mean predictives. Panels A to C report the average out-of-sample
scores for the ARCH(1), GARCH(1,1) and mixture predictive class,
respectively. The rows in each panel refer to the update method used. The
columns refer to the out-of-sample measure used to compute the average
scores. The figures in bold are the largest average scores according to a
given out-of-sample measure.}}
\label{Table 1}
\end{table}

The results in Table \ref{Table 1} broadly validate the asymptotic theory.
With minor deviations, the expected appearance of bold figures on the main
diagonal of each panel is in evidence - most notably in Panels A and B.
Hence, with reference to the \textit{first }question outlined at the
beginning of Section \ref{overview}: {an} {initial }sample size exceeding $%
n=500$, {expanded to }$n=2,499$ {in the production of 2,000 one-step-ahead
predictions,} is sufficient for the use of in-sample focusing to reap
benefits out-of-sample.\footnote{%
We reiterate that in this numerical assessment of predictive performance
based on expanding estimation windows, there are two sample sizes that play
a role: i) the size of the estimation period on which the posterior (over
predictives) is based, and from which the mean (one-step-ahead) predictive
and numerical score are extracted; and ii) the size of the out-of-sample
period over which the \textit{average} (one-step-ahead) scores are computed.
With expanding estimation windows, an increase in the out-of-sample period
goes hand-in-hand with a continued increase in the estimation period, i.e.
an increase in $n$.} When there \textit{is} a deviation from the strict
diagonal pattern, such as in the CS$_{<10\%}$ column of Panel A and in the CS%
$_{<10\%}$, CS$_{<20\%}$ and CS$_{>80\%}$ columns of Panel C, the difference
between the relevant (non-diagonal) bold value and the value on the diagonal
is negligible.

With reference to the \textit{second }question, the out-of-sample dominance
of FBP over exact Bayes declines as the predictive class becomes less
misspecified. In particular, the results in Panel C - for the mixture
predictive class - reveal that the average scores computed using a given
out-of-sample measure are very similar for all six updating methods. The
extent of the misspecification of the true DGP {clearly }does matter.%
\footnote{%
For a large enough sample of course, for a predictive class that \textit{%
contains} the true DGP, any updating method based on a proper score should
(under regularity) recover the true predictive mechanism and, hence, should
yield predictive performance out-of-sample (however measured) that matches
that of an update based on an alternative proper score. See \cite%
{gneiting2007strictly} for an early exposition of this sort of point, in the
context of frequentist point estimation using scoring rules.}

With respect to the \textit{third }question: there are two notable results
regarding the \textit{differential} impact of the degree of misspecification
on the different versions of FBP. First, when the degree of misspecification
is most severe (as with the ARCH(1) and GARCH(1,1) predictive classes) use
of the {update} that focuses on the upper tail (FBP-CS$_{>80\%}$ or FBP-CS$%
_{>90\%}$) produces poor out-of-sample accuracy according to the LS and
lower tail measures (CS$_{<10\%}$ and CS$_{<20\%}$). We provide some
graphical insight into this specific phenomenon in Section \ref{animation};
however, the point is that focusing \textit{incorrectly }can\textit{\ }hurt,
in particular when the predictive class is a poor match for the true DGP.
Once misspecification of the predictive class is reduced, the performance of
both FBP-CS$_{>80\%}$ and FBP-CS$_{>90\%}$ - according to \textit{all}
out-of-sample measures - broadly matches that of the other updating methods,
as can be seen in Panel C.

The second differential impact of misspecification pertains to the exact
(misspecified) Bayesian update, relative to all four tail-focused methods
(FBP-CS$_{<10\%}$, FBP-CS$_{<20\%}$, FBP-CS$_{>80\%}$ and FBP-CS$_{>90\%}$).
For example, the values of CS$_{>90\%}$\ for FBP-CS$_{>90\%}$ (the three
bolded figures in the very last column of Table \ref{Table 1}) change very
little over the three panels, as the degree of misspecification lessens. A
similar comment applies to the three values of CS$_{<10\%}$\ for FBP-CS$%
_{<10\%}.$ In contrast, the improvement in performance in the tails (so the
values of CS$_{>90\%}$ and CS$_{<10\%}$) for exact Bayes, as one moves from
the most to the least misspecified predictive class, is more marked; which
makes sense. The focused methods do not aim to get the model correct;
instead, they are deliberately tailored to a particular predictive task
(accurate prediction of extreme values in this case). Hence,
misspecification of the model \textit{per se }matters less. The predictive
performance of exact Bayes, on the other hand, depends entirely on the match
between the model that underpins the method and the truth; there is nothing
else that exact Bayes brings to the table; if the model is wrong, prediction
(however measured) will be adversely affected by that error.

{To gauge the sensitivity of the findings to the size of the out-of-sample
evaluation period, we plot the average one-step-ahead score as a function of
the latter. In the quest for brevity, we perform this task for two
out-of-sample measures only: }\ CS{$_{<10\%}$ }and CS{$_{>90\%}.$ }In Panel
A of Figure \ref{Fig:relativeaccuracy} the cumulative average of CS{$%
_{<10\%} $ (}for 400 to 2,000 out-of-sample periods) is plotted for two
forms of in-sample updates only: exact Bayes (the dashed green line) and {%
FBP-}CS{$_{<10\%}$, (the blue full line). }Each of the three figures (A.1,
A.2, A.3) corresponds respectively to results for each of the three
predictive classes (ARCH(1), GARCH(1,1) and the mixture). In Panel B (B.1,
B.2 and B.3) the corresponding results for the cumulative average of CS{$%
_{>90\%}$ }are presented, based on exact Bayes (the dashed green line) and
FBP{-}CS{$_{>90\%}$, (the blue full line). }In all figures, the \textit{final%
} numerical values plotted correspond to the relevant values reported in
Table \ref{Table 1}.

Beginning with Figure A.1, we see that a sufficiently large of out-of-sample
evaluation period (exceeding approximately 600) is needed for the dominant
performance of FBP over exact Bayes to be in evidence visually; with
in-sample estimation periods exceeding $n=1,100$ contributing to these
average score results. However, beyond this point, the amount by which the
full line exceeds the dashed one stabilizes, reflecting the extent to which {%
FBP-}CS{$_{<10\%}$ }produces more accurate predictions of extreme (lower
tail) returns than does exact Bayes, asymptotically. Tallying with the
interpretation of the numerical results in Table \ref{Table 1}, the extent
to which {FBP-}CS{$_{<10\%}$ }is superior to exact Bayes is successively
less in Figures A.2 and A.3, with the dashed line `moving up' to match the
full line, as the misspecification of the predictive class is reduced. The
size of the evaluation period required to produce a visual distinction
between the out-of-sample performance of exact Bayes and {FBP-}CS{$_{<10\%}$ 
}is larger, the less misspecified is the class.

\begin{figure}[h]
\centering%
\begin{tabular}{cc}
\textbf{Panel A:} & \textbf{Panel B:} \\ 
\textbf{Lower tail accuracy} & \textbf{Upper tail accuracy} \\ 
&  \\ 
\includegraphics[scale= 0.7]{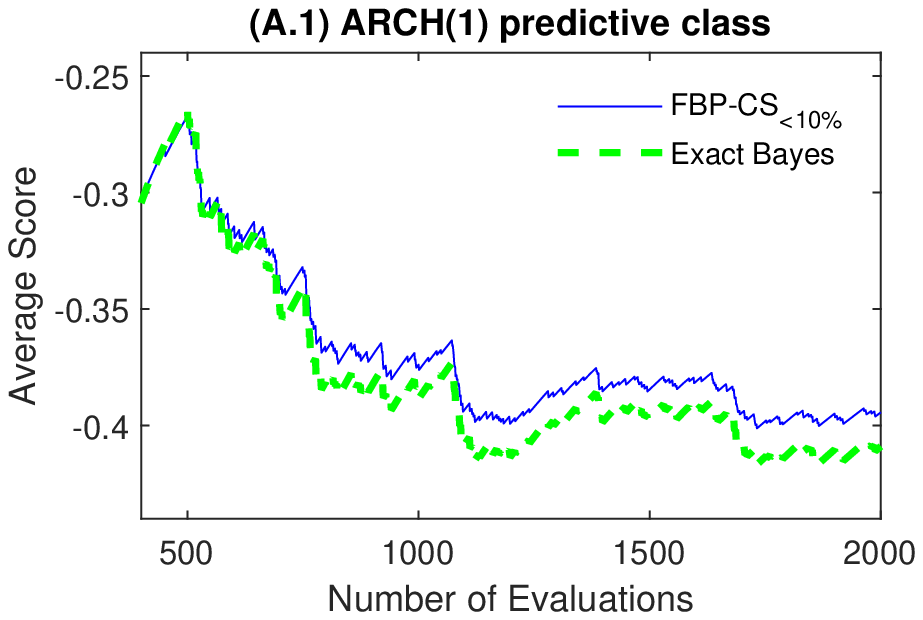} & 
\includegraphics[scale=
0.7]{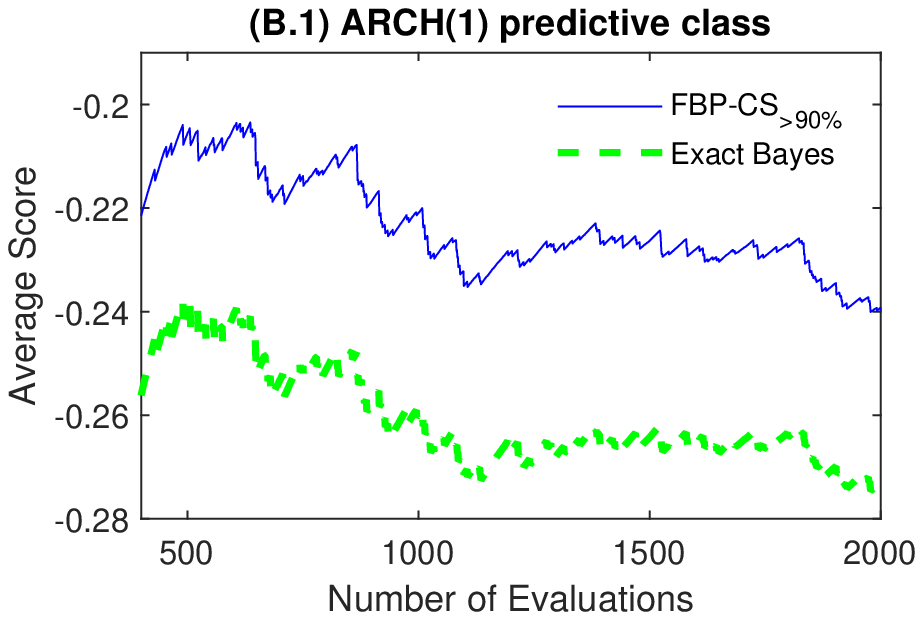} \\ 
&  \\ 
\includegraphics[scale= 0.7]{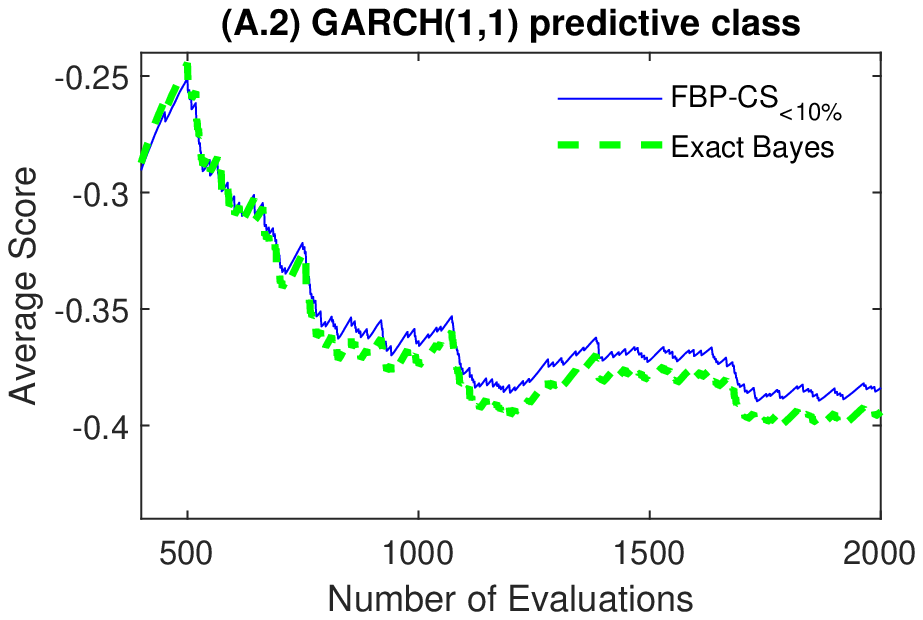} & 
\includegraphics[scale=
0.7]{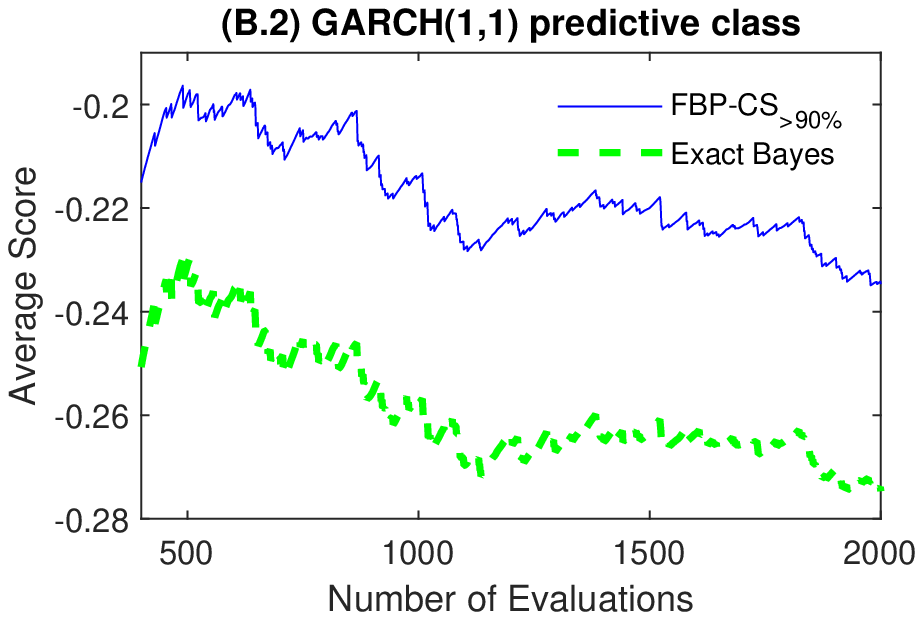} \\ 
&  \\ 
\includegraphics[scale= 0.7]{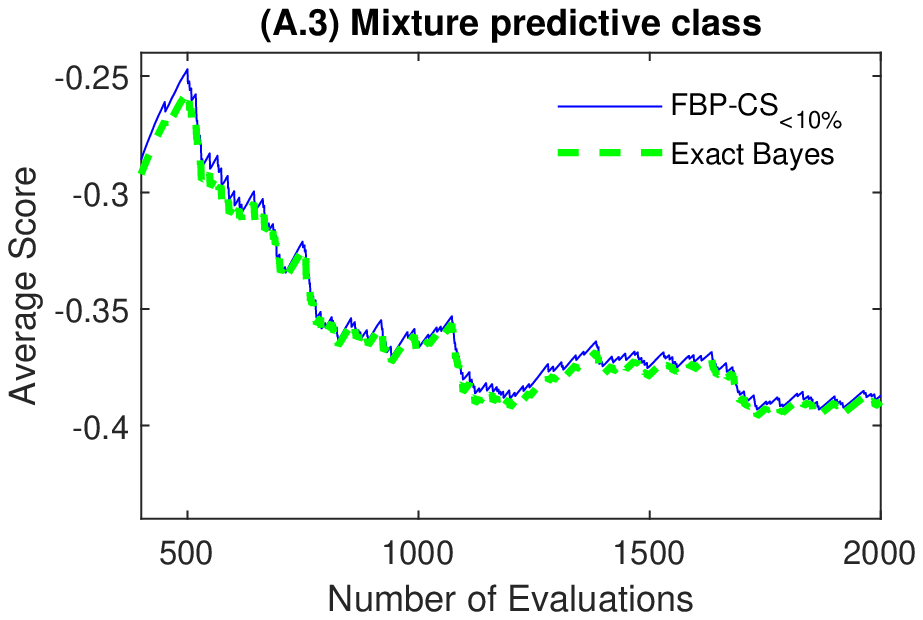} & 
\includegraphics[scale=
0.7]{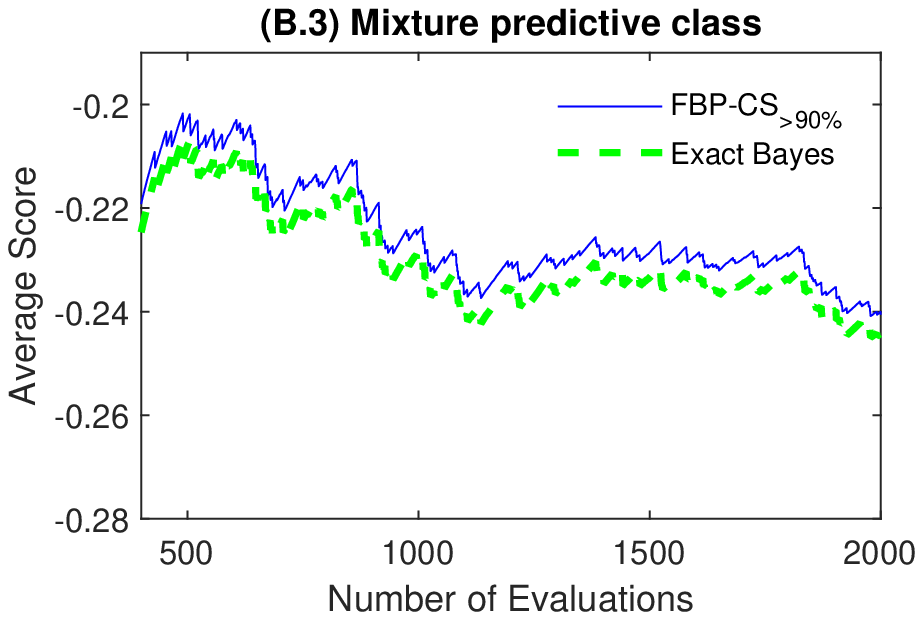}%
\end{tabular}%
\caption{{\protect\footnotesize Out-of-sample performance of exact Bayes and
CS-based FBP: relative predictive accuracy in the lower tail (Panel A) and
the upper tail (Panel B). The three plots in Panel A depict the cumulative
average (over an expanding evaluation period) of CS{$_{<10\%}$} for the
exact Bayes (dashed green line) and {FBP-}CS{$_{<10\%}$ (full blue line)
updates, using the ARCH(1), GARCH (1,1) and mixture predictive classes
respectively. The three plots in Panel B display the cumulative average of }%
CS{$_{>90\%}$ for the }exact Bayes (dashed green line) and {FBP-}CS{$%
_{>90\%} $ (full blue line) updates, using the three different predictive
classes.}}}
\label{Fig:relativeaccuracy}
\end{figure}

In Panel B, the superiority of FBP{-}CS{$_{>90\%}$ }over exact Bayes, in
terms of accurately predicting extremely \textit{large} returns is in stark
evidence. In this case, the relative performance of the two {updating}
methods is less affected by the move from the ARCH(1) to the GARCH(1,1)
predictive class. However, once again, {use of the} more flexible mixture of
predictives to underpin the exact Bayes update brings its performance much
closer to that of the focused {update}, with the accuracy of the latter
being reasonably robust to the choice of predictive class.

In the following section we provide some insight into \textit{why} the
focused {update} in the \textit{upper} tail reaps more benefit out-of-sample
than does the lower-tail {update}, relative to exact Bayes, and the role
that misspecification plays here.

\subsubsection{Animation of the mean predictives\label{animation}}

In Figures \ref{Fig:animate_low} and \ref{Fig:animate_up} we display
animated plots of the one-step-ahead mean predictives produced using
expanding windows of $n=500$ to $n=2,499$, and based solely on the (most
misspecified) Gaussian ARCH(1) predictive class. The mean predictives
produced by both {updating} methods (FBP and exact Bayes) are superimposed
upon the true predictive, produced using simulation from (\ref{h})-(\ref%
{skew}). Figure \ref{Fig:animate_low} presents the results for lower tail
focus ({FBP-}CS{$_{<10\%}$ versus exact Bayes) }and Figure \ref%
{Fig:animate_up} the results for upper tail focus (FBP{-}CS{$_{>90\%}$
versus exact Bayes). The vertical lines in each plot indicate the return
that defines the quantile }$A$ in (\ref{csr_prelim}). 
\begin{figure}[h]
\begin{center}
\animategraphics[controls,loop,scale =0.75]{3}{Animate_out_low}{}{} 
\end{center}
\caption{{\protect\footnotesize Animation over time of mean predictives
based on the exact Bayes (dashed green curve) and {FBP-CS$_{<10\%}$ (full
blue curve) updates. The red curve is the true predictive, produced using
simulation from (\protect\ref{h})-(\protect\ref{skew}). The predictive class
used is Gaussian ARCH(1). The black vertical line denotes the threshold used
in the FBP-CS$_{<10\%}$ update, in the production of the mean predictive for
time period }$n+1.$}}
\label{Fig:animate_low}
\end{figure}
\begin{figure}[h]
\begin{center}
\animategraphics[controls,loop,scale =0.75]{3}{Animate_out_up}{}{} 
\end{center}
\caption{{\protect\footnotesize Animation over time of mean predictives
based on the exact Bayes (dashed green curve) and {FBP-CS$_{>90\%}$ (full
blue curve) updates. The red curve is the true predictive, produced using
simulation from (\protect\ref{h})-(\protect\ref{skew}). The predictive class
used is Gaussian ARCH(1). The black vertical line denotes the threshold used
in the FBP-CS$_{>90\%}$ update, in the production of the mean predictive for
time period }$n+1.$}}
\label{Fig:animate_up}
\end{figure}

Two things are clear from Figure \ref{Fig:animate_low}: one, the lower tails
of both the {FBP-}CS{$_{<10\%}$ }and exact Bayes predictives are quite
similar; two, both tails are - for some time points - quite good at picking
up the shape of the true predictive tail, but with the {FBP-}CS{$_{<10\%}$ }%
predictive tail being a better match most of the time. These plots thus
provide some explanation of the summary results in Panel A (column 3) of
Table \ref{Table 1} and Panel A.1 of Figure 1, in which {FBP-}CS{$_{<10\%}$ }%
dominates exact Bayes, but with the improvement in forecast accuracy being
reasonably small, despite the misspecification of the predictive class.

In Figure \ref{Fig:animate_up} however, the animated display is very
different. The upper tail of the exact Bayes predictive \textit{consistently}
fails to pick up the shape of the true: the misspecified nature of the
Gaussian ARCH(1) model has a marked impact on predictive accuracy in this
region of the support of $y_{n+1}.$ In contrast, FBP{-}CS{$_{>90\%}$ }has
the flexibility to focus only what matters - upper tail predictive accuracy
- and, as such, produces predictives with upper tails that are much closer
in shape to the true, and which are often visually indistinguishable from
the true. These plots thus explain the clear numerical dominance of FBP{-}CS{%
$_{>90\%}$ }over exact Bayes in Panel A (column 6) of Table \ref{Table 1}
and Panel B.1 of Figure \ref{Fig:relativeaccuracy}.

We finish by noting that focus on upper tail accuracy \textit{does} - as
highlighted by the relevant figures in the middle columns of Panel A in
Table \ref{Table 1} - impact quite severely on the ability of FBP{-} CS{$%
_{>90\%}$ to }pick up the lower tail of the true predictive. This outcome
highlights the fact that the \textit{ex-ante }decision as to \textit{what }%
form of accuracy to focus on is critical, and most notably so in the very
misspecified case.

\subsubsection{The differential effect of posterior variation\label%
{Sect:sect433}}

When adopting the conventional Bayesian paradigm for prediction, a single
question needs to be addressed: which model (or set of models) is to be used
to produce the predictive distribution? Once that model (or set of models)
has been chosen, computational methods are used to integrate out the
posterior uncertainty associated with that choice, and a single (marginal)
predictive distribution thereby produced. Posterior parameter (and model)
uncertainty affects the location, shape, and degree of dispersion of the
marginal predictive, and any predictive conclusions drawn from it; however,
it is not the convention to explicitly quantify the impact of posterior
variation on prediction.

Our new proposal introduces an additional choice into the mix: which measure
of predictive accuracy is to drive the production of a predictive
distribution? Each different form of in-sample update serves as a different
`window' through which a choice of predictive model (or mixture of
predictive models) - and all posterior uncertainty associated with that
choice - impinges on predictive outcomes. For example, one choice of update
may yield a posterior distribution over a given predictive class that is
very diffuse; another choice may lead to a very concentrated posterior.
Hence, finite sample posterior variation itself has import, since it is not
unique, even given a particular choice of predictive class.

We illustrate this point using $M=4,000$ posterior draws from (\ref{post})
for both FBP-CS$_{<10\%}$\ and FBP-CS$_{>90\%}$, using the Gaussian ARCH(1)
predictive class, and for selected values of $n+1$ between $501$ and $2,500$%
. From the draws from (\ref{post}) based on the FBP-CS$_{<10\%}$ update, we
produce the corresponding $4,000$ values of the {ES }for period $n+1$\ for a
portfolio that is `long' in the asset,\footnote{%
See, for example, {\cite{embrechts1997}.}} 
\begin{equation}
\text{ES}_{0.1}\left( P_{\boldsymbol{\theta }}^{n}\right) =\bm{-}%
\int_{y<A_{0.1}}yp\left( y|\mathcal{F}_{n},\boldsymbol{\theta }\right) dy%
\text{.}  \label{ES_long}
\end{equation}%
The integral bound $A_{0.1}$ in (\ref{ES_long}) denotes the $10\%$\
quantile, and ES$_{0.1}\left( P_{\boldsymbol{\theta }}^{n}\right) $ denotes
the (negative of the) mean of the random variable $y_{n+1}$\ conditional on
the future return falling into the \textit{lower} 10\% tail of $p\left(
\cdot |\mathcal{F}_{n},\boldsymbol{\theta }\right) $; that is, the
(negative) expected return in a worst case scenario. For a `short'
portfolio, in which extremely large returns are harmful, the relevant
function is:%
\begin{equation}
\text{ES}_{0.9}\left( P_{\boldsymbol{\theta }}^{n}\right)
=\int_{y>A_{0.9}}yp\left( y|\mathcal{F}_{n},\boldsymbol{\theta }\right) dy,
\label{ES_short}
\end{equation}%
where $A_{0.9}$ in (\ref{ES_short}) denotes the $90\%$\ quantile, and ES$%
_{0.9}\left( P_{\boldsymbol{\theta }}^{n}\right) $ denotes the mean of $%
y_{n+1}$\ conditional on the future return falling into the \textit{upper}
10\% tail of $p\left( \cdot |\mathcal{F}_{n},\boldsymbol{\theta }\right) .$
In this case we produce $4,000$ values of ES$_{0.9}\left( P_{\boldsymbol{%
\theta }}^{n}\right) $ using the draws from (\ref{post}) based on the FBP-CS$%
_{>90\%}$ update. For the Gaussian ARCH(1) predictive class, ES$_{0.1}\left(
P_{\boldsymbol{\theta }}^{n}\right) $\ and ES$_{0.9}\left( P_{\boldsymbol{%
\theta }}^{n}\right) $ have closed-form solutions, for any given value of $%
\boldsymbol{\theta }.$

In each case we use kernel density estimation to produce an estimate of the
marginal posterior of the scalar ES, based on the $4,000$ draws from (\ref%
{post}).\textbf{\ }As a comparator in each case, we perform the same
exercise but using the exact (likelihood-based) Bayes update to define (\ref%
{post}).\ The `true' values of ES$_{0.1}$ and ES$_{0.9}$, computed using
simulation from (\ref{h})-(\ref{skew}), are reproduced on the respective
plots as red vertical lines. All computations are performed for the selected
values of $n+1$, and animated graphics are used to illustrate how all
quantities change over time.

\begin{figure}[tbp]
\begin{center}
\animategraphics[controls,loop,scale =0.75]{1}{Animate_out_long}{}{} %
\animategraphics[controls,loop,scale =0.75]{1}{Animate_out_short}{}{}
\end{center}
\caption{{\protect\footnotesize Animation over time of the posterior
densities for: ES$_{0.1}\left( \boldsymbol{\protect\theta }\right) $ (Panel
A), and ES$_{0.9}\left( \boldsymbol{\protect\theta }\right) $ (Panel B). In
Panel A, the posterior densities are based on the exact Bayes (dashed green
curve) and {FBP-CS$_{<10\%}$ (full blue curve) updates. In Panel B, the
posterior densities are based on the exact Bayes (dashed green curve) and
FBP-CS$_{>90\%}$ (full blue curve) updates. The red vertical line in each
panel is the true expected shortfall (ES$_{0.1}\boldsymbol{\ }$in Panel A
and ES$_{0.9}\boldsymbol{\ }$in Panel B), produced using simulation from (%
\protect\ref{h})-(\protect\ref{skew}). Note: both on the horizontal axis and
in the label for each panel we use the superscript `n+1' to remind the
reader that the quantities plotted relate to prediction of ES for time point
n+1. In the animation, n+1 in the panel label changes to reflect the actual
time point.}}}
\label{Fig:ES}
\end{figure}

From Panel A in\textbf{\ }Figure \ref{Fig:ES}\textbf{\ }we see results that
broadly confirm the lack of dominance of the FBP-CS$_{<10\%}$ update over
exact Bayes, in terms of accurately reproducing the lower tail of the true
predictive. Neither the FBP-based posterior of ES$_{0.1}\left( P_{%
\boldsymbol{\theta }}^{n}\right) $ (the full blue curve), nor the exact
Bayes posterior (the dashed green curve) is located uniformly closer to the
true vertical (red) value than the other, across time. That said, the
posterior variation in the exact Bayes posterior is \textit{always} less
than that of the FBP posterior, for the selected time points considered.

In contrast, in Panel B in Figure \ref{Fig:ES} there is a much more marked
tendency for the FBP posterior of ES$_{0.9}\left( P_{\boldsymbol{\theta }%
}^{n}\right) $ to be located closer to the true value of ES$_{0.9}$ than is
the exact Bayes posterior, in addition to being much more concentrated.
Hence, there are several instances over time in which the FBP posterior is
extremely concentrated around - or very near to - the true ES: providing a
further illustration of the benefits reaped by focusing on upper tail
accuracy in the update.

\section{Empirical Illustrations\label{empapp}}

We now illustrate the potential of the focused approach in empirical
settings. In Section \ref{fin} we produce predictive results for two
empirical return series, using the same predictive classes adopted in the
simulation experiments above. We document the accuracy of the FBP posterior
mean predictive in a similar manner to the documentation of the results in
Table \ref{Table 1}. {Furthermore, we illustrate the practical benefits of
FBP using two empirically relevant loss measures: one, empirical exceedances
for predictive VaRs; two, }out-of-sample accuracy for ES, as measured by a 
\textit{consistent} class of scoring functions (\citealp{ziegel2020robust}).{%
\ In Section \ref{m4} we provide a quite different - and ambitious -
empirical illustration, by using FBP to predict the 23,000 annual time
series from the 2018 `M4' forecasting competition. {In all illustrations we
estimate the mean predictive in (\ref{Eq:pred_density_BHW_general}) (for the
relevant predictive class and updating rule) using (after thinning) 4,000
MCMC posterior draws of }$\boldsymbol{\theta }$},{{\ and base all
out-of-sample assessments on this mean predictive. The MCMC scheme remains
as described in Appendix \ref{comp}, }}apart from certain minor
modifications required for the `M4' example, as detailed in Section \ref{m4}.

\subsection{Financial returns\label{fin}}

\subsubsection{Preliminary diagnostics}

The two series used in the first empirical illustration are: i) $4,000$
observations of daily returns on the U.S. dollar currency index (DXY), from
3 Jan 2000\textbf{\ }to 3 Nov\textbf{\ }2015; and ii) $4,000$ observations
of {daily }returns on the S\&P500 index, from 3 Jan 1996 to 3 Feb 2012. {Both%
} series {are} supplied by the Securities Industries Research Centre of Asia
Pacific (SIRCA) on behalf of Reuters. All returns are continuously
compounded. To match the simulation exercise, the last 2,000 observations in
each series are used to perform all out-of-sample assessments, with the
one-step-ahead mean predictives produced in the same manner as described in
Section \ref{sim_results} (using expanding estimation samples),\textbf{\ }%
apart from the fact that the first estimation sample for both empirical
series is of length $n=2,000$ (rather than $n=500$). {We also adopt the same
three predictive classes, defined by the Gaussian ARCH(1), Gaussian
GARCH(1,1) and mixture models.\footnote{%
As in the simulation exercise, the parameters in the constituent models in
the mixture are {estimated via maximum likelihood.}}}

As tallies with the typical features exhibited by financial {returns}, the
descriptive statistics reported in Table \ref{tab:summaries} provide
evidence of time-varying volatility (significant serial correlation in
squared returns) and marginal non-Gaussianity (significant non-normality in
the level of returns) in both series.{\ Hence, we can conclude that the
simple }Gaussian{\ ARCH(1) and GARCH(1,1) predictive classes are likely to
be misspecified, and more so than the more flexible mixture class. }As such,
we would anticipate accuracy gains - by using FBP rather than exact Bayes -
to be most in evidence for the ARCH(1) predictive class, with decreasing
relative gains expected as the predictive class becomes less misspecified.

\begin{table}[tbp]
\centering
\scalebox{0.7}{
	\begin{tabular}{lrrrrrrrrrr}
		\hline\hline
		        &         &        &        &        &        &        &          &          &         &         \\
		        &     Min &    Max &   Mean & Median & St.Dev &  Range & Skewness & Kurtosis & JB stat & LB stat \\
		        &         &        &        &        &        &        &          &          &         &         \\
		DXY     &  -2.913 &  1.645 & -0.010 & -0.007 &  0.316 &  4.558 &   -0.303 &    4.517 &    3461 &     271 \\
		S\&P500 & -21.070 & 19.510 &  0.035 &  0.164 &  2.590 & 40.581 &   -0.260 &    5.762 &    5579 &    2783 \\ \hline\hline
	\end{tabular}}
\caption{{\protect\footnotesize Summary statistics. `JB stat' is the test
statistic for the Jarque-Bera test of normality, with a critical value of
5.99. `LB stat' is the test statistic for the Ljung-Box test of serial
correlation in the squared returns; the critical value based on a lag length
of 3 is 7.82. `Skewness' is the Pearson measure of sample skewness, and
`Kurtosis' a sample measure of excess kurtosis. The labels `Min' and `Max'
refer to the smallest and largest value, respectively, while `{Range}' is
the difference between these two. The remaining statistics have the obvious
interpretations. }}
\label{tab:summaries}
\end{table}

\begin{table}[tbp]
\begin{center}
\scalebox{0.7}{
			\begin{tabular}{lcccccccc}
				\hline\hline
				                             &         &                          &  &                  &                  &  &                  &                  \\
				                             &                            \multicolumn{8}{c}{\textbf{Panel A: ARCH(1) predictive class}}                            \\
				                             &         &                          &  &                  &                  &  &                  &                  \\
				                             &                                   \multicolumn{8}{c}{\textbf{Out-of-sample score}}                                   \\ \cline{2-9}
				                             &         &                          &  &                  &                  &  &                  &                  \\
				                             & \multicolumn{2}{c}{Center Focused} &  &  \multicolumn{2}{c}{Left Focused}   &  &  \multicolumn{2}{c}{Right Focused}  \\
				                             &   LS    &           CRPS           &  &     FSR10\%      &     FSR20\%      &  &     FSR80\%      &     FSR90\%      \\
				\textbf{Updating method}    &         &                          &  &                  &                  &  &                  &                  \\ \cline{1-1}
				                             &         &                          &  &                  &                  &  &                  &                  \\
				\underline{\textit{DXY}}     &         &                          &  &                  &                  &  &                  &                  \\
				Exact Bayes                  & -0.2528 &         -0.1701          &  &     -0.2810      &     -0.4113      &  &     -0.3953      &     -0.2729      \\
				FBP                          & -0.2528 &     \textbf{-0.1689}     &  & \textbf{-0.2676} & \textbf{-0.3921} &  & \textbf{-0.3901} & \textbf{-0.2692} \\
				\underline{\textit{S\&P500}} &         &                          &  &                  &                  &  &                  &                  \\
				Exact Bayes                  & -2.2984 &         -1.3002          &  &     -0.4975      &     -0.8182      &  &     -0.7209      &     -0.4191      \\
				FBP                          & -2.2984 &     \textbf{-1.2922}     &  & \textbf{-0.4661} & \textbf{-0.7859} &  & \textbf{-0.7100} & \textbf{-0.4090} \\ \hline\hline
				                             &         &                          &  &                  &                  &  &                  &                  \\
				                             &                          \multicolumn{8}{c}{\textbf{Panel B: GARCH(1,1) predictive class}}                           \\
				                             &         &                          &  &                  &                  &  &                  &                  \\
				                             &                                   \multicolumn{8}{c}{\textbf{Out-of-sample score}}                                   \\ \cline{2-9}
				                             &         &                          &  &                  &                  &  &                  &                  \\
				                             & \multicolumn{2}{c}{Center Focused} &  &  \multicolumn{2}{c}{Left Focused}   &  &  \multicolumn{2}{c}{Right Focused}  \\
				                             &   LS    &           CRPS           &  &     FSR10\%      &     FSR20\%      &  &     FSR80\%      &     FSR90\%      \\
				\textbf{Updating method}    &         &                          &  &                  &                  &  &                  &                  \\ \cline{1-1}
				                             &         &                          &  &                  &                  &  &                  &                  \\
				\underline{\textit{DXY}}     &         &                          &  &                  &                  &  &                  &                  \\
				Exact Bayes                  & -0.1798 &         -0.1664          &  &     -0.2484      &     -0.3712      &  &     -0.3832      &     -0.2660      \\
				FBP                          & -0.1798 &     \textbf{-0.1659}     &  & \textbf{-0.2445} & \textbf{-0.3656} &  & \textbf{-0.3771} & \textbf{-0.2582} \\
				\underline{\textit{S\&P500}} &         &                          &  &                  &                  &  &                  &                  \\
				Exact Bayes                  & -2.1257 &         -1.2433          &  &     -0.4251      &     -0.7410      &  &     -0.6531      &     -0.3573      \\
				FBP                          & -2.1257 &     \textbf{-1.2422}     &  & \textbf{-0.4213} & \textbf{-0.7350} &  & \textbf{-0.6510} & \textbf{-0.3569} \\ \hline\hline
				                             &         &                          &  &                  &                  &  &                  &                  \\
				                             &                            \multicolumn{8}{c}{\textbf{Panel C: Mixture predictive class}}                            \\
				                             &         &                          &  &                  &                  &  &                  &                  \\
				                             &                                   \multicolumn{8}{c}{\textbf{Out-of-sample score}}                                   \\ \cline{2-9}
				                             &         &                          &  &                  &                  &  &                  &                  \\
				                             & \multicolumn{2}{c}{Center Focused} &  &  \multicolumn{2}{c}{Left Focused}   &  &  \multicolumn{2}{c}{Right Focused}  \\
				                             &   LS    &           CRPS           &  &     FSR10\%      &     FSR20\%      &  &     FSR80\%      &     FSR90\%      \\
				\textbf{Updating method}    &         &                          &  &                  &                  &  &                  &                  \\ \cline{1-1}
				                             &         &                          &  &                  &                  &  &                  &                  \\
				\underline{\textit{DXY}}     &         &                          &  &                  &                  &  &                  &                  \\
				Exact Bayes                  & -0.1859 &         -0.1664          &  & \textbf{-0.2559} & \textbf{-0.3796} &  &     -0.3795      &     -0.2616      \\
				FBP                          & -0.1859 &         -0.1664          &  &     -0.2566      &     -0.3801      &  & \textbf{-0.3791} & \textbf{-0.2606} \\
				\underline{\textit{S\&P500}} &         &                          &  &                  &                  &  &                  &                  \\
				Exact Bayes                  & -2.1261 &         -1.2437          &  & \textbf{-0.4239} & \textbf{-0.7397} &  &     -0.6543      &     -0.3586      \\
				FBP                          & -2.1261 &     \textbf{-1.2428}     &  &     -0.4246      &     -0.7401      &  & \textbf{-0.6528} & \textbf{-0.3575} \\ \hline\hline
			\end{tabular}   
		}
\end{center}
\caption{{\protect\footnotesize Predictive accuracy based on FBP and exact
Bayes mean predictives. Panels A to C report the average out-of-sample
scores for the ARCH(1), GARCH(1,1) and mixture predictive class,
respectively. The rows in each panel refer to the update method used. The
columns refer to the out-of-sample measure used to compute the average
scores. {Each FBP figure is based on an }update{\ that matches the given
out-of-sample accuracy measure. }The bold figures are the largest average
scores according to a given out-of-sample measure. The two time series are
labelled as \textit{DXY} and \textit{S\&P500}.}}
\label{Tab:empresults}
\end{table}

\begin{table}[tbp]
\centering
\scalebox{0.7}{
	\begin{tabular}{lcccccccccccccc}
		\hline\hline
		                   &                 &                 &                &                &  &                &                                    &                &                &  &                &                &                 &                 \\
		                   &       \multicolumn{4}{c}{\textbf{Panel A: Simulated dataset}}       &  &                       \multicolumn{4}{c}{\textbf{Panel B: DXY}}                       &  &            \multicolumn{4}{c}{\textbf{Panel C: S\&P500}}            \\
		                   &                 &                 &                &                &  &                &                                    &                &                &  &                &                &                 &                 \\
		                   &       \multicolumn{4}{c}{\textbf{Out-of-sample exceedances}}        &  &                \multicolumn{4}{c}{\textbf{Out-of-sample exceedances}}                 &  &       \multicolumn{4}{c}{\textbf{Out-of-sample exceedances}}        \\ \cline{2-5}\cline{7-10}\cline{12-15}
		                   &                 &                 &                &                &  &                &                                    &                &                &  &                &                &                 &                 \\
		                   &   VaR$_{0.1}$   &   VaR$_{0.2}$   &  VaR$_{0.8}$   &  VaR$_{0.9}$   &  &  VaR$_{0.1}$   &            VaR$_{0.2}$             &  VaR$_{0.8}$   &  VaR$_{0.9}$   &  &  VaR$_{0.1}$   &  VaR$_{0.2}$   &   VaR$_{0.8}$   &   VaR$_{0.9}$   \\
		\textbf{Updating} &                 &                 &                &                &  &                &                                    &                &                &  &                &                &                 &                 \\
		\textbf{method}    &                 &                 &                &                &  &                &                                    &                &                &  &                &                &                 &                 \\
		                   &                 &                 &                &                &  &                &                                    &                &                &  &                &                &                 &                 \\
		Exact Bayes        &      0.117      &      0.196      &     0.217      &     0.062*     &  &     0.073*     &               0.142*               &     0.157*     &     0.085      &  &     0.081*     &     0.152*     &     0.131*      &     0.060*      \\
		FBP-CS$_{<10\%}$   & \textbf{0.102 } & \textbf{0.198 } &     0.05*      &     0.002*     &  & \textbf{0.084} &               0.237*               &     0.027*     &     0.011*     &  & \textbf{0.096} &     0.225*     &     0.023*      &     0.008*      \\
		FBP-CS$_{<20\%}$   &      0.103      &      0.203      &     0.023*     &     0.000*     &  &     0.077*     & \multicolumn{1}{r}{\textbf{0.180}} &     0.059*     &     0.022*     &  &     0.083      & \textbf{0.192} &     0.039*      &     0.017*      \\
		FBP-CS$_{>80\%}$   &     0.338*      &     0.413*      & \textbf{0.197} & \textbf{0.104} &  &     0.031*     &               0.060*               & \textbf{0.204} &     0.094      &  &     0.050*     &     0.102*     &     0.164*      &     0.072*      \\
		FBP-CS$_{>90\%}$   &     0.471*      &     0.551*      &     0.148*     &     0.085      &  &     0.015*     &               0.037*               &     0.247*     & \textbf{0.100} &  &     0.042*     &     0.074*     & \textbf{0.182*} & \textbf{0.078*} \\ \hline\hline
	\end{tabular}}
\caption{{\protect\footnotesize Predictive Value at Risk assessment. The
figures recorded are the proportion of times that the observed out-of-sample
values `exceed' the predictive VaR$_{\protect\alpha }$\ indicated by the
column heading. Panels A to C, respectively, record results for the data
simulated from (\protect\ref{h})-(\protect\ref{skew}),{\ the DXY empirical
data, and the S\&P500} empirical data. The bold value in each column
indicates the empirical coverage that is closest to the nominal tail
probability, whilst an asterisk indicates rejection of the null hypothesis
of independence and correct coverage at the 1\% level of significance, using
the Christoffersen test.}}
\label{Tab:VaRresults}
\end{table}

{In the following sections we assess the relative performance of FBP and
exact Bayes for these two series in terms of, respectively: average
out-of-sample scores, VaR exceedances, }and the average scoring function for
ES proposed in \cite{ziegel2020robust}. For the third exercise, we assess
the \textit{dominance} of FBP over exact Bayes - in terms of predicting ES -
by visually comparing the results using the Murphy diagrams proposed in \cite%
{ehm2016quantiles}.

\subsubsection{Results: Average out-of-sample scores}

In Table \ref{Tab:empresults},\ we {reproduce results for both the
likelihood-based }update{\ (exact Bayes) and the FBP }update {that matches
the out-of-sample accuracy measure used} (as captured by the relevant
scoring rule).{\footnote{%
The choice of $w_{n}$ for each FBP method remains as described in Section %
\ref{overview}.}} That is, for each of the two empirical series, and for
each predictive class, there is a single row of accuracy results labelled
`FBP', with the update underlying the FBP figure in any particular column 
\textit{matching} the accuracy measure in the column label.

The results confirm our expectations. In Panel A, the figures based on the
ARCH(1) predictive class tell a clear story: using an update that focuses on
the measure that is assessed out-of-sample reaps accuracy gains. In \textit{%
all} cases, and for \textit{both} series, the exact (misspecified) Bayesian
predictives are out-performed by FBP. In Panel B, FBP based on the
GARCH(1,1) class continues to dominate exact Bayes uniformly; however the
degree of dominance is less marked than in Panel A. The dominance of FBP
over exact Bayes is no longer \textit{uniform} in Panel C, for the case of
the (least misspecified) mixture predictive class. Nevertheless, in all four
instances, FBP still out-performs exact Bayes in the upper tail.

Hence, the empirical results \textit{mimic} the patterns observed in the
simulation setting, and continue to send the clear signal: when model
misspecification is marked, FBP is beneficial, and particularly in terms of
upper tail accuracy in the case of negatively skewed data. In the following
sections we assess the \textit{practical significance} of that benefit by
conducting an out-of-sample assessment of the $\alpha \times 100\%$
predictive VaR for period\textbf{\ }$n+1$ - or predictive quantile (denoted
by VaR$_{\alpha }$) - {and the $\alpha \times 100\%$ }ES (denoted by{\textbf{%
\ }}ES$_{\alpha }$),{\ both} computed from the relevant mean predictive, and
for the most misspecified class only: Gaussian ARCH(1). For comparative
purposes, we also perform the exercise for the simulated data analyzed in
Section \ref{sim}.\textbf{\ }Note that, for ease of notation we do not make
explicit the dependence of both VaR$_{\alpha }$ and ES$_{\alpha }$ on a
particular predictive distribution, unless this is necessary.

\subsubsection{Results: VaR exceedances\label{var}}

{To assess predictive accuracy of the VaR, f}or all updating methods and all
return series (both empirical and simulated), we first compute the
probability of `exceedance', $\hat{\alpha}$, as the proportion of realized
out-of-sample values that are \textit{less} than the predictive VaR$_{\alpha
}$ for $\alpha =0.1$ and $\alpha =0.2$, as is relevant for a long portfolio.
We then compute $1-\hat{\alpha}$, as the proportion of realized values that
are \textit{greater} than the predictive VaR$_{\alpha }$ for $\alpha =0.8$
and $\alpha =0.9$, as pertains to a short portfolio. Table \ref%
{Tab:VaRresults} reports the values of $\hat{\alpha}$ (or $1-\hat{\alpha}$),
for the exact Bayes update and four different versions of FBP that focus
explicitly on tail accuracy: FBP-CS$_{<10\%}$, FBP-CS$_{<20\%}$, FBP-CS$%
_{>80\%}$ and FBP-CS$_{>90\%}.$ The bold figure in each column indicates the
empirical exceedance that is closest to the nominal tail probability. The
asterisk then indicates whether the null hypothesis of $\hat{\alpha}=\alpha $
(or $1-\hat{\alpha}=1-\alpha $) \textit{and} independence in the exceedances
is rejected at the $1\%$ significance level by the \cite%
{christoffersen1998evaluating} test. In this particular illustration there
is, of course, no exact match between update and out-of-sample measure;
however we would anticipate that the FBP methods that focus on accuracy in
the lower tail would yield better predictions of VaR$_{0.1}$ and VaR$_{0.2}$%
\ (and, hence, better associated performance statistics) than would both
exact Bayes (with no lower tail focus) and the FBP methods that focus on
accuracy in the upper tail; with the corresponding conclusions expected for
upper tail focus, and accurate prediction of VaR$_{0.8}$ and VaR$_{0.9}$.\
Hence the usefulness of reproducing all five sets of results for each
scenario.

Looking first at the bold figures in Table \ref{Tab:VaRresults}, a simple
conclusion can be drawn: for all three series, an `appropriate' form of FBP
(i.e. with an update that rewards accuracy in the relevant tail) produces
empirical exceedances that are closer to the nominal values than {do both}
exact Bayes and the `inappropriate' FBP (i.e. with an update that rewards
accuracy in the opposite tail). By this measure, focussing \textit{correctly}
reaps VaR accuracy benefits for all three series. For the DXY series ({Panel
B}) we can hone this conclusion further: the updates that focus on accuracy
in a particular \textit{marginal} tail (remembering that the threshold used
in the CS score is based on an estimate of a marginal quantile) also yield
the best empirical exceedances for the conditional VaR$_{\alpha }$ with an
equivalent probability. All four bold diagonal exceedances in {Panel B} are
also insignificantly different from the nominal value, and are not
associated with rejection of the null of independent violations. Indeed,
with the exception of the FBP-CS$_{>80\%}$ exceedance for FBP-CS$_{>90\%}$,
these four diagonal figures are the \textit{only }ones associated with a
failure to reject the joint null of correct coverage and independent
violations.

For the data simulated from (\ref{h})-(\ref{skew}) (Panel A), with one
exception, all FBP methods that focus on the tail that is relevant for VaR$%
_{\alpha }$\ prediction - so the lower tail for VaR$_{0.1}$ and VaR$_{0.2}$,
and the upper tail for VaR$_{0.8}$ and VaR$_{0.9}$\ - yield statistics that
fail to reject the joint null. Moreover, as is somewhat consistent with the
particular dominance of FBP over exact Bayes in the \textit{upper} tail
illustrated in Section \ref{sim_results}, there is a slight\ tendency for
FBP to also be \textit{more} superior to exact Bayes in terms of prediction
of the \textit{upper} tail VaR$_{\alpha }$'s$.$\ For the S\&P500 data, the
FBP methods that focus on lower tail accuracy (FBP-CS$_{<10\%}$ {and} FBP-CS$%
_{<20\%}$) are the only ones that do not formally reject the joint null -
when used to predict VaR$_{0.1}$ and VaR$_{0.2}.$ However, again, in terms
of raw exceedances for VaR$_{\alpha }$'s in a particular tail, the FBP
methods that reward accuracy in that same tail outperform {everything} else.

\subsubsection{Results: Murphy diagrams for ES predictions}

Lastly, we compare the performance of FBP and exact Bayes in terms of the
accuracy with which we predict {ES}. Reiterating that these results are
based on the \textit{mean} predictives (for both FBP and exact Bayes), for
the Gaussian ARCH(1) class, this means that in the definition $\text{ES}%
_{\alpha }(P_{\boldsymbol{\theta }}^{n})$ in equations \eqref{ES_long}-%
\eqref{ES_short}, we replace $P_{\boldsymbol{\theta }}^{n}$ with the mean
predictive defined in (\ref{Eq:pred_density_BHW_general}), and for both FBP
and exact Bayes updating.{\footnote{%
The lack of a closed-form solution for the mean predictive requires that ES
must be estimated at each time point. Herein, ES is estimated via Monte
Carlo integration based on a large number of draws from the mean predictive.}%
} We then follow \cite{ziegel2020robust} and measure the out-of-sample
accuracy of both forms of ES predictions using a class of scoring functions
that are indexed by some known parameter $\eta \in \mathbb{R}$, and which
are \textit{consistent} for the joint functional of VaR and ES. The specific
details of the scoring function, as well as the definition of consistency
are delayed until Appendix \ref{app:ES}.

{With reference to a }positively-oriented{\ class} of consistent scoring
functions, \cite{ziegel2020robust} {say that} predictive Method A, e.g. the
mean predictive obtained under FBP, \textit{dominates} predictive method B,
e.g. the mean predictive obtained under exact Bayes, if on average
predictions of ES obtained under Method A yield {a score that is greater
than or equal to the score that is obtained under Method B,} uniformly over
all scoring functions in the class (i.e., all $\eta \in \mathbb{R}$).%
\footnote{%
A more formal definition of dominance is provided in Appendix \ref{app:ES}.}
Sampling variability aside, dominance can then be visualized using the
Murphy diagram proposed in \cite{ehm2016quantiles}, which plots the average
difference between the scores under Methods A and B, over the
out-of-sample-evaluation period and across a range of values for $\eta $.
For the class of consistent scoring functions defined in Appendix \ref%
{app:ES}, and across a grid of values for $\eta $, it is then feasible to
evaluate if predictions made using FBP dominate those made under exact Bayes
(when both approaches are based on the same {predictive} class) in terms of
their accuracy {in predicting} {ES}.

{To this end, we produce and discuss Murphy diagrams for both the DXY and
S\&P500 return series. }Panels {A.1 and A.2 of Figure}~\ref%
{Fig:relativeaccuracy_ES_DXY} present {the Murphy diagrams} for ES$_{0.1}$
and ES$_{0.2}$, respectively, for the DXY {series}. In each panel,
predictions {based on} the `appropriate' form of FBP are compared to exact
Bayes. {The black solid line displays the }average difference (over the
out-of-sample period){\ }between the scores calculated under FBP and exact
Bayes, across a grid of values for $\eta $.{\textbf{\ }}The evaluation
period used to calculate the average difference is the same as that used to
compute the VaR exceedances in Section \ref{var}. The shaded region
corresponds to a 95\% bootstrap confidence interval for the average
difference, calculated at each $\eta $, based on a block bootstrap (%
\citealp{kunsch1989jackknife}) with block length of size $b=10$ and $1000$
replications, while the red dashed line is the horizontal line at $y=0$.%
\footnote{%
We remind the reader that only sampling variability characterizes these
out-of-sample computations; hence\textbf{\ }the production of a
(bootstrap-based) \textit{confidence} interval at each value of $\eta $.}
From these two plots the black line indicates that, for virtually all values
of $\eta $, FBP outperforms exact Bayes, since the average score difference
is generally positive and the confidence intervals usually do not encompass
zero. This result supports the claim that FBP helps improve predictive
performance in the lower tail of DXY returns. On the other hand, Panels B.1
and B.2 indicate that no such dominance is observed in terms of\textbf{\ }%
the upper-tail ES predictions, although in Panel B.1 the bulk of the
confidence intervals do cover positive values, and the FBP forecasts are not
dominated by exact Bayes in either panel.

{Figure~\ref{Fig:relativeaccuracy_ES_SP500} presents the corresponding
results for the S\&P500 returns}. Panels\ A.1 and B.1 show that there is no\
dominance {of FBP over exact Bayes for }ES$_{0.1}$\textbf{\ }and ES$_{0.9}${%
. }On the other hand, Panels A.2 and B.2 show more evidence of FBP having
better predictive performance than exact Bayes. In particular, in Panel B.2,
for most\textbf{\ }values of $\eta $ FBP outperforms exact Bayes, with the
average score difference being\textbf{\ }generally positive and most of%
\textbf{\ }the confidence intervals not encompassing zero.\ For completeness
we also include the results for the data simulated from (\ref{h})-(\ref{skew}%
). In this case we have illustrated, in Section \ref{Sect:sect433}, the
particular accuracy of the FBP-based ES results in the upper 10\% tail
(relative to exact Bayes). This finding is supported by the stark result in
Panel B.1, in which virtually complete dominance of FBP over exact Bayes is
on display. 
\begin{figure}[h]
\centering 
\begin{tabular}{cc}
\hspace{1.5cm}\textbf{Panel A:} & \hspace{1cm}\textbf{Panel B:} \\ 
\hspace{1.5cm}\textbf{Lower tail expected shortfall} & \hspace{1cm} \textbf{%
Upper tail expected shortfall} \\ 
&  \\ 
\multicolumn{2}{c}{\includegraphics[scale= 1]{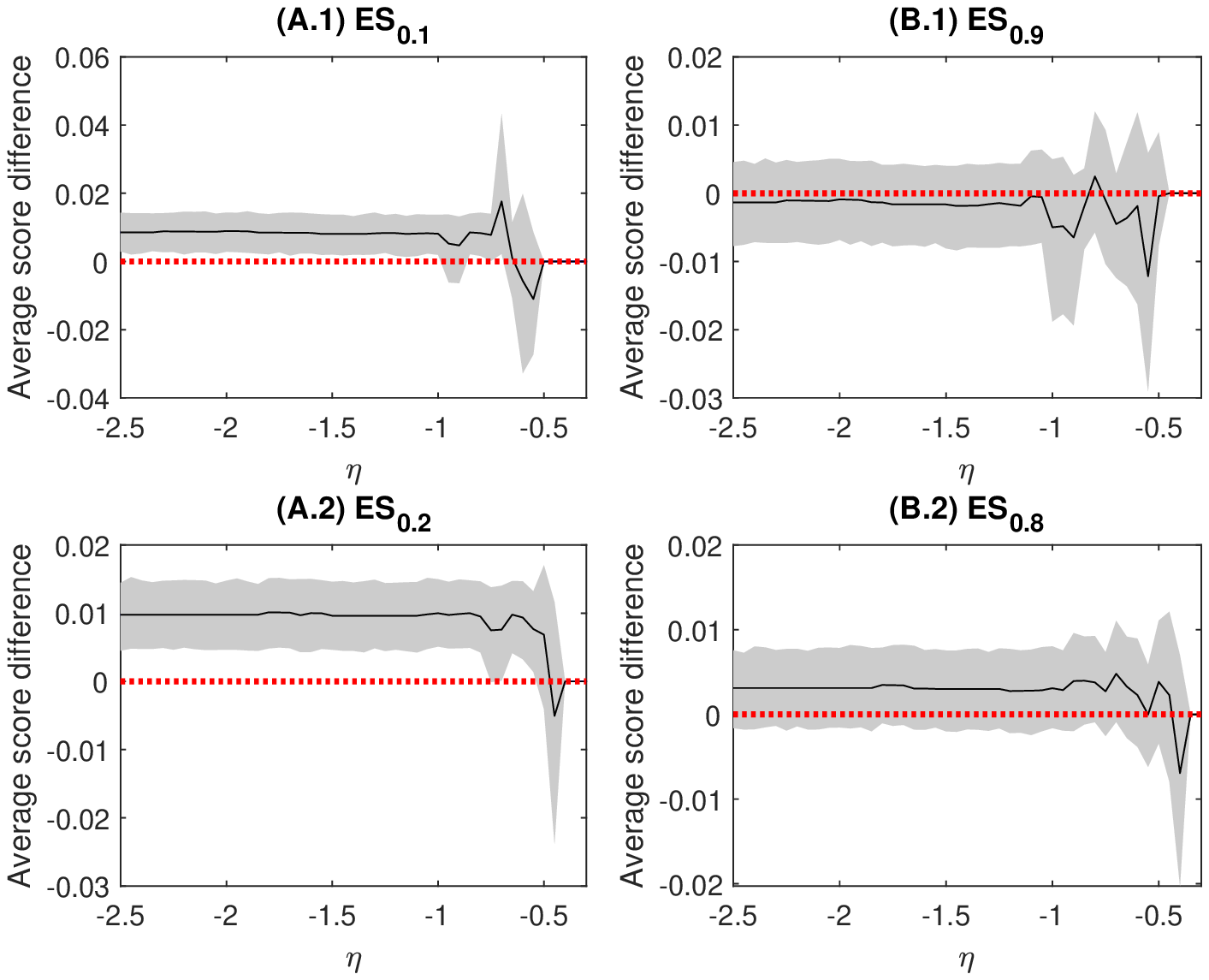}} \\ 
& 
\end{tabular}%
\caption{{\protect\footnotesize ES predictive performance for the DXY
series, measured by the average difference between the FBP and exact Bayes
out-of-sample scores. The scoring rule is as defined in (\protect\ref%
{Eq:ESscore}). Panels (A.1) and (A.2) correspond to ES}$_{0.1}$%
{\protect\footnotesize \ and ES}$_{0.2}${\protect\footnotesize ,
respectively; Panels (B.1) and (B.2) correspond to ES}$_{0.9}$%
{\protect\footnotesize \ and ES}$_{0.8}${\protect\footnotesize ,
respectively. The solid black line indicates the average difference; the
shaded area denotes the 95\% bootstrap confidence interval for each }$%
\protect\eta ${\protect\footnotesize ; and the red dashed line is a
horizontal line at zero. Panels (A.1), (A.2), (B.1) and (B.2) are
constructed using the FBP-CS}$_{\text{$<$10\%}}${\protect\footnotesize ,
FBP-CS}$_{\text{$<$20\%}}${\protect\footnotesize , FBP-CS}$_{\text{$>$90\%}}$%
{\protect\footnotesize \ and FBP-CS}$_{\text{$>$80\%}}$%
{\protect\footnotesize \ methods, respectively, as the relevant FBP method.}}
\label{Fig:relativeaccuracy_ES_DXY}
\end{figure}

\begin{figure}[h]
\centering 
\begin{tabular}{cc}
\hspace{1.5cm}\textbf{Panel A:} & \hspace{1cm}\textbf{Panel B:} \\ 
\hspace{1.5cm}\textbf{Lower tail expected shortfall} & \hspace{1cm} \textbf{%
Upper tail expected shortfall} \\ 
&  \\ 
\multicolumn{2}{c}{\includegraphics[scale= 1]{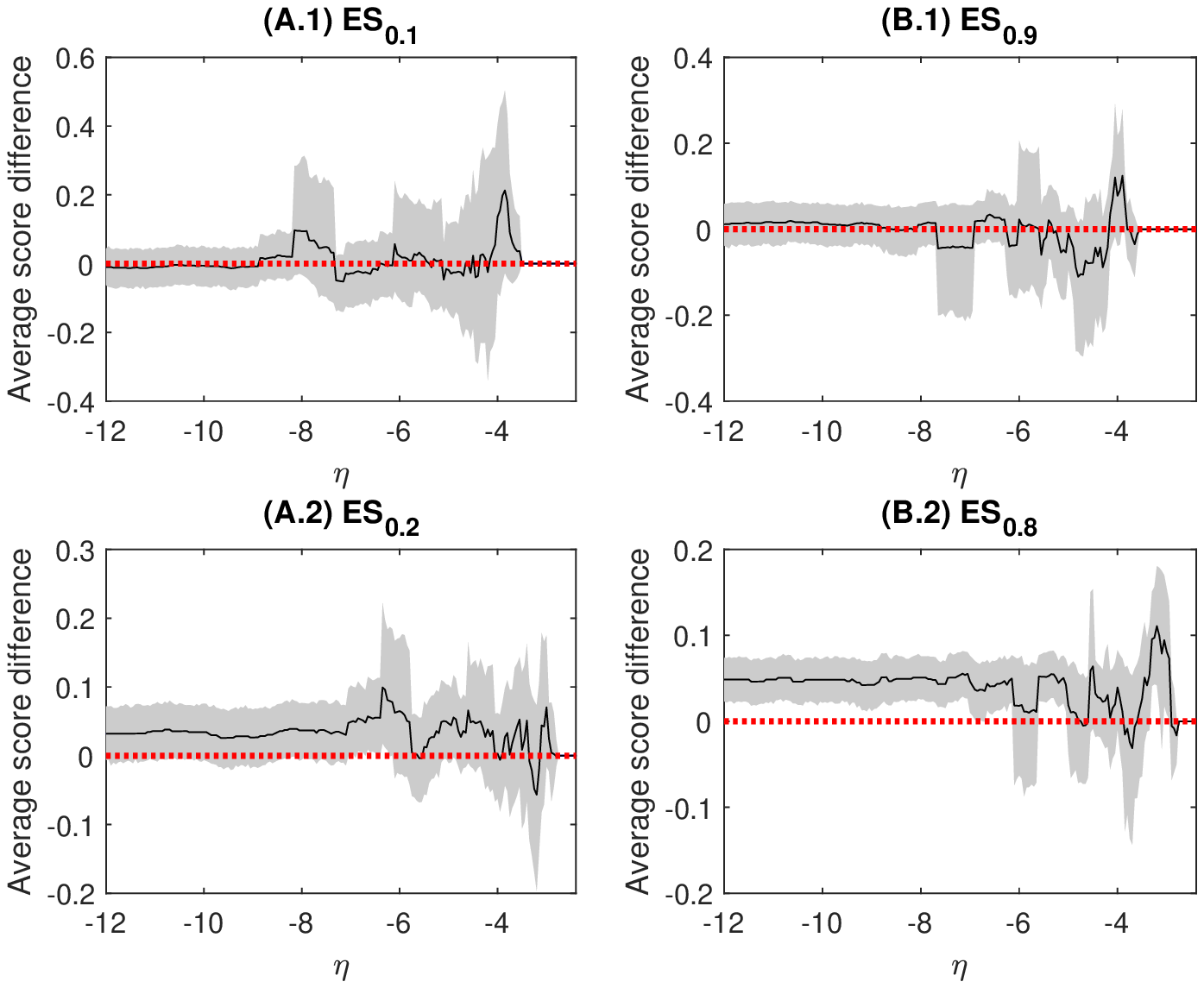}} \\ 
& 
\end{tabular}%
\caption{{\protect\footnotesize ES predictive performance for the S\&P500
series, measured by the average difference between the FBP and exact Bayes
out-of-sample scores. The scoring rule is as defined in (\protect\ref%
{Eq:ESscore}). Panels (A.1) and (A.2) correspond to ES}$_{0.1}$%
{\protect\footnotesize \ and ES}$_{0.2}${\protect\footnotesize ,
respectively; Panels (B.1) and (B.2) correspond to ES}$_{0.9}$%
{\protect\footnotesize \ and ES}$_{0.8}${\protect\footnotesize ,
respectively. The solid black line indicates the average difference; the
shaded area denotes the 95\% bootstrap confidence interval for each }$%
\protect\eta ${\protect\footnotesize ; and the red dashed line is a
horizontal line at zero. Panels (A.1), (A.2), (B.1) and (B.2) are
constructed using the FBP-CS}$_{\text{$<$10\%}}${\protect\footnotesize ,
FBP-CS}$_{\text{$<$20\%}}${\protect\footnotesize , FBP-CS}$_{\text{$>$90\%}}$%
{\protect\footnotesize \ and FBP-CS}$_{\text{$>$80\%}}$%
{\protect\footnotesize \ methods, respectively, as the relevant FBP method.}}
\label{Fig:relativeaccuracy_ES_SP500}
\end{figure}

\begin{figure}[h]
\centering 
\begin{tabular}{cc}
\hspace{1.5cm}\textbf{Panel A:} & \hspace{1cm}\textbf{Panel B:} \\ 
\hspace{1.5cm}\textbf{Lower tail expected shortfall} & \hspace{1cm} \textbf{%
Upper tail expected shortfall} \\ 
&  \\ 
\multicolumn{2}{c}{\includegraphics[scale= 1]{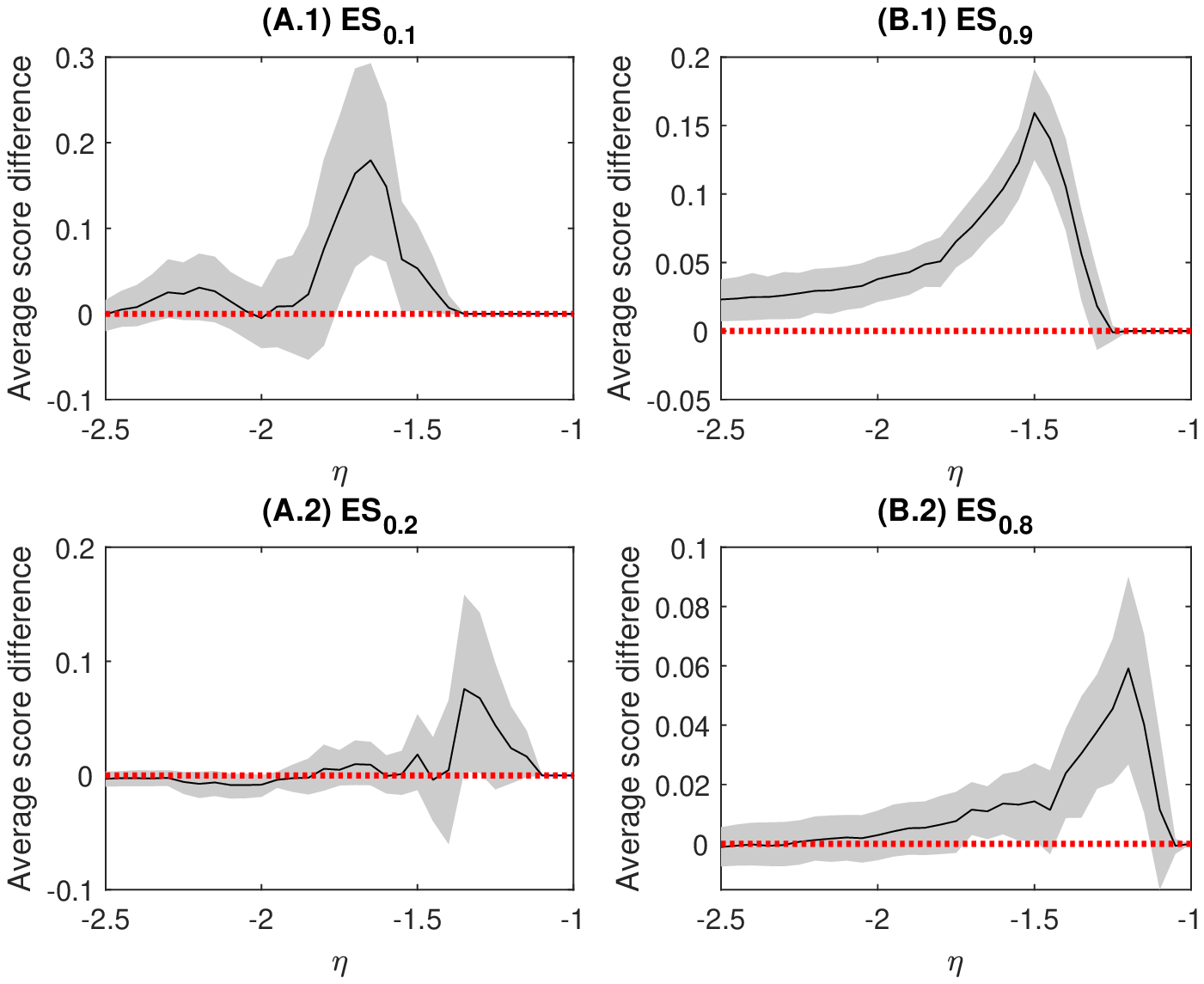}} \\ 
& 
\end{tabular}%
\caption{{\protect\footnotesize ES predictive performance for\ the data
simulated from (\protect\ref{h})-(\protect\ref{skew}), measured by the
average difference between the FBP and exact Bayes out-of-sample scores. The
scoring rule is as defined in (\protect\ref{Eq:ESscore}). Panels (A.1) and
(A.2) correspond to ES}$_{0.1}${\protect\footnotesize \ and ES}$_{0.2}$%
{\protect\footnotesize , respectively; Panels (B.1) and (B.2) correspond to
ES}$_{0.9}${\protect\footnotesize \ and ES}$_{0.8}${\protect\footnotesize ,
respectively. The solid black line indicates the average difference; the
shaded area denotes the 95\% bootstrap confidence interval for each }$%
\protect\eta ${\protect\footnotesize ; and the red dashed line is a
horizontal line at zero. Panels (A.1), (A.2), (B.1) and (B.2) are
constructed using the FBP-CS}$_{\text{$<$10\%}}${\protect\footnotesize ,
FBP-CS}$_{\text{$<$20\%}}${\protect\footnotesize , FBP-CS}$_{\text{$>$90\%}}$%
{\protect\footnotesize \ and FBP-CS}$_{\text{$>$80\%}}$%
{\protect\footnotesize \ methods, respectively, as the relevant FBP method.}}
\label{Fig:relativeaccuracy_ES_simul}
\end{figure}

\subsection{M4 forecasting competition\label{m4}}

\label{Sect:M4}

{The M4 competition was an exploration of forecast performance organized by
the University of Nicosia and the New York University Tandon School of
Engineering in 2018. A total of 100,000 time series - of differing
frequencies and lengths - were made available to the public. Each
forecasting expert (or expert team) was then to submit a vector of }$h${%
-step ahead forecasts, }$h=1,2,...,H,${\ for each of the series. The winner
of the competition in a particular category was the expert (team) who
achieved the best average out-of-sample predictive accuracy according to the
measure of accuracy that defined that category, over all horizons and all
series.\footnote{%
Details of all aspects of the competition can be found via the following
link: \href{https://www.m4.unic.ac.cy/wp-content/uploads/2018/03/M4-Competitors-Guide.pdf}%
{M4 competition details.}} }

{One category was concerned with predictive interval accuracy, as measured
by the {mean scaled interval score} (MSIS) proposed in \cite%
{gneiting2007strictly}. The MSIS formula used in the competition, defined
over the $100\left( 1-\alpha \right) $\% prediction interval, is given by 
\begin{equation*}
\text{MSIS}=\frac{1}{H}\frac{\sum_{h=1}^{H}\left( u_{t+h}-l_{t+h}+\frac{2}{%
\alpha }\left( l_{t+h}-y_{t+h}\right) \boldsymbol{1}\{y_{t+h}<l_{t+h}\}+%
\frac{2}{\alpha }\left( y_{t+h}-u_{t+h}\right) \boldsymbol{1}%
\{y_{t+h}>u_{t+h}\}\right) }{\frac{1}{n-m}\sum_{t=m+1}^{n}|y_{t}-y_{t-m}|},
\end{equation*}%
where $H$ denotes the longest predictive horizon considered, $l_{t+h}$ and $%
u_{t+h}$ denote the $100\left( \frac{\alpha }{2}\right) $\% and $100\left( 1-%
\frac{\alpha }{2}\right) $\% predictive quantile, respectively, $y_{t+h}$ is
the realized value at time $t+h$, }$h=1,2,...,H$, and $m$ denotes the
frequency of the data. {The overall predictive accuracy according to this
score was measured by the mean MSIS over the 100,000 series. }

The fourth{\ best performance in this particular forecasting category was
achieved by }the `\textit{M4 team'}, {who produced prediction intervals
using a model with an exponential smoothing structure for the level, trend
and seasonal components, and a Gaussian distributional assumption, referred
to as the ETS model hereafter \citep{hyndman2002state}. A}ll three smoothing
components of this model can be either additive or multiplicative in
structure, while the trend component can also be specified to be damped;
that is, the trend component can be forced to taper off over longer
predictive horizons. A variety of ETS {specifications} were fit to each
series using maximum likelihood estimation (MLE), and the best specification
then selected via Akaike's Information Criterion (AIC). {The ETS model used
to define the predictive interval was then characterized by }at most{\ four
parameters: $\theta _{1}$, the smoothing parameter in the level; $\theta
_{2} $, the smoothing parameter in the trend; $\theta _{3}$, the smoothing
parameter in the seasonality; }and $\theta ${$_{4}$ the damping parameter in
the trend; with $\boldsymbol{z}$} denoting{\ the vector of{\ }}initial{\
states.\footnote{%
It is worth noting that for the ETS model there is no variance parameter;
instead the variance of the residuals is employed to construct the
predictive densities.} }The\ MLEs of the parameters and {$\boldsymbol{z}$},
for the selected model, were then `plugged into' the assumed Gaussian
predictive.

We now use the ETS model as the predictive class to which FBP is applied,
with the MSIS score used to define the {update}. To keep the exercise
computationally feasible we adopt the same (MLE/AIC-based) strategy as the 
\textit{M4 team} for choosing the specification of the\ ETS model for each
series.{\footnote{%
We employ the ETS function of the forecast package in R to perform this
exercise.}} We also perform the exercise only\textit{\ }for the 23,000 
\textit{annual} time series, and for $H=6.$\ For each of the $i$ annual
series, $i=1,2,...,23,000$, we define: the full sample vector of
observations of length $n_{i}$ as {$\mathbf{y}_{n_{i}}$; }the vector of
observations up to time $t$ as $\mathbf{y}_{i,t}$; {the observed value at
any particular time point $t+h$ as $y_{i,t+h}$; the h-step-ahead }$i^{th}$ {%
ETS predictive }distribution as $P_{\boldsymbol{\theta }_{i}}^{t}${\ (}with $%
\boldsymbol{\theta }_{i}${\ denoting the unknown parameters for the }%
selected model for the $i^{th}$ series); {and the $100\left( \frac{\alpha }{2%
}\right) $\% and $100\left( 1-\frac{\alpha }{2}\right) $\% quantiles of the }%
$i^{th}$ {ETS predictive density }as {$l_{i,t+h}$ and $u_{i,t+h}$
respectively. Using this notation (and acknowledging a slight abuse of the
notation used earlier for the \textit{one-step-ahead} predictives and
associated scores), we define the (positively-oriented) MSIS-based sample
criterion as }$S_{n_{i}}(P_{\boldsymbol{\theta }_{i}}^{n_{i}})=\sum%
\nolimits_{t=0}^{n_{i}-1}S_{MSIS}(P_{\boldsymbol{\theta }_{i}}^{t},\bm{y}%
_{i,t,H}),${\ where}\ 
\begin{align}
S_{MSIS}(P_{\boldsymbol{\theta }_{i}}^{t},\bm{y}_{i,t,H})=-\frac{1}{%
{\scriptsize \text{min$\left( H,n_{i}-t\right) $}}}\sum_{h=1}^{\text{min}%
\left( H,n_{i}-t\right) }\left( {}\right. & \left. u_{i,t+h}-l_{i,t+h}+\frac{%
2}{\alpha }\left( l_{i,t+h}-y_{i,t+h}\right) \boldsymbol{1}%
\{y_{i,t+h}<l_{i,t+h}\}\right.  \label{MSIS} \\
+& \left. \frac{2}{\alpha }\left( y_{i,t+h}-u_{i,t+h}\right) \boldsymbol{1}%
\{y_{i,t+h}>u_{i,t+h}\}\right)  \notag
\end{align}%
and $\bm{y}_{i,t,H}=\left( y_{i,t+1},\dots ,y_{i,t+\text{min}\left(
H,n_{i}-t\right) }\right) ^{\prime }$. {Because there are 23,000 time
series, and associated posteriors of the form of (\ref{post}), the scale
factor for each series }$i$, call it $w_{i,n_{i}}$, {needs to be selected
via a computationally efficient and automated method. We set each }$%
w_{i,n_{i}}$ as $w_{i,n_{i}}=n_{i}d_{\theta _{i}}/2S_{n_{i}}(P_{\hat{%
\boldsymbol{\theta }_{i}}}^{n_{i}}),$ {where }$\widehat{\boldsymbol{\theta }}%
_{i}$ denotes the {MLE }of $\boldsymbol{\theta }_{i}$, and $d_{\theta _{i}}$
is the dimension of $\boldsymbol{\theta }_{i}$. This choice of $w_{i,n_{i}}$
ensures that the scale and convergence of the FBP posterior variance are
anchored to those of a well understood benchmark posterior. More details are
provided in Appendix \ref{msis}. Posterior draws of the ETS predictives for
the $i^{th}$ series are defined by the draws of the underlying $\boldsymbol{%
\theta }_{i}$, which are, in turn, produced using the same MCMC scheme
described in Appendix \ref{comp} for the GARCH(1,1) model, with appropriate
adjustments made for the bounds on the ETS parameters.\footnote{%
The parameter restrictions for the ETS model are: $0<\theta _{1},\theta
_{2},\theta _{3}<1$ and $0.8<\theta _{4}<0.98$.}

{Table~\ref{Tab:M4results} documents the accuracy with which five
alternative methods forecast the 23,000 annual series, with accuracy
measured {exclusively} by the (positively-oriented) MSIS rule. The methods
are the top four performers in the full M4 competition (in which all 100,000
series are forecast, and accuracy is measured by MSIS) - denoted by }M4-1$^{%
\text{st}}$, M4-2$^{\text{nd}}$, M4-3$^{\text{rd}}$\ and M4-4$^{\text{th}}${%
\ respectively - and the focused Bayesian method with }update {based on the
MSIS rule, denoted simply by FBP. The four first place-getters are, in brief,%
} the hybrid method of exponential smoothing and recurrent neural networks
in \cite{smyl2019hybrid} (M4-1$^{\text{st}}$), the feature-based forecast
combination method by \cite{montero2019fforma} (M4-2$^{\text{nd}}$), the
Card forecasting method in \cite{doornik2019card} (M4-3$^{\text{rd}}$),{\
with }M4-4$^{\text{th}}$,{\ as noted earlier, being the ETS/MLE-based method 
}of the\ {\textit{M4 team. }We use the} forecasts {provided by the
competitors (available in the M4 package in R) to} assess the predictive
accuracy of{\ }M4-1$^{\text{st}}$, M4-2$^{\text{nd}}$, M4-3$^{\text{rd}}$\
and M4-4$^{\text{th}}${\ for the 23,000 annual series.}

\begin{table}[tbp]
\begin{center}
{\normalsize \ 
\scalebox{0.72}{
					\begin{tabular}{lccccc}
						\hline\hline
						                     &                    &                    &                    &                    &                 \\
						                     &                 \multicolumn{5}{c}{\textbf{Panel A: Out-of-sample MSIS summaries}}                  \\
						                     &                    &                    &                    &                    &                 \\
						                     &                           \multicolumn{5}{c}{\textbf{Predictive method}}                            \\ \cline{2-6}
						                     &                    &                    &                    &                    &                 \\
						                     & M4-1$^{\text{st}}$ & M4-2$^{\text{nd}}$ & M4-3$^{\text{rd}}$ & M4-4$^{\text{th}}$ &       FBP       \\
						\textbf{Summary}     &                    &                    &                    &                    &                 \\
						                     &                    &                    &                    &                    &                 \\
						Mean                 &  \textbf{-23.90}   &       -27.48       &       -30.20       &       -34.90       &     -34.04      \\
						Median               &       -16.18       &       -16.09       &       -18.47       &       -15.49       & \textbf{-14.70} \\
						SD                   &   \textbf{48.17}   &       65.37        &       65.30        &       76.84        &      75.90      \\
						Quant 1\%            &  \textbf{-183.17}  &      -297.98       &      -214.84       &      -348.52       &     -340.00     \\
						Quant 10\%           &       -32.23       &  \textbf{-28.12}   &       -52.33       &       -66.71       &     -65.12      \\
						Quant 20\%           &       -24.14       &  \textbf{-19.52}   &       -35.07       &       -34.50       &     -33.14      \\
						Quant 30\%           &       -20.45       &  \textbf{-17.72}   &       -27.19       &       -24.43       &     -23.28      \\
						Quant 40\%           &       -18.02       &  \textbf{-16.77}   &       -22.20       &       -19.02       &     -17.95      \\
						Quant 50\%           &       -16.18       &       -16.09       &       -18.47       &       -15.49       & \textbf{-14.70} \\
						Quant 60\%           &       -14.62       &       -15.54       &       -15.35       &       -12.71       & \textbf{-12.21} \\
						Quant 70\%           &       -13.29       &       -15.02       &       -12.65       &       -10.45       & \textbf{-10.36} \\
						Quant 80\%           &       -12.01       &       -14.41       &       -10.43       &   \textbf{-8.95}   &      -8.96      \\
						Quant 90\%           &       -10.28       &       -13.41       &       -8.35        &       -7.16        & \textbf{-7.03}  \\
						Quant 99\%           &       -5.29        &       -7.25        &       -4.17        &   \textbf{-3.30}   &      -3.47      \\ \hline\hline
						                     &                    &                    &                    &                    &                 \\
						                     &               \multicolumn{5}{c}{\textbf{Panel B: Grouped predictive performance}  }                \\
						                     &                    &                    &                    &                    &                 \\
						                     &                           \multicolumn{5}{c}{\textbf{Predictive method}}                            \\ \cline{2-6}
						                     &                    &                    &                    &                    &                 \\
						                     & M4-1$^{\text{st}}$ & M4-2$^{\text{nd}}$ & M4-3$^{\text{rd}}$ & M4-4$^{\text{th}}$ &       FBP       \\
						\textbf{Time series} &                    &                    &                    &                    &                 \\
						\textbf{group}       &                    &                    &                    &                    &                 \\
						                     &                    &                    &                    &                    &                 \\
						Group 1              &        1359        &   \textbf{1732}    &        766         &        304         &       439       \\
						Group 2              &        814         &   \textbf{2136}    &        405         &        555         &       690       \\
						Group 3              &        989         &        558         &        557         &        1118        &  \textbf{1378}  \\
						Group 4              &        522         &         58         &        813         &        1534        & \textbf{ 1673}  \\
						Group 5              &        338         &         48         &        1006        &        1565        &  \textbf{1643}  \\
						Total                &        4022        &        4532        &        3547        &        5076        &  \textbf{5823}  \\ \hline\hline
					\end{tabular}
			} }
\end{center}
\caption{{\protect\footnotesize Predictive accuracy of five competing
methods for the 23,000 annual time series from the M4 competition. The
labels M4-1$^{\text{st}}$, M4-2$^{\text{nd}}$, M4-3$^{\text{rd}}$\ and M4-4$%
^{\text{th}}$\ denote the first, second, third and fourth best methods
(overall) in the MSIS category of M4. The label FBP refers to the focused
Bayesian prediction method applied to the ETS predictive class and using the
MSIS rule as the update. Panel A records various summaries of the 23,000
(positively-oriented) MSIS values; Panel B reports the number of series for
which a particular method performs best in each of the five groups described
in the text. {The bold font is used to indicate the best performing method
according to a given performance measure.}}}
\label{Tab:M4results}
\end{table}
Columns 1 to 4{\ of Panel A in }Table~\ref{Tab:M4results}{\ report various
summaries of the 23,000 MSIS values {- mean, median, standard deviation, and
multiple quantiles - }for each of the four competition methods. Based on the
mean results, the ranking of the four methods for the 23,000 annual series
is seen to equal their ranking in the full competition (as indicated by the
column labels).} Column 5 presents the summary statistics for FBP. First, we
observe that both the mean and the median of the{\ }MSIS values for FBP are 
\textit{larger} than those for M4-4$^{\text{th}}${\ (ETS/MLE)}, which
indicates that focusing \textit{does }enhance predictive performance over
and above simple use of the maximizer of the\textbf{\ }likelihood function
of the (inevitably misspecified) ETS model. Second, we observe that in terms
of median MSIS, FBP produces the \textit{best} results out of all five
methods considered. By looking at the quantiles of the distribution of the
23,000 scores for all methods we can glean why this is so. In particular,
the values of MSIS at the 1$^{st}$ quantile for M4-4$^{\text{th}}$ and FBP
are -348.52 and -340.00 respectively, whilst for M4-1$^{\text{st}}$, M4-2$^{%
\text{nd}}$ and M4-3$^{\text{rd}}$ the corresponding values are -183.17,
-297.98 and -214.84, respectively. In other words both ETS-based methods
forecast 1\% of the series quite poorly, which impacts on their mean MSIS.
In contrast, the ETS-based methods - and particularly FBP - forecast quite
accurately those series that produce MSIS values that are in the middle of
the distribution, or in the upper tail.

To provide further insight here, we divide the 23,000 series into five
different groups of 4,600 series. Group 1 has the 4,600 series with the
smallest average MSIS,\textbf{\textit{\ }}computed across the five methods
(i.e. the series with the worst average results according to this measure),
while Group 5 comprises the 4,600 series with the largest average MSIS
values. The remaining groups span the middle. For each of these groups we
count how many series are best predicted by each of the five methods. The
results are presented in Panel B of Table~\ref{Tab:M4results}. First, we
observe that for\textit{\ all }groups FBP has a larger number of favourable
results than does the likelihood-based ETS method (M4-4$^{\text{th}}$).
Second, for the last three groups FBP outperforms \textit{all four} other
methods. Last, observe that the method that is best at predicting the highest%
\textbf{\ }number of series overall is FBP, with a total of 5,823 series
(out of 23,000) for which it is the best performer. Focusing on predictive
interval accuracy in the update has clearly yielded benefits out-of-sample.

\section{Discussion{\label{section:discuss}}}

We have proposed a new paradigm for Bayesian prediction that can deliver
accurate predictions in whatever metric is most meaningful for the problem
at hand. By replacing the conventional likelihood function in the Bayesian
update with an appropriate function of a proper scoring rule, the resultant
posterior - by construction - gives high probability mass to predictive
distributions that yield high scores. Under regularity, the posterior
asymptotically concentrates onto the predictive that maximizes the expected
scoring rule and, {in this sense}, yields the most accurate predictions in
the given class. Movement away from the conventional Bayesian approach does
involve the choice of a scale factor, which determines the relative weight
of the prior and data-based components of the posterior. However, this
factor can be chosen via common sense criteria in finite samples.

The choice of predictive class used is not pre-determined, and can be any
plausible class of predictives, or combination of predictives, that {captures%
} the features of the data {that} matter for prediction. In particular, the
class can be deliberately selected to be computationally simple; {the aim of 
}the {approach }\textit{not}{\ being to} perfectly capture \textit{all}
aspects of the true DGP, but to achieve a particular type of\textbf{\ }%
predictive accuracy via a suitable Bayesian update. Whilst the emphasis in
the paper has been on classes of parametric predictives, or finite
combinations of predictives, the paradigm, in principle, applies to
nonparametric predictive settings also. Similarly, the focus on time series
data, and `forecasting' future values of random variables indexed by time,
has also been a choice on our part, and not something that is intrinsic to
the methodology.

{The scope and value of our }new approach to prediction{\ are demonstrated
using illustrations with simulated and empirical data. The benefits reaped
by `focusing' on the type of predictive accuracy that matters are stark,
with }virtually \textit{all }numerical results showing gains over and above
conventional (misspecified) likelihood-based prediction. Whilst a reduction
in the\textbf{\ }degree of misspecification of the predictive class does
lead to more homogeneous predictions across methods, for a large enough
sample the superiority of FBP is still almost always in evidence, even when
the predictive class is a reasonable representation of the true DGP.

Some obvious avenues of future research remain open. First, the ability to
generate accurate predictions via a simplistic representation of the truth
means that models with, for example, intractable likelihood functions, can
now be treated using exact simulation methods like MCMC, rather than via
approximate methods like approximate Bayesian computation (%
\citealp{sisson2018overview}) or Bayesian synthetic likelihood (%
\citealp{price2018bayesian}). Comparison between the results yielded by FBP
and those yielded by predictives based on such approximate methods of
inference (e.g. \citealp{frazier2019approximate}) would be of interest.
Second, in cases where the set of unknowns is very large, and an MCMC
simulation method is likely to be inefficient, variational methods could be
explored. In particular {- and in the spirit of the }\textit{generalized
variational inference }{proposed by \cite{GVI2019} }- the choice of
variational family, {and/or the }distance measure to be minimized, could be
driven by the specific prediction focus of any problem. Finally, all results
in this paper have been based on proper scoring rules only. This is by no
means essential, but has simply been a decision made to render the scope of
the paper manageable. In practice, any loss function in which predictive
accuracy plays a role (e.g. financial loss functions associated with optimal
portfolios of predicted returns; asymmetric loss functions associated with
predicted energy demand) can drive the Bayesian update, and thereby give
high weight to predictive distributions that yield small loss.

\baselineskip14pt

{\normalsize 
\bibliographystyle{Myapalike}
\bibliography{Project_bib_Gael}

\begin{thebibliography}{}

\bibitem[Aastveit\textit{ et~al.}, 2019]{aastveit2018evolution}
Aastveit, K.~A., Mitchell, J., Ravazzolo, F., and van Dijk, H.~K. (2019).
\newblock The evolution of forecast density combinations in economics.
\newblock {\em Oxford Research Encyclopedias: Economics and Finance}, 4:1--39.

\bibitem[Azzalini, 1985]{azzalini1985class}
Azzalini, A. (1985).
\newblock A class of distributions which includes the normal ones.
\newblock {\em Scandinavian Journal of Statistics}, 12(2):171--178.

\bibitem[Bassetti\textit{ et~al.}, 2018]{Bassetti2018}
Bassetti, F., Casarin, R., and Ravazzolo, F. (2018).
\newblock Bayesian nonparametric calibration and combination of predictive
  distributions.
\newblock {\em Journal of the American Statistical Association},
  113(522):675--685.

\bibitem[Baştürk\textit{ et~al.}, 2019]{BASTURK2019}
Baştürk, N., Borowska, A., Grassi, S., Hoogerheide, L., and van Dijk, H.
  (2019).
\newblock Forecast density combinations of dynamic models and data driven
  portfolio strategies.
\newblock {\em Journal of Econometrics}, 210(1):170--186.

\bibitem[Billio\textit{ et~al.}, 2013]{Billio2013}
Billio, M., Casarin, R., Ravazzolo, F., and van Dijk, H. (2013).
\newblock Time-varying combinations of predictive densities using nonlinear
  filtering.
\newblock {\em Journal of Econometrics}, 177(2):213--232.

\bibitem[Bissiri\textit{ et~al.}, 2016]{bissiri2016general}
Bissiri, P.~G., Holmes, C.~C., and Walker, S.~G. (2016).
\newblock A general framework for updating belief distributions.
\newblock {\em Journal of the Royal Statistical Society: Series B (Statistical
  Methodology)}, 78(5):1103--1130.

\bibitem[Blackwell and Dubins, 1962]{blackwell1962}
Blackwell, D. and Dubins, L. (1962).
\newblock Merging of opinions with increasing information.
\newblock {\em The Annals of Mathematical Statistics}, 33(3):882--886.

\bibitem[Casarin\textit{ et~al.}, 2019]{casarin2020}
Casarin, R., Grassi, S., Ravazzollo, F., and van Dijk, H. (2019).
\newblock Forecast density combinations with dynamic learning for large data
  sets in economics and finance.
\newblock {\em Tinbergen Institue Discussion Paper 2019-025/III}.

\bibitem[Casarin\textit{ et~al.}, 2015a]{casarin2015jss}
Casarin, R., Grassi, S., Ravazzolo, F., and van Dijk, H. (2015a).
\newblock Parallel sequential {M}onte {C}arlo for efficient density
  combination: The deco {MATLAB} toolbox.
\newblock {\em Journal of Statistical Software, Articles}, 68(3):1--30.

\bibitem[Casarin\textit{ et~al.}, 2015b]{casarin2015}
Casarin, R., Leisen, F., Molina, G., and ter Horst, E. (2015b).
\newblock A {B}ayesian beta {M}arkov random field calibration of the term
  structure of implied risk neutral densities.
\newblock {\em Bayesian Analysis}, 10(4):791--819.

\bibitem[Casarin\textit{ et~al.}, 2016]{casarin2016}
Casarin, R., Mantoan, G., and Ravazzolo, F. (2016).
\newblock Bayesian calibration of generalized pools of predictive
  distributions.
\newblock {\em Econometrics}, 4(1):1--24.

\bibitem[Chernozhukov and Hong, 2003]{CH03}
Chernozhukov, V. and Hong, H. (2003).
\newblock An {MCMC} approach to classical estimation.
\newblock {\em Journal of Econometrics}, 115(2):293--346.

\bibitem[Christoffersen, 1998]{christoffersen1998evaluating}
Christoffersen, P.~F. (1998).
\newblock Evaluating interval forecasts.
\newblock {\em International Economic Review}, 39(4):841--862.

\bibitem[Clements and Harvey, 2011]{Clements2011}
Clements, M. and Harvey, D. (2011).
\newblock Combining probability forecasts.
\newblock {\em International Journal of Forecasting}, 27(2):208--223.

\bibitem[Diks\textit{ et~al.}, 2011]{diks2011likelihood}
Diks, C., Panchenko, V., and Van~Dijk, D. (2011).
\newblock Likelihood-based scoring rules for comparing density forecasts in
  tails.
\newblock {\em Journal of Econometrics}, 163(2):215--230.

\bibitem[Doornik\textit{ et~al.}, 2020]{doornik2019card}
Doornik, J.~A., Castle, J.~L., and Hendry, D.~F. (2020).
\newblock Card forecasts for {M4}.
\newblock {\em International Journal of Forecasting}, 36(1):129--134.

\bibitem[Ehm\textit{ et~al.}, 2016]{ehm2016quantiles}
Ehm, W., Gneiting, T., Jordan, A., and Kr{\"u}ger, F. (2016).
\newblock Of quantiles and expectiles: consistent scoring functions, choquet
  representations and forecast rankings.
\newblock {\em Journal of the Royal Statistical Society: Series B (Statistical
  Methodology)}, 78(3):505--562.

\bibitem[Elliott and Timmermann, 2008]{elliott2008economic}
Elliott, G. and Timmermann, A. (2008).
\newblock Economic forecasting.
\newblock {\em Journal of Economic Literature}, 46(1):3--56.

\bibitem[Embrechts\textit{ et~al.}, 1997]{embrechts1997}
Embrechts, P., Kluppelberg, C., and Mikosch, T. (1997).
\newblock {\em Modelling Extremal Events for Insurance and Finance}.
\newblock Springer.

\bibitem[Fissler and Ziegel, 2016]{fissler2016higher}
Fissler, T. and Ziegel, J.~F. (2016).
\newblock Higher order elicitability and {O}sband’s principle.
\newblock {\em The Annals of Statistics}, 44(4):1680--1707.

\bibitem[Frazier\textit{ et~al.}, 2019]{frazier2019approximate}
Frazier, D.~T., Maneesoonthorn, W., Martin, G.~M., and McCabe, B.~P. (2019).
\newblock Approximate {B}ayesian forecasting.
\newblock {\em International Journal of Forecasting}, 35(2):521--539.

\bibitem[Ganics, 2017]{Ganics2017}
Ganics, G. (2017).
\newblock Optimal density forecast combinations.
\newblock {\em Technical Report No. 1751, Bank of Spain}.

\bibitem[Geweke, 2005]{geweke2005}
Geweke, J. (2005).
\newblock {\em Contemporary Bayesian econometrics and statistics}.
\newblock John Wiley \& Sons.

\bibitem[Geweke and Amisano, 2011]{Geweke2011}
Geweke, J. and Amisano, G. (2011).
\newblock Optimal prediction pools.
\newblock {\em Journal of Econometrics}, 164(1):130--141.

\bibitem[Giummol{\`e}\textit{ et~al.}, 2017]{giummole2017objective}
Giummol{\`e}, F., Mameli, V., Ruli, E., and Ventura, L. (2017).
\newblock Objective {B}ayesian inference with proper scoring rules.
\newblock {\em TEST}, pages 1--28.

\bibitem[Gneiting\textit{ et~al.}, 2007]{gneiting2007probabilistic}
Gneiting, T., Balabdaoui, F., and Raftery, A.~E. (2007).
\newblock Probabilistic forecasts, calibration and sharpness.
\newblock {\em Journal of the Royal Statistical Society: Series B (Statistical
  Methodology)}, 69(2):243--268.

\bibitem[Gneiting and Raftery, 2007]{gneiting2007strictly}
Gneiting, T. and Raftery, A.~E. (2007).
\newblock Strictly proper scoring rules, prediction, and estimation.
\newblock {\em Journal of the American Statistical Association},
  102(477):359--378.

\bibitem[Gneiting\textit{ et~al.}, 2005]{gneiting2005calibrated}
Gneiting, T., Raftery, A.~E., Westveld~III, A.~H., and Goldman, T. (2005).
\newblock Calibrated probabilistic forecasting using ensemble model output
  statistics and minimum {CRPS} estimation.
\newblock {\em Monthly Weather Review}, 133(5):1098--1118.

\bibitem[Gneiting and Ranjan, 2011]{gneiting2011comparing}
Gneiting, T. and Ranjan, R. (2011).
\newblock Comparing density forecasts using threshold-and quantile-weighted
  scoring rules.
\newblock {\em Journal of Business \& Economic Statistics}, 29(3):411--422.

\bibitem[Gneiting and Ranjan, 2013]{gneiting2013}
Gneiting, T. and Ranjan, R. (2013).
\newblock Combining predictive distributions.
\newblock {\em Electronic Journal of Statistics}, 7:1747--1782.

\bibitem[Grünwald and van Ommen, 2017]{grunwald2017}
Grünwald, P. and van Ommen, T. (2017).
\newblock Inconsistency of {B}ayesian inference for misspecified linear models,
  and a proposal for repairing it.
\newblock {\em Bayesian Analysis}, 12(4):1069--1103.

\bibitem[Guedj, 2019]{guedj2019}
Guedj, B. (2019).
\newblock A primer on {PAC-B}ayesian learning.
\newblock {\em arXiv preprint arXiv:1901.05353}.

\bibitem[Hall and Mitchell, 2007]{Hall2007}
Hall, S. and Mitchell, J. (2007).
\newblock Combining density forecasts.
\newblock {\em International Journal of Forecasting}, 23(1):1--13.

\bibitem[Holmes and Walker, 2017]{holmes2017assigning}
Holmes, C. and Walker, S. (2017).
\newblock Assigning a value to a power likelihood in a general {B}ayesian
  model.
\newblock {\em Biometrika}, 104(2):497--503.

\bibitem[Hyndman\textit{ et~al.}, 2002]{hyndman2002state}
Hyndman, R.~J., Koehler, A.~B., Snyder, R.~D., and Grose, S. (2002).
\newblock A state space framework for automatic forecasting using exponential
  smoothing methods.
\newblock {\em International Journal of forecasting}, 18(3):439--454.

\bibitem[Jiang and Tanner, 2008]{jiang2008}
Jiang, W. and Tanner, M.~A. (2008).
\newblock Gibbs posterior for variable selection in high-dimensional
  classification and data-mining.
\newblock {\em Annals of Statistics}, 36(5):2207--2231.

\bibitem[Johnson and West, 2018]{Johnson2018}
Johnson, M. and West, M. (2018).
\newblock Bayesian predictive synthesis: Forecast calibration and combination.
\newblock {\em arXiv preprint arXiv:1803.01984v2}.

\bibitem[Jore\textit{ et~al.}, 2010]{Jore2010}
Jore, A.~S., Mitchell, J., and Vahey, S.~P. (2010).
\newblock Combining forecast densities from {VAR}s with uncertain
  instabilities.
\newblock {\em Journal of Applied Econometrics}, 25(4):621--634.

\bibitem[Kapetanios\textit{ et~al.}, 2015]{Kap2015}
Kapetanios, G., Mitchell, J., Price, S., and Fawcett, N. (2015).
\newblock Generalised density forecast combinations.
\newblock {\em Journal of Econometrics}, 188(1):150--165.

\bibitem[Kascha and Ravazzolo, 2010]{Kascha2010}
Kascha, C. and Ravazzolo, F. (2010).
\newblock Combining inflation density forecasts.
\newblock {\em Journal of Forecasting}, 29(1‐2):231--250.

\bibitem[Knoblauch\textit{ et~al.}, 2019]{GVI2019}
Knoblauch, J., Jewson, J., and Damoulas, T. (2019).
\newblock Generalized variational inference: Three arguments for deriving new
  posteriors.
\newblock {\em arXiv preprint arXiv:1904.02063}.

\bibitem[Kr{\"u}ger and Ziegel, 2020]{kruger2020generic}
Kr{\"u}ger, F. and Ziegel, J.~F. (2020).
\newblock Generic conditions for forecast dominance.
\newblock {\em Journal of Business \& Economic Statistics}, pages 1--12.

\bibitem[Kunsch, 1989]{kunsch1989jackknife}
Kunsch, H.~R. (1989).
\newblock The jackknife and the bootstrap for general stationary observations.
\newblock {\em The Annals of Statistics}, 17(3):1217--1241.

\bibitem[Lee and Hansen, 1994]{lee1994asymptotic}
Lee, S.-W. and Hansen, B.~E. (1994).
\newblock Asymptotic theory for the {GARCH}(1,1) quasi-maximum likelihood
  estimator.
\newblock {\em Econometric Theory}, 10(1):29--52.

\bibitem[Lyddon\textit{ et~al.}, 2019]{lyddon2019general}
Lyddon, S., Holmes, C., and Walker, S. (2019).
\newblock General {B}ayesian updating and the loss-likelihood bootstrap.
\newblock {\em Biometrika}, 106(2):465--478.

\bibitem[Miller and Dunson, 2019]{miller2019robust}
Miller, J.~W. and Dunson, D.~B. (2019).
\newblock Robust {B}ayesian inference via coarsening.
\newblock {\em Journal of the American Statistical Association},
  114(527):1113--1125.

\bibitem[Montero-Manso\textit{ et~al.}, 2020]{montero2019fforma}
Montero-Manso, P., Athanasopoulos, G., Hyndman, R.~J., and Talagala, T.~S.
  (2020).
\newblock {FFORMA: F}eature-based forecast model averaging.
\newblock {\em International Journal of Forecasting}, 36(1):86--92.

\bibitem[Opschoor\textit{ et~al.}, 2017]{opschoor2017combining}
Opschoor, A., Van~Dijk, D., and van~der Wel, M. (2017).
\newblock Combining density forecasts using focused scoring rules.
\newblock {\em Journal of Applied Econometrics}, 32(7):1298--1313.

\bibitem[Patton, 2019]{patton2019comparing}
Patton, A.~J. (2019).
\newblock Comparing possibly misspecified forecasts.
\newblock {\em Journal of Business \& Economic Statistics}, pages 1--43.

\bibitem[Pesaran and Skouras, 2004]{pesaran2004decision}
Pesaran, M.~H. and Skouras, S. (2004).
\newblock Decision-based methods for forecast evaluation.
\newblock {\em A Companion to Economic Forecasting}, pages 241--267.

\bibitem[Pettenuzzo and Ravazzolo, 2016]{Pett2016}
Pettenuzzo, D. and Ravazzolo, F. (2016).
\newblock Optimal portfolio choice under decision-based model combinations.
\newblock {\em Journal of Applied Econometrics}, 31(7):1312--1332.

\bibitem[Price\textit{ et~al.}, 2018]{price2018bayesian}
Price, L.~F., Drovandi, C.~C., Lee, A., and Nott, D.~J. (2018).
\newblock Bayesian synthetic likelihood.
\newblock {\em Journal of Computational and Graphical Statistics}, 27(1):1--11.

\bibitem[Ranjan and Gneiting, 2010]{ranjan2010}
Ranjan, R. and Gneiting, T. (2010).
\newblock Combining probability forecasts.
\newblock {\em Journal of the Royal Statistical Society: Series B (Statistical
  Methodology)}, 72(1):71--91.

\bibitem[Roberts and Rosenthal, 2009]{roberts2009examples}
Roberts, G.~O. and Rosenthal, J.~S. (2009).
\newblock Examples of adaptive {MCMC}.
\newblock {\em Journal of Computational and Graphical Statistics},
  18(2):349--367.

\bibitem[Sisson\textit{ et~al.}, 2018]{sisson2018overview}
Sisson, S., Fan, Y., and Beaumont, M. (2018).
\newblock Overview of {ABC}.
\newblock {\em Handbook of Approximate {B}ayesian Computation}, pages 3--54.

\bibitem[Smith, 2015]{smith2015copula}
Smith, M.~S. (2015).
\newblock Copula modelling of dependence in multivariate time series.
\newblock {\em International Journal of Forecasting}, 31(3):815--833.

\bibitem[Smith and Maneesoonthorn, 2018]{smith2018inversion}
Smith, M.~S. and Maneesoonthorn, W. (2018).
\newblock Inversion copulas from nonlinear state space models with an
  application to inflation forecasting.
\newblock {\em International Journal of Forecasting}, 34(3):389--407.

\bibitem[Smyl, 2020]{smyl2019hybrid}
Smyl, S. (2020).
\newblock A hybrid method of exponential smoothing and recurrent neural
  networks for time series forecasting.
\newblock {\em International Journal of Forecasting}, 36(1):75--85.

\bibitem[Syring and Martin, 2019]{Syring2019}
Syring, N. and Martin, R. (2019).
\newblock Calibrating general posterior credible regions.
\newblock {\em Biometrika}, 106(2):479--486.

\bibitem[Van~der Vaart, 1998]{van1998}
Van~der Vaart, A.~W. (1998).
\newblock {\em Asymptotic Statistics}.
\newblock Cambridge University Press.

\bibitem[Zhang, 2006a]{Zhang2006a}
Zhang, T. (2006a).
\newblock From eps-entropy to kl entropy: analysis of minimum information
  complexity density estimation.
\newblock {\em Annals of Statistics}, 34:2180–2210.

\bibitem[Zhang, 2006b]{Zhang2006b}
Zhang, T. (2006b).
\newblock Information-theoretic upper and lower bounds for statistical
  estimation.
\newblock {\em IEEE Transactions on Information Theory}, 52(4):1307–1321.

\bibitem[Ziegel\textit{ et~al.}, 2020]{ziegel2020robust}
Ziegel, J.~F., Kr{\"u}ger, F., Jordan, A., and Fasciati, F. (2020).
\newblock Robust forecast evaluation of expected shortfall.
\newblock {\em Journal of Financial Econometrics}, 18(1):95--120.

\end{thebibliography}
}

\appendix

\baselineskip14pt

\section{Assumptions and Proofs of Main Results}

\label{app:A}

{\normalsize Consider a stochastic process $\{y_{t}:\Omega \rightarrow 
\mathbb{R},t\in \mathbb{N}\}$ defined on the complete probability space $%
(\Omega ,\mathcal{F},P)$. Let $\mathcal{F}_{t}:=\sigma (y_{1},\dots ,y_{t})$
denote the natural sigma-field. The model predictive class is given by $%
\mathcal{P}^{t}:=\left\{ P_{\boldsymbol{\theta }}^{t}:\boldsymbol{\theta }%
\in \Theta \right\} $, where the parameter space is a (possibly non-compact)
subset of the Euclidean space $\Theta \subseteq \mathbb{R}^{d_{\theta }}$, }%
with $d_{\theta }$ the dimension of $\boldsymbol{\theta }$, {\normalsize and
where for any $A\subset \mathcal{F}$, $P_{\boldsymbol{\theta }}^{t}(A):=P(A|%
\mathcal{F}_{t},\boldsymbol{\theta })$.}

{\normalsize Define the $\delta $-neighborhood of $\Theta $, around the
point $\boldsymbol{\theta }_{\ast }$, as $\mathcal{N}_{\delta }(\boldsymbol{%
\theta }_{\ast }):=\{\boldsymbol{\ \theta }:\Vert \boldsymbol{\theta }-%
\boldsymbol{\theta }_{\ast }\Vert \leq \delta \},$ where $\Vert \cdot \Vert $
denotes the Euclidean norm. We say two sequences $x_{n}$ and $z_{n}$ satisfy 
$x_{n}\lesssim z_{n}$ if there is a constant $C>0$ and an $n^{\prime }$ such
that for all $n>n^{\prime }$, $x_{n^{}}\leq Cz_{n}$. Throughout the
remainder, we let $C$ denote a generic constant that can change from line to
line. }

With a slight change in notation, {\normalsize recall the empirical scoring
rule calculated from $\mathbf{y}_{n}$: 
\begin{equation*}
S_{n}(\boldsymbol{\theta }):=\sum_{t=0}^{n-1}S(P_{\boldsymbol{\theta }%
}^{t},y_{t+1}),
\end{equation*}%
}and recall the limit optimizer,{\normalsize \ 
\begin{equation*}
\boldsymbol{\theta }_{\ast }:=\arg \max_{P_{\boldsymbol{\theta }}^{t}\in 
\mathcal{P}}\left[ \plim_{n\rightarrow \infty }\sum_{t=0}^{n-1}S(P_{%
\boldsymbol{\theta }}^{t},y_{t+1})/v_{n}^{2}\right] ,
\end{equation*}%
}for $v_{n}$ a positive sequence diverging to $\infty $ as $n\rightarrow
\infty ${\normalsize \textbf{.} The value $\boldsymbol{\theta }_{\ast }$
corresponds to the predictive $P_{\boldsymbol{\theta }}^{t}\in \mathcal{P}%
^{t}$ such that, in the limit, we minimize the expected }(i.e., asymptotic) 
{\normalsize loss, in terms of $S(\cdot ,y)$, between the assumed predictive
class, $\mathcal{P}$, and the true DGP }{that underlies the observed values
of }$y_{t+1}.$

{\normalsize We consider the following regularity conditions on the function 
$S_{n}$ and the prior $\pi (\cdot )$.}

\begin{assumption}
{\normalsize \label{ass:post} The prior density $\pi(\boldsymbol{\theta})$
is continuous on $\|\boldsymbol{\theta}-\boldsymbol{\theta}_*\|\leq\delta$,
for some $\delta>0$, and positive for all $\boldsymbol{\theta}\in\Theta$. In
addition, for any $t>0$, and any $\kappa$, $0<\kappa<\infty$, there exists
some $C>0$, and $p> \kappa$ such that 
\begin{equation*}
\Pi\left(\|\boldsymbol{\theta}-\boldsymbol{\theta}_*\|> t\right)\leq C
t^{-p}.
\end{equation*}
}
\end{assumption}

\begin{assumption}
{\normalsize \label{ass:ident} There exists a sequence $v_n$ diverging to $%
\infty$ such that, for any $\delta>0$ there exists some $\epsilon>0$, 
\begin{equation*}
\lim\inf _{n \rightarrow \infty} \text{\emph{Pr}}\left\{\sup _{\boldsymbol{\
\theta}\in\mathcal{N}^c_\delta(\boldsymbol{\theta}_*)} \frac{1}{v_n^2}\left[
S_{n}(\boldsymbol{\theta})-S_{n}(\boldsymbol{\theta}_{*})\right]
\leq-\epsilon \right\}=1.
\end{equation*}
}
\end{assumption}

\begin{assumption}
{\normalsize \label{ass:expand} For some $\delta>0$, the following are
satisfied uniformly for $\boldsymbol{\theta}\in\mathcal{N}_\delta( 
\boldsymbol{\theta}_*)$: There exists a sequence $v_n$ diverging to $\infty$%
, a vector function $\boldsymbol{\Delta}_n(\boldsymbol{\theta})$, matrices $%
H_n$ and $V_n$, such that }

\begin{enumerate}
\item {\normalsize $S_n(\boldsymbol{\theta})-S_n(\boldsymbol{\theta}
_*)=v_n\left(\boldsymbol{\theta}-\boldsymbol{\theta}_{*}\right)^{\prime} 
\boldsymbol{\Delta}_{n}\left(\boldsymbol{\theta}_{*}\right)/v_n-\frac{1}{2}
v_n\left(\boldsymbol{\theta}-\boldsymbol{\theta}_{*}\right)^{\prime}
H_{n}v_n\left(\boldsymbol{\theta}- \boldsymbol{\theta}_{*}\right)+R_{n}(%
\boldsymbol{\theta})$ }

\item {\normalsize $V_{n}^{-1/2}\boldsymbol{\Delta} _{n}\left(\boldsymbol{%
\theta}_{*}\right)/v_n\Rightarrow N(\boldsymbol{0},I)$ under $P$. }

\item {\normalsize For $n$ large enough, $H_n^{}$ is positive-definite and,
for some matrix $H$, $H_n\rightarrow_p H$. }

\item {\normalsize For any $\epsilon>0$, there exists $M\rightarrow\infty $
and $\delta=o(1)$ such that 
\begin{align*}
\lim\sup_{n \rightarrow \infty}& \text{\emph{Pr}}\left[\sup _{\frac{M}{v_n}
\le \|\boldsymbol{\theta}-\boldsymbol{\theta}_*\|\le\delta}\frac{|R_n( 
\boldsymbol{\theta})|}{v_n^2\|\boldsymbol{\theta}-\boldsymbol{\theta}_*\|^2}
>\epsilon \right]<\epsilon\text{ and } \lim\sup_{n \rightarrow \infty} & 
\text{\emph{Pr}}\left[\sup _{ \|\boldsymbol{\theta}-\boldsymbol{\theta}
_*\|\le \frac{M}{v_n}}\frac{|R_n(\boldsymbol{\theta})|}{v_n^2\|\boldsymbol{\
\theta}-\boldsymbol{\theta}_*\|^2}>\epsilon \right]=0
\end{align*}
}
\end{enumerate}
\end{assumption}

\begin{remark}
The use of $v_{n}$ in the assumptions allows us to capture cases where the
scoring rules may converge at rates other than the canonical $\sqrt{n}$.
Such instances include situations where the data are trended or otherwise
non-stationary.
\end{remark}

\begin{remark}
Assumption \ref{ass:post} places regularity conditions on the prior used
within the analysis and is standard in the literature on Bayesian inference.
Assumptions \ref{ass:ident} and \ref{ass:expand} are sufficient conditions
that can be used to deduce $v_{n}$-consistency and asymptotic normality of
the M-estimator based on minimizing $S_{n}(\boldsymbol{\theta })$, with
these assumptions being similar to those used in \cite{CH03}. These
assumptions are reasonable in any setting where the resulting M-estimator is
expected to display regular asymptotic behavior. In the case of log-score,
these assumptions are akin to the usual sufficient conditions used to
demonstrate consistency and asymptotic normality of the quasi-maximum
likelihood estimator. Indeed, in the context of the log-score, Assumptions %
\ref{ass:ident} and \ref{ass:expand} can be directly verified, for example,
for the ARCH(1) and GARCH(1,1) models used in Sections 4 and 5.1 using the
results of \cite{lee1994asymptotic}. However, we remark that the validation
of Assumptions \ref{ass:ident} and \ref{ass:expand} in the context of other
models, and other scoring rules, such as the censored score, CRPS or MSIS,%
\textbf{\ }entails substantive technical challenges that are best analyzed
in future work.
\end{remark}

\subsection{\protect\normalsize Posterior concentration}

{Consider}{\normalsize \ the FBP posterior }pdf 
\begin{equation}
\pi _{w}\left( \boldsymbol{\theta }|\mathbf{y}_{n}\right) =\frac{\exp \left[
w_{n}S_{n}({\boldsymbol{\theta }})\right] \pi (\boldsymbol{\theta })}{%
\int_{\Theta }\exp \left[ w_{n}S_{n}({\boldsymbol{\theta }})\right] \pi (%
\boldsymbol{\theta })\text{d}\boldsymbol{\theta }}.  \label{post_pdf}
\end{equation}%
{\normalsize Under Assumptions \ref{ass:post}-\ref{ass:expand}, we prove a
generalization of the result in Lemma \ref{lem:two}, which demonstrates that
the scaled posterior ${\tilde{\pi}}_{w}\left( {\boldsymbol{\eta }}|\mathbf{y}%
_{n}\right) :=\pi _{w}(\boldsymbol{\theta }|\mathbf{y}_{n})/v_{n}$, {where }}%
$\pi _{w}(\boldsymbol{\theta }|\mathbf{y}_{n})$ {is defined in (\ref%
{post_pdf})}, {\normalsize converges to a Gaussian version in the total
variation of moments norm. The result of Lemma \ref{lem:one} then follows by
taking $v_{n}=\sqrt{n}$ and considering the standard definition of $S_{n}(%
\boldsymbol{\theta })$. }

\begin{proposition}
{\normalsize \label{prop:bvm} For 
\begin{equation*}
\boldsymbol{\eta}:=v_n\left(\boldsymbol{\theta}-\boldsymbol{\theta}
_{*}\right)-H_n^{-1}\boldsymbol{\Delta} _{n}\left(\boldsymbol{\theta}%
_{*}\right)/v_n,
\end{equation*}
and 
\begin{equation*}
\phi\left(\boldsymbol{\eta}|H_n\right):= Q_n\exp\left[-\frac{C}{2} 
\boldsymbol{\eta}^{\prime }H_n\boldsymbol{\eta}\right]
\end{equation*}
with $Q_n:=\sqrt{C\text{det}\left[H_n\right] (2\pi)^{-d_\theta}}$, if
Assumptions \ref{ass:post}-\ref{ass:expand} are satisfied and $\lim_n w_n=C$%
, $0<C<\infty$, then 
\begin{equation*}
\int\left[1+\|\boldsymbol{\eta}\|^\kappa\right]\left|\left.\tilde{\pi}_w( 
\boldsymbol{\eta}|\mathbf{y}_n)-\phi\left(\boldsymbol{\eta}
|H_n\right)\right.\right| \text{d}\boldsymbol{\eta}=o_{p}(1).
\end{equation*}
}
\end{proposition}

{\normalsize 
\begin{proof}
Define 
	$$\boldsymbol{L}_n:=\boldsymbol{\theta}_*+H_n^{-1}\boldsymbol{\Delta}_n(\boldsymbol{\theta}^*)/v_n^2
	$$ and note that by a change of variables $\boldsymbol{\theta}\mapsto\boldsymbol{\eta}$, we obtain the posterior
	$$\tilde{\pi}[\boldsymbol{\eta}|\mathbf { y } _ { n }]:=\frac{\pi\left(\boldsymbol{\eta}/v_n+\boldsymbol{L}_n\right)\exp\left[w_{n}S_n(\boldsymbol{\eta}/v_n+\boldsymbol{L}_n)\right]}{\int\pi\left(\boldsymbol{u}/v_n+\boldsymbol{L}_n\right)\exp\left[w_{n}S_n(\boldsymbol{u}/v_n+\boldsymbol{L}_n)\right]\dx \boldsymbol{u}}.$$ 
	Define 
	$$\gamma(\boldsymbol{\eta}):=S_n\left(
	\boldsymbol{\eta}/v_n+\boldsymbol{L}_n\right)-S_n(\boldsymbol{\theta}_*)-\frac{1}{2v_n^2}\Delta_n(\boldsymbol{\theta}_*)'H^{-1}_n\Delta_n(\boldsymbol{\theta}_*)$$ and note that 
	$$
	\tilde{\pi}[\boldsymbol{\eta}|\mathbf { y } _ { n }]:=\frac{\pi\left(\boldsymbol{\eta}/v_n+\boldsymbol{L}_n\right)\exp\left(w_{n} \gamma(\boldsymbol{\eta})\right)}{C_n},
	$$for $$C_n=\int \pi\left(\boldsymbol{u}/v_n+\boldsymbol{L}_ { n }\right)\exp(w_{n}\cdot\gamma(\boldsymbol{u}))\dx \boldsymbol{u}.$$

{\normalsize Note that we can write 
\begin{align*}
\int \|\boldsymbol{\eta}\|^\kappa\left|\tilde{\pi}\right(\boldsymbol{\eta}| 
\mathbf{\ y } _ { n }\left)-\phi\left(\boldsymbol{\eta}|H_n\right)\right|\dx 
\boldsymbol{\eta}&= \frac{1}{C_n} \int\|\boldsymbol{\eta}\|^\kappa {%
| }e^{w_n\gamma_{}(\boldsymbol{\eta})} \pi\left(\boldsymbol{\eta%
}/v_n+ \boldsymbol{L}_n\right)-\phi\left(\boldsymbol{\eta}|H_n\right)C_n {|}%
\dx \boldsymbol{\eta} \\
&={J_{n}}/{C_n},
\end{align*}
where 
\begin{align}
J_{n}:&=\int\|\boldsymbol{\eta}\|^\kappa\left|e^{w_n\gamma(\boldsymbol{\eta}
)} \pi_{}\left(\boldsymbol{\eta}/v_n+\boldsymbol{L}_n\right)-\phi\left( 
\boldsymbol{\eta}|H_n\right)C_n\right| \dx\boldsymbol{\eta}.
\end{align}
Now, bound $J_{n}$ as follows
\begin{align}
J_{n}&\le\int\|\boldsymbol{\eta}\|^\kappa\left\{\left|e^{w_n\gamma( 
	\boldsymbol{\eta})} \pi_{}\left(\boldsymbol{\eta}/v_n+\boldsymbol{L}_n\right)-\pi\left(\boldsymbol{\theta}_*\right)\phi\left(\boldsymbol{\eta}
|H_n\right)\textcolor{black}{\text{$Q_n^{-1}$}}\right|+\left|\left[\pi\left(\boldsymbol{\theta}_*\right)\textcolor{black}{\text{$Q_n^{-1}$}}-{%
	C }_{n}^{}\right]\phi\left(\boldsymbol{\eta}|H_n\right)\right|\right\} \dx 
\boldsymbol{\eta},  \notag\\&\leq J_{1n}+J_{2n},\notag
\end{align}
where 
\begin{align}
J_{1n}&:= \int\|\boldsymbol{\eta}\|^\kappa\left|e^{w_n\gamma(\boldsymbol{%
\eta })} \pi_{}\left(\boldsymbol{\eta}/v_n+\boldsymbol{L}_n\right)-
\pi_{}\left(\boldsymbol{\theta}_*\right)\phi\left(\boldsymbol{\eta}
|H_n\right)\textcolor{black}{\text{$Q_n^{-1}$}}\right| \dx\boldsymbol{\eta}  \label{eq:post_norm1} \\
J_{2n}&:=\left|\pi(\boldsymbol{\theta}_*)\textcolor{black}{\text{$Q_n^{-1}$}}-C_n\right|\int \| 
\boldsymbol{\eta}\|^\kappa\phi\left(\boldsymbol{\eta}|H_n\right) \dx 
\boldsymbol{\eta}  \label{eq:post_norm2}
\end{align}
}

{\normalsize The result follows if we can prove that $%
J_{1n}=o_{p}(1)$ since, taking $\kappa=0$, $J_{1n}=o_{p}(1)$ implies that 
\begin{align*}
\left|C_n-\pi(\boldsymbol{\theta}_*)\textcolor{black}{\text{$Q_n^{-1}$}} \right|&=\left|\int
e^{w_n\gamma(\boldsymbol{\eta})} \pi_{}\left(\boldsymbol{\eta}/v_n+ 
\boldsymbol{L}_n\right) \dx\boldsymbol{\eta}-\pi(\boldsymbol{\theta}_*)\textcolor{black}{\text{$Q_n^{-1}$}}\int
\phi\left(\boldsymbol{\eta}|H_n\right) \dx\boldsymbol{\eta}\right| \\
&=o_{p}(1),
\end{align*}
which implies that $J_{2n}=o_{p}(1)$ and therefore $J_{n}=o_{p}(1)$. 
}

{\normalsize Using Assumption \ref{ass:expand}.1, the
fact that $(\boldsymbol{\theta}-\boldsymbol{\theta}_{*})=\left( 
\boldsymbol{\eta}-\boldsymbol{\eta}_n^{*}\right)/v_n$ with $\boldsymbol{\eta}
_n^{*}:=-H_n^{-1}\boldsymbol{\Delta}_n/v_n$, and the definitions
of $\boldsymbol{L}_n$ and $\gamma(\boldsymbol{\eta})$, we re-express $\gamma( 
\boldsymbol{\eta})$ as 
\begin{align*}
\gamma(\boldsymbol{\eta})&=S_n\left(\boldsymbol{\theta}\right)- S_n( 
\boldsymbol{\theta}_*)-\frac{1}{2v_n^2}\boldsymbol{\Delta}_n(\boldsymbol{\
\theta}_*)^{\prime}\textcolor{black}{\text{$H_n^{-1}$}}\boldsymbol{\Delta}_n( 
\boldsymbol{\theta}_*) \\
&= \frac{\left(\boldsymbol{\eta}-\boldsymbol{\eta}_n^{*}\right)}{v_n}
^{\prime }\boldsymbol{\Delta}_n(\boldsymbol{\theta}_*)-\frac{v^2_n}{2}\frac{
\left(\boldsymbol{\eta}-\boldsymbol{\eta}_n^{*}\right)}{v_n}^{\prime }H_n 
\frac{\left(\boldsymbol{\eta}-\boldsymbol{\eta}_n^{*}\right)}{v_n}-\frac{1}{
2v_n^2}\boldsymbol{\Delta}_n(\boldsymbol{\theta}_*)^{\prime }\textcolor{black}{\text{$H_n^{-1}$}}\boldsymbol{%
\ \Delta}_n(\boldsymbol{\theta}_*)+R_n\left(\boldsymbol{\eta}/v_n+ 
\boldsymbol{L}_n\right) \\
&=-\frac{1}{2}{\boldsymbol{\eta}^{\prime }}H_n{\boldsymbol{\eta}}+R_n\left( 
\boldsymbol{\eta}/v_n+\boldsymbol{L}_n\right).
\end{align*}
}

{\normalsize To demonstrate that $J_{1n}=o_p(1)$, we split the integral for $J_{1n}$ into three regions: for $%
0<M<\infty$, and some $\delta>0$, }

\begin{enumerate}
\item {\normalsize $\|\boldsymbol{\eta}\|\leq M$ }

\item {\normalsize $M<\|\boldsymbol{\eta}\|\leq \delta v_n$ }

\item {\normalsize $\|\boldsymbol{\eta}\|\geq \delta v_n$. }
\end{enumerate}

{\normalsize \noindent\textbf{\textbf{Area 1}:} Over $\|\boldsymbol{\eta}
\|\leq M$, 
\begin{align*}
\sup_{\|\boldsymbol{\eta}\|\leq M}&\left|\pi\left(\boldsymbol{
\eta}/v_n+\boldsymbol{L}_n\right)-\pi(\boldsymbol{\theta}_*)\right|=o_{p}(1), \\
\sup_{\|\boldsymbol{\eta}\|\leq M}&\left|R_n\left(\boldsymbol{
\eta}/v_n+\boldsymbol{L}_n\right)\right|=o_{p}(1).
\end{align*}
The first result follows from continuity of $\pi(\cdot)$ in Assumption \ref%
{ass:post}, the convergence of $\boldsymbol{\Delta}_n(\boldsymbol{\theta}_*)$
in Assumption \ref{ass:expand}.2, and the definition of $\boldsymbol{L}_n$. The
second result follows from the second part of Assumption \ref{ass:expand}.4.
The dominated convergence theorem then allows us to deduce that $%
J_{1n}=o_{p}(1)$ over $\|\boldsymbol{\eta}\|\leq M$. \bigskip \newline
\noindent\textbf{\textbf{Area 2}:} Over $M\leq \|\boldsymbol{\eta}\| \leq
\delta v_{n}$, the second term in the integral can be {made arbitrarily}
small by taking $M$ large enough and $\delta=o(1)$. It therefore suffices to
show that, for $M$ large enough and $\delta$ small enough, 
\begin{equation*}
\tilde{J}_{1n}:=\int_{M\leq \|\boldsymbol{\eta}\|\le \delta v_n}\| 
\boldsymbol{\eta}\|^{\kappa}\exp\left[w_n\gamma(\boldsymbol{\eta})\right]
\pi\left(\boldsymbol{\eta}/v_n+\boldsymbol{L}_n\right)\dx\boldsymbol{\eta}
=o_p(1).
\end{equation*}
From the definition of $\gamma(\boldsymbol{\eta})$, it follows that 
\begin{equation*}
\exp\left[w_n\gamma(\boldsymbol{\eta})\right]\leq \exp\left[-\frac{w_n}{2} 
\boldsymbol{\eta}^{\prime }H_n \boldsymbol{\eta}+w_n\left|R_n\left(\boldsymbol{\eta}/v_n+\boldsymbol{L}_n\right)\right|\right].
\end{equation*}
By the first part of Assumption \ref{ass:expand}.4, on the
set $\{\boldsymbol{\eta}: M\leq\|\boldsymbol{\eta}\| \leq \delta v_n\}$, 
\begin{align*}
\left|R_{n}\left(\boldsymbol{\eta}/v_n+\boldsymbol{L}
_n\right)\right|&=o_p\left(v_n^2\|v_n^{-1}\left[\boldsymbol{\eta}- 
\boldsymbol{\eta}_n^*\right]\|^{2}\right)=o_p\left(\|\boldsymbol{\eta}+\frac{1}{v_n}H_n^{-1}\boldsymbol{\Delta}
_n\|^{2}\right).
\end{align*}
}

{\normalsize Since $\|\boldsymbol{\eta}_n^*\|^2=O_p(1)$, we conclude that,
for some $C>0$, on the set $A_{M,\delta} = \{\boldsymbol{\eta}: M\leq\| 
\boldsymbol{\eta}\| \leq \delta v_n\}$, 
\begin{equation*}
\exp\left[w_n\gamma(\boldsymbol{\eta})\right]\leq C\exp\left[-\frac{w_n}{2} 
\boldsymbol{\eta}^{\prime }H_n \boldsymbol{\eta}+o_p\left(w_n\|\boldsymbol{\
\eta}\|^2\right)\right].
\end{equation*}
We then have that, for some $M\rightarrow\infty$, $\tilde{J}_{1n}\leq
K_{1n}+K_{2n}+K_{3n}$, where 
\begin{align*}
K_{1n}:=&\int_{A_{M,\delta}}\|\boldsymbol{\eta}\|^\kappa\exp\left(-\frac{w_n}{
2}\boldsymbol{\eta}^{\prime }H_n\boldsymbol{\eta}\right)\sup _{\|\boldsymbol{%
\ \eta}\| \leq M}\left|\exp\left[o_p\left(w_n\|\boldsymbol{\eta}
\|^2\right)\right]\pi_{}\left(\boldsymbol{\eta}/v_n+ 
\boldsymbol{L}_n\right)-\pi_{}\left(\boldsymbol{\theta}_*\right)\right|\dx 
\boldsymbol{\eta} \\
K_{2n}:=&\int_{A_{M,\delta}}\|\boldsymbol{\eta}\|^\kappa\exp\left[-\frac{w_n}{
2}\boldsymbol{\eta}^{\prime }H_n \boldsymbol{\eta}+o_p\left(w_n\|\boldsymbol{%
\ \eta}\|^2\right)\right] \pi_{}\left(\boldsymbol{%
\eta}/v_n +\boldsymbol{L}_n\right)\dx\boldsymbol{\eta} \\
K_{3n}:=&\pi_{}\left(\boldsymbol{\theta}_*\right) \int_{A_{M,\delta}}\| 
\boldsymbol{\eta}\|^\kappa\exp\left[-\frac{w_n}{2}\boldsymbol{\eta}^{\prime
}H_n \boldsymbol{\eta}+o_p\left(w_n\|\boldsymbol{\eta}\|^2\right) %
\right]\dx\boldsymbol{\eta}
\end{align*}
For any fixed $M$, $K_{1n}=o_{p}(1)$, hence, for some sequence $%
M\rightarrow\infty$, by the dominated convergence theorem, it follows that $%
K_{1n}=o_{p}(1)$. For the second and third terms, we note the following: }

{\normalsize \noindent (i) For any $0<\kappa<\infty$, on the set $\{ 
\boldsymbol{\eta}:\|\boldsymbol{\eta}\|\ge M\}$, there exists some $%
M^{\prime }$ large enough such that for all $M>M^{\prime }$: 
\begin{equation*}
\|\boldsymbol{\eta}\|^\kappa\exp\left(-\boldsymbol{\eta}^{\intercal}H_n 
\boldsymbol{\eta}\right)=O(1/M).
\end{equation*}
}

{\normalsize \noindent (ii) From the definition of $\boldsymbol{L}_n$, 
$\boldsymbol{L}_n\rightarrow_{p}\boldsymbol{\theta}_*$ and by continuity of $%
\pi(\cdot)$, on the set $\{\boldsymbol{\eta}:M\le\|\boldsymbol{\eta}
\|\le \delta v_n\}$, for $\delta=o(1)$, we conclude that 
\begin{equation*}
\pi(\boldsymbol{L}_n+\boldsymbol{\eta}/v_n)=\pi(\boldsymbol{\theta}_*)+o_p(1).
\end{equation*}
}

{\normalsize Applying (i) yields $K_{3n}=o_{p}(1)$. Applying (i) and (ii)
together yields $K_{2n}=o_{p}(1)$. }

{\normalsize \medskip }

{\normalsize \noindent\textbf{\textbf{Area 3}:} Over $\|\boldsymbol{\eta}
\|\geq \delta v_n$. For $\delta v_n$ large 
\begin{equation*}
Q_n^{-1}\int_{\|\boldsymbol{\eta}\|\ge \delta v_n}\|\boldsymbol{\eta}\|^\kappa\phi( 
\boldsymbol{\eta}|H_n)\pi(\boldsymbol{\theta}_*)\dx\boldsymbol{\eta}=o(1).
\end{equation*}
Now, focus on 
\begin{equation*}
\tilde{J}_{1n}:=\int_{\|\boldsymbol{\eta}\|\ge \delta v_n}\|\boldsymbol{\eta}
\|^\kappa e^{w_n\gamma(\boldsymbol{\eta})} \pi\left(\boldsymbol{%
\ \eta}/v_n+L_{n}\right)\dx\boldsymbol{\eta}.
\end{equation*}
The change of variables $\boldsymbol{\theta}=\boldsymbol{\eta}/v_n+ 
\boldsymbol{L}_n$, then yields 
\begin{align*}
\tilde{J}_{1n}=v_{n}^{d_\theta+\kappa}\int_{\|\boldsymbol{\theta}-\boldsymbol{L}_n\|\ge \delta }\|\boldsymbol{\theta}-\boldsymbol{L}_n\|^{\kappa}e^{w_n\left[%
S_n( \boldsymbol{\theta})-S_n(\boldsymbol{\theta}_*)-\frac{1}{2v_n^2}%
\boldsymbol{\ \Delta}_n(\boldsymbol{\theta}_*)^{\prime }H_{n}^{-1}\boldsymbol{\Delta}_n(\boldsymbol{\theta}_*)\right]}
\pi\left(\boldsymbol{\theta}\right)\dx \boldsymbol{\theta}.
\end{align*}
Use the fact that $\boldsymbol{L}_n=\boldsymbol{\theta}_*+o_p(1)$, and bound $\tilde{J}_{1n}$ by
\begin{equation*}
Ce^{-\frac{w_n}{2v_n^2}\boldsymbol{\Delta}_n(\boldsymbol{\theta}_*)^{\prime
}H_{n}^{-1}\boldsymbol{\Delta}_n(\boldsymbol{\theta}
_*)}v_{n}^{d_\theta+\kappa}\int_{\|\boldsymbol{\theta}-\boldsymbol{\theta}
_*\|\ge \delta }\|\boldsymbol{\theta}-\boldsymbol{\theta}_*\|^{\kappa}e^{w_n %
\left[S_n(\boldsymbol{\theta})-S_n(\boldsymbol{\theta}_*)\right]} \pi\left( 
\boldsymbol{\theta}\right)\dx \boldsymbol{\theta},
\end{equation*}
}

{\normalsize From Assumption \ref{ass:expand}.2, $\frac{1}{2v_n^2} 
\boldsymbol{\Delta}_n(\boldsymbol{\theta}_*)^{\prime }H_{n}^{-1}\boldsymbol{\Delta}_n(\boldsymbol{\theta}_*)=O_{p}(1)$ so that $%
e^{-\frac{w_n}{2v_n^2}\boldsymbol{\Delta}_n(\boldsymbol{\theta}_*)^{\prime
}H_{n}^{-1}(\boldsymbol{\theta}_*)\boldsymbol{\Delta}_n(\boldsymbol{\theta}
_*)}=O_{p}(1)$. By Assumption \ref{ass:ident}, for any $\delta>0$, there
exists some $\epsilon>0$ such that 
\begin{equation*}
\lim_{n\rightarrow\infty}\text{Pr}\left\{\sup_{\|\boldsymbol{\theta}- 
\boldsymbol{\theta}_*\|\geq \delta}\left[S_n(\boldsymbol{\theta})-S_n( 
\boldsymbol{\theta}_*)\right]\leq -v_n^2\epsilon\right\}=1.
\end{equation*}
Applying the above conclusion to $\tilde{J}_{1n}$ then yields, for $w_n$ a
positive convergence sequence, for $n$ large enough, 
\begin{align*}
\tilde{J}_{1n}\lesssim e^{-\epsilon w_nv_n^2}v_{n}^{d_\theta+\kappa}\int_{\| 
\boldsymbol{\theta}-\boldsymbol{\theta}_*\|\ge \delta }\|\boldsymbol{\theta}%
- \boldsymbol{\theta}_*\|^{\kappa}\pi\left(\boldsymbol{\theta}\right)\dx 
\boldsymbol{\theta}.
\end{align*}
Apply the above bound, Assumption \ref{ass:post}, and the dominated
convergence theorem to deduce 
\begin{align*}
\tilde{J}_{1n}\lesssim e^{-\epsilon w_nv_n^2}v_{n}^{d_\theta+\kappa}\int_{\| 
\boldsymbol{\theta}-\boldsymbol{\theta}_*\|\ge \delta }\|\boldsymbol{\theta}%
- \boldsymbol{\theta}_*\|^{\kappa}\pi\left(\boldsymbol{\theta}\right)\dx 
\boldsymbol{\theta}&\lesssim e^{-\epsilon
w_nv_n^2}v_{n}^{d_\theta+\kappa}\int_{\|\boldsymbol{\theta}-\boldsymbol{\
\theta}_*\|\ge \delta }\|\boldsymbol{\theta}-\boldsymbol{\theta}
_*\|^{\kappa}\pi\left(\boldsymbol{\theta}\right)\dx \boldsymbol{\theta}
+o_{p}(1) \\
&\leq e^{-\epsilon w_nv_n^2}v_{n}^{d_\theta +\kappa}\int_{ }\|\boldsymbol{\
\theta}-\boldsymbol{\theta}_*\|^{\kappa}\pi\left(\boldsymbol{\theta}\right)\dx 
\boldsymbol{\theta}+o_{p}(1) \\
&\lesssim e^{-\epsilon w_nv_n^2}v_{n}^{d_\theta+\kappa}+o_{p}(1).
\end{align*}
}
\end{proof}}

\subsection{\protect\normalsize Bayesian and frequentist agreement}

{\normalsize We now prove the result of Lemma \ref{lem:two} and,
subsequently, Theorem \ref{thm2}. 
\begin{proof}[Proof of Lemma \ref{lem:two}]
We first prove that the sequence $\|\widehat{\boldsymbol{\theta}}-\boldsymbol{\theta}_*\|=O_{p}(v_n^{-1})$. By assumption,
$$
o_p(v_n^{-1})\leq S_n(\hat{\boldsymbol{\theta}})-S_n(\boldsymbol{\theta}_*).
$$Applying
 \ref{ass:expand}.1 to the above we obtain
\begin{equation}\label{eq:expands}
o_p(v_n^{-1})\leq v_{n}\left(\widehat{\boldsymbol{\theta}}-\boldsymbol{\theta}_{\cdot}\right)^{\prime} \boldsymbol{\Delta}_{n}\left(\boldsymbol{\theta}_{*}\right) / v_{n}-\frac{1}{2} v_{n}\left(\widehat{\boldsymbol{\theta}}-\boldsymbol{\theta}_{*}\right)^{\prime}H_{n}v_{n}\left(\widehat{\boldsymbol{\theta}}-\boldsymbol{\theta}_{*}\right)+R_{n}(\widehat{\boldsymbol{\theta}}).
\end{equation}
From the consistency of $\widehat{\boldsymbol{\theta}}$, for any $\epsilon>0$, the following holds with probability at least $1-\epsilon$: there exists some sequence $\delta_{\epsilon,n}\rightarrow0$, such that $\|\widehat{\boldsymbol{\theta}}-\boldsymbol{\theta}_*\|< \delta_{\epsilon,n}$; by Assumption \ref{ass:expand}.2, there exists some $K_{\epsilon,n}$ such that $\|\boldsymbol{\Delta}_n(\boldsymbol{\theta}_*)/v_n\|<K_{\epsilon,n}$; for $\delta_{\epsilon,n}$ as before and for $H$ as in the statement of the theorem, $\|H_n-H\|<\delta_{\epsilon,n}$. Define $\boldsymbol{h}_n:=v_n(\widehat{\boldsymbol{\theta}}-\boldsymbol{\theta}_*)$. There then exists a sequence $K_{\epsilon,n}^{*}$ such that if we apply the above inequalities to equation \eqref{eq:expands} we obtain
\begin{flalign*}
o_p(v_n^{-1})&\leq -\|H^{1/2}\boldsymbol{h}_n\|^2/2+K^*_{\epsilon,n}\left(\|\boldsymbol{h}_n\|+o_{p}(\|\boldsymbol{h}_n\|^2)\right),
\end{flalign*}where the last term follows by applying the
second part of Assumption \ref{ass:expand}.4 and {the consistency} of $\widehat{\boldsymbol{\theta}}$.
We can rearrange this equation to obtain
	\begin{flalign*}
o_p(v_n^{-1})+\|H^{1/2}\boldsymbol{h}_n\|^2/2-K^*_{\epsilon,n}\left(\|\boldsymbol{h}_n\|+o_{p}(\|\boldsymbol{h}_n\|^2)\right)\leq 0.
	\end{flalign*}The above implies that, for some $K^{**}_{n\epsilon}=O_p(1)$, $\|\boldsymbol{h}_n\|<K^{**}_{n\epsilon}$, which yields the result.

The result now follows along lines similar to Theorem 5.23 in \cite{van1998}. For $\boldsymbol{h}_n:=v_n(\hat{\boldsymbol{\theta}}-\boldsymbol{\theta}_*)$, apply Assumption \ref{ass:expand}.3 to obtain
\begin{flalign*}
S(\boldsymbol{\theta}_*+\boldsymbol{h}_n/v_n)-S_n(\boldsymbol{\theta}_*)&= \boldsymbol{h}_n^{\prime} \boldsymbol{\Delta}_{n}\left(\boldsymbol{\theta}_{*}\right) / v_{n}+\boldsymbol{h}_n^{\prime}\left[-H_{n}\right] \boldsymbol{h}_n/2+o_p(1),\\
S(\boldsymbol{\theta}_*+H_{}^{-1}\boldsymbol{\Delta}_n/v_n^2)-S_n(\boldsymbol{\theta}_*)&=- \frac{1}{2}\frac{1}{v_n^2}\boldsymbol{\Delta}_n'\left[-H_{n}^{-1}\right]\boldsymbol{\Delta}_{n}\left(\boldsymbol{\theta}_{*}\right) +o_p(1),
\end{flalign*}
where the last term in the first equation follows from the $v_n$-consistency of $\widehat{\boldsymbol{\theta}}$ and in the second instance by Assumption \ref{ass:expand}.2. By the definition of $\widehat{\boldsymbol{\theta}}$ the LHS of the first equation is {greater than}, up to an $o_p(1)$ term, the LHS of the second equation. Therefore, subtracting the second equation from the first, and completing the square we have 
$$
\left(\boldsymbol{h}_n+H^{-1}\boldsymbol{\Delta}_n/v_n\right)^{\prime}\left[-H_n\right]\left(\boldsymbol{h}_n+H^{-1}\boldsymbol{\Delta}_n/v_n\right) +o_p(1)\ge0.
$$Because $\left[-H_{n}^{}\right]$ converges in probability to a negative definite matrix, it must follow that $\|\boldsymbol{h}_n-H_{}^{-1}\boldsymbol{\Delta}_n/v_n\|=o_p(1)$. The result then follows from the asymptotic normality in Assumption \ref{ass:expand}.2.

\end{proof}}

{\normalsize 
\begin{proof}[Proof of Theorem \ref{thm2}] 
Let $\rho_{H}$ denote the Hellinger metric: for absolutely
continuous probability measures $P$ and $G$,
\[
\rho_{H}\{P,G\}=\left\{  \frac{1}{2}\int\left[  \sqrt{{\dx P}}-\sqrt{{\dx G}%
}\right]  ^{2}\dx\mu\right\}  ^{1/2},\quad0\leq\rho_{H}\{P,G\}\leq1,
\]
for $\mu$ the Lebesgue measure, and define $\rho_{TV}$ to be the total
variation metric,
\[
\rho_{TV}\{P_{{}},G_{{}}\}=\sup_{B\in\mathcal{F}}|P(B)-G(B)|,\quad0\leq
\rho_{TV}\{P,G\}\leq2.
\]
Recall that, according to the definition of merging in \cite{blackwell1962}, two predictive measures $P$ and $G$ are said
to merge if
\[
\rho_{TV}\{P,G\}=o_{p}(1).
\]

Recall the definitions of $P^n_w$, $P^n_*$ in equations \eqref{P_fore} and \eqref{P_true} in the main text. Likewise, consider the following frequentist version of the FBP predictive: Let $\Pi[\cdot|\y_n]:=\mathcal{N}(\hat{\boldsymbol{\theta}},H^{-1}V_{}H^{-1}/v_n^{2})$ denote a normal measure with mean $\hat{\boldsymbol{\theta}}$ and variance $H^{-1}V_{}H^{-1}/v_n^{2}$ and consider the frequentist equivalent of the FBP predictive:
\begin{equation}\label{eq:other}
P^n_{\hat{\boldsymbol{\theta}}}=\int_\Theta P^n_{\boldsymbol{\theta}}\dx\Pi[\boldsymbol{\theta}|\y_n].
\end{equation}

The result follows similar arguments to those given in the proof of Theorem 1 in \cite{frazier2019approximate}. Fix $\epsilon>0$ and define the set $V_{\epsilon}:=\{\boldsymbol{\theta}\in
\Theta:\rho^{2}_{H}\{P^n_*,P^n_{\boldsymbol{\theta}}\}>\epsilon/4\}$. By convexity of $\rho^{2}_{H}\{P^n_*,\cdot\}$,
and Jensen's inequality, \begin{flalign*}
\rho^2_{H}\{P^n_*,P^n_w\}&\leq \int_{\Theta}\rho^2_{H}\{P^n_*,P^n_{\boldsymbol{\theta}}\}\dx\Pi_w[\boldsymbol{\theta}|\y_n]\\&= \int_{V_{\epsilon}}\rho^2_{H}\{P^n_*,P^n_{\boldsymbol{\theta}}\}\dx\Pi_w[\boldsymbol{\theta}|\y_n]+\int_{V_{\epsilon}^{c}}\rho^2_{H}\{P^n_*,P^n_{\boldsymbol{\theta}}\}\dx\Pi_w[\boldsymbol{\theta}|\y_n]\\&= \Pi_w[V_{\epsilon}|\y_n]+\frac{\epsilon}{4}\Pi_w[V_{\epsilon}^{c}|\y_n].
\end{flalign*}By definition, $\boldsymbol{\theta}_{*}\notin V_{\epsilon}$ and
therefore, by the posterior concentration of $\Pi_w[\cdot|\y_n]$ in Lemma \ref{lem:two}, $\Pi_w\lbrack V_{\epsilon
}|\mathbf{y}_n]=o_{p}(1)$. Hence, we can conclude:
\begin{equation}
\rho^{2}_{H}\{P^n_*,P^n_w\}\leq o_{p}
(1)+\frac{\epsilon}{4}. \label{hell1}%
\end{equation}
Likewise, for $\Pi[\cdot|\y_n]$ defined in equation \eqref{eq:other}, a similar argument yields
\begin{flalign}
\rho^2_{H}\{P^n_*,P^n_{\hat{\boldsymbol{\theta}}}\}&\leq \int_{\Theta}\rho^2_{H}\{P^n_*,P^n_{\boldsymbol{\theta}}\}\dx\Pi[\boldsymbol{\theta}|\y_n]\nonumber\\&= \int_{V_{\epsilon}}\rho^2_{H}\{P^n_*,P^n_{\boldsymbol{\theta}}\}\dx\Pi[\boldsymbol{\theta}|\y_n]+\int_{V_{\epsilon}^{c}}\rho^2_{H}\{P^n_*,P^n_{\boldsymbol{\theta}}\}\dx\Pi[\boldsymbol{\theta}|\y_n]\nonumber\\&= \Pi[V_{\epsilon}|\y_n]+\frac{\epsilon}{4}\Pi[V_{\epsilon}^{c}|\y_n]\nonumber\\&= o_{p}(1)+\frac
{\epsilon}{4}, \label{hell2}
\end{flalign}where the concentration in \eqref{hell2} follows as a result of Lemma \ref{lem:two} and the definition of $\Pi[\cdot|\y_n]$. Now, note that \begin{flalign*}
\frac{1}{2} \left[\rho^2_{H}\{P^n_*,P^n_w\}+ \rho^2_{H}\{P^n_*,P^n_{\hat{{\boldsymbol{\theta}}}}\}\right]&\geq \frac{1}{4}\left[\rho_{H}\{P^n_*,P^n_w\}+\rho_{H}\{P^n_*,P^n_{\hat{{\boldsymbol{\theta}}}}\}\right]^2\\&\geq \frac{1}{4}\left[\rho_H\{P^n_w,P^n_{\hat{{\boldsymbol{\theta}}}}\}\right]^2,
\end{flalign*}where the first line follows from the Cauchy-Schwartz inequality
and the second line from the triangle inequality. Applying equations
\eqref{hell1} and \eqref{hell2}, we then obtain \begin{flalign*}
\rho^2_{H}\{P^n_w,P^n_{\hat{{\boldsymbol{\theta}}}}\}&\leq {\epsilon}+o_{p}(1).
\end{flalign*}
Recall that, for probability distributions $P,G$,
\[
0\leq\rho^{2}_{TV}\{P,G\}\leq4\cdot\rho^{2}_{H}\{P,G\}.
\]
Applying this relationship between $\rho^{2}_{H}$ and $\rho^{2}_{TV}$, yields
the stated result.
	
\end{proof}}

\section{Computational Details}

\subsection{Setting $w_{n}$ for the CRPS rule\label{Appen:w_crps}}

{\normalsize Let $\Sigma _{w,n}$ denote the posterior covariance matrix of $%
\bm{\theta}$ calculated under the FBP posterior density $\pi _{w}\left( 
\boldsymbol{\theta }|\mathbf{y}_{n}\right) $, and which is associated with
an arbitrary choice of $w_{n}$. As discussed in the main text, the role of
the tuning sequence $w_{n}$ is to control the relative weight given to the
sample score criterion and {the prior} within the update. The choice of }$%
w_{n}$ {\normalsize is subjective in nature, in particular when the score
criterion cannot be interpreted as a likelihood function, as is the case for
the CRPS scoring rule. In this case then, we set $w_{n}$ }so that the rate%
{\normalsize \ of posterior {update} of CRPS-based }FBP is \textit{%
comparable\ }to that of exact (likelihood-based) Bayes.

To this end, let $\psi (n)$ denote the trace of the covariance matrix of $%
\boldsymbol{\theta }$ calculated under the exact Bayesian posterior in %
\eqref{exact_Bayes}. We then propose finding a sequence $w_{n}^{\ast }$ such
that the trace of $\Sigma _{w,n}$, when calculated under this sequence,
satisfies tr$\left[ \Sigma _{w^{\ast },n}\right] =\psi \left( n\right) .$
However, directly computing such a $w_{n}^{\ast }$ would entail solving the%
{\normalsize \ computationally intensive optimization problem, 
\begin{equation*}
w_{n}^{\ast }=\underset{w_{n}\in W}{\text{argmin}}\left( \left\{ \text{tr}%
\left[ \Sigma _{w,n}\right] -\psi \left( n\right) \right\} \right) ^{2}.
\end{equation*}%
{Hence,} {rather than pursuing} $w_{n}^{\ast }$ {directly, }we {seek} a
computationally efficient {approximation }}$\widetilde{w}_{n}^{\ast }$%
{\normalsize , defined as 
\begin{equation}
\widetilde{w}_{n}^{\ast }=\frac{\mathbb{E}_{\pi \left( \bm\theta |\mathbf{y}%
_{n}\right) }\left[ \log p\left( \mathbf{y}_{n}|\bm{\theta}\right) \right] }{%
\mathbb{E}_{\pi \left( \bm\theta |\mathbf{y}_{n}\right) }\left[ S_{n}(P_{%
\boldsymbol{\ \theta }}^{n})\right] },  \label{Eq:settingw1}
\end{equation}%
where $\mathbb{E}_{\pi \left( \bm\theta |\mathbf{y}_{n}\right) }$ indicates
an expectation computed under {(\ref{exact_Bayes})}},{\normalsize \textbf{\ }%
}$\log p\left( \mathbf{y}_{n}|\bm{\theta}\right) =\sum_{t=0}^{n-1}S_{\text{LS%
}}(P_{\boldsymbol{\theta }}^{t},y_{t+1})$ and $S_{n}(P_{\boldsymbol{\ \theta 
}}^{n})=\sum_{t=0}^{n-1}S_{\text{CRPS}}(P_{\boldsymbol{\theta }%
}^{t},y_{t+1}) $.{\normalsize \ To see how this sub-optimal choice of $w_{n}$
resembles $w_{n}^{\ast }$, consider the optimally-scaled scoring function }$%
S_{n}^{\ast }(P_{\boldsymbol{\theta }}^{n})=$ {\normalsize $w_{n}^{\ast }$}$%
S_{n}(P_{\boldsymbol{\theta }}^{n})${\normalsize . {Substituting} }$S_{n}(P_{%
\boldsymbol{\theta }}^{n})=\frac{S_{n}^{\ast }(P_{\boldsymbol{\theta }}^{n})%
}{{\normalsize w_{n}^{\ast }}}${\normalsize \ {into (\ref{Eq:settingw1})
yields} }

{\normalsize \ 
\begin{equation}
\widetilde{w}_{n}^{\ast }=w_{n}^{\ast }\frac{\mathbb{E}_{\pi \left( \bm%
\theta |\mathbf{y}_{n}\right) }\left[ \log p\left( \mathbf{y}_{n}|\bm{\theta}%
\right) \right] }{\mathbb{E}_{\pi \left( \bm\theta |\mathbf{y}_{n}\right) }%
\left[ S_{n}^{\ast }(P_{\boldsymbol{\theta }}^{n})\right] }.
\label{Eq:settingw2}
\end{equation}%
{From (\ref{Eq:settingw2})} we see that as long as the posterior means of
the optimally-scaled score }$S_{n}^{\ast }(P_{\boldsymbol{\theta }}^{n})$%
{\normalsize \ and $\log p\left( \mathbf{y}_{n}|\bm{\theta}\right) $ are
reasonably similar then, }$\widetilde{w}_{n}^{\ast }\approx $ {\normalsize $%
w_{n}^{\ast }$}, {and }{\normalsize $\text{tr}\left[ \Sigma _{\widetilde{w}%
_{n}^{\ast },n}\right] \approx \text{tr}\left[ \Sigma _{w^{\ast },n}\right]
=\psi \left( n\right) $ {as a consequence}. This assumption is not
unrealistic, since {\ - by construction - }}$S_{n}^{\ast }(P_{\boldsymbol{%
\theta }}^{n})${\normalsize \ and $\log p\left( \mathbf{y}_{n}|\bm{\theta}%
\right) $ are {score} functions {for which the respective sums of posterior
variances are equal}. Additionally, this sequence }$\widetilde{w}_{n}^{\ast
} ${\normalsize \ converges, }as{\normalsize \textbf{\ }}$n\rightarrow
\infty $,{\normalsize \ to a constant $C,$ and Theorem \ref{thm2} }thus%
{\normalsize \textbf{\ }still applies}. Of course, in practice, $\widetilde{w%
}_{n}^{\ast } $ itself needs to be estimated via draws from (\ref%
{exact_Bayes}). However, as this computation needs to be performed only
once, a sufficiently accurate estimate of $\widetilde{w}_{n}^{\ast }$ can be
produced via a large enough number of posterior draws. To keep the notation
streamlined in the main text, in Section \ref{overview} we simply record the
formula in (\ref{Eq:settingw1}) as the setting for $w_{n}$, and refer to the
simulation-based estimate{\normalsize \ used in the computations as }$%
\widehat{{\normalsize w}}${\normalsize $_{n}.$}

\subsection{Computational scheme\label{comp}}

For the predictive classes i) and ii) all posterior {updates} are performed
numerically, using a Metropolis-Hasting (MH) scheme to select draws of the
underlying $\boldsymbol{\theta }$. The MH acceptance ratio, at iteration $j$%
, is of the form 
\begin{equation}
\alpha =\text{min}\left\{ 0,\frac{\exp \left[ \widehat{w}_{n}S_{n}(P_{%
\boldsymbol{\theta }^{(c)}}^{n})\right] \pi \left( \boldsymbol{\theta }%
^{(c)}\right) }{\exp \left[ \widehat{w}_{n}S_{n}(P_{\boldsymbol{\theta }%
^{(j-1)}}^{n})\right] \pi \left( \boldsymbol{\theta }^{(j-1)}\right) }\times 
\frac{q\left( \boldsymbol{\theta }^{(j-1)}|\boldsymbol{\theta }^{(c)}\right) 
}{q\left( \boldsymbol{\theta }^{(c)}|\boldsymbol{\theta }^{(j-1)}\right) }%
\right\} ,  \label{mh_ratio}
\end{equation}%
where $q(.)$ denotes the candidate density function, $\boldsymbol{\theta }%
^{(c)}$ the draw from $q(.)$, $\boldsymbol{\theta }^{(j-1)}$ the previous
draw in the chain, {and }$\widehat{w}_{n}${\ the scale factor defined in (%
\ref{w_hat})}.

For the Gaussian ARCH(1) predictive class, the candidate density\ is a
truncated normal, $q\left( \boldsymbol{\theta }^{(c)}|\boldsymbol{\theta }%
^{(j-1)}\right) \propto \phi \left( \boldsymbol{\theta }^{(c)};\boldsymbol{\
\theta }^{(j-1)},\sigma _{q}^{2}I_{3}\right) I\left\{ \theta _{3}^{(c)}\in
\lbrack 0,1),\theta _{2}^{(c)}>0\right\} ,$ where $\phi $ denotes the normal
density function, $I_{3}$ is the three-dimensional identity matrix, and $%
\sigma _{q}$ controls the step size of the proposal, initiated at $\sigma
_{q}=0.05$ and then set adaptively to target acceptance rates between $30\%$
and $70\%$.

For the Gaussian GARCH(1,1) predictive class, the four elements of $%
\boldsymbol{\theta }$ are not sampled in one block; instead, they are
randomly assigned to two pairs at the beginning of each iteration. The pairs
are then sampled, one pair conditional on the other, with the
two-dimensional candidate density in each case equal to the product of two
independent scalar normals, truncated where appropriate to reflect the
parameter restrictions for the GARCH(1,1) model: $\theta _{2}>0$, $\theta
_{3}\in \lbrack 0,1)$,\textbf{\ }$\theta _{3}+\theta _{4}<1,$ and $\theta
_{4}\in \lbrack 0,1)$. The random assignment into pairs, plus the use of
independent normal candidates for each individual parameter, means that the
step size - and hence, targeted acceptance rates \ - for each parameter can
be controlled separately. The same form of adaptive scheme as described
above is used for each parameter. Because the sampling is conducted via two
MH steps, the formula for the acceptance ratio in (\ref{mh_ratio}) is
modified accordingly.\footnote{%
See \cite{roberts2009examples} for details on adaptive MCMC, and \cite%
{smith2015copula} for an illustration of its use with random allocation.}

Finally, for predictive class iii) the posterior {update} for the scalar
parameter $\theta _{1}$ is approximated numerically over a fine grid for the
scalar parameter $\theta _{1}$, on the unit interval. Draws of $\theta _{1}$
from the numerically evaluated posterior then define draws of predictives
from the posterior defined over the mixture predictive class.

\subsection{Setting $w_{n_{i}}$ for the MSIS rule\label{msis}}

For the {M4 competition }example in Section \ref{Sect:M4}, {we require the
setting of }$w_{n_{i}}$ for $S_{n_{i}}(P_{\boldsymbol{\theta }_{i}}^{n_{i}})$
{defined in terms of the MSIS scoring rule in (\ref{MSIS}), for all }$%
i=1,2,...,23,000.${\ In principle, the same approach could be adopted as
described in Section \ref{overview} for the case of the CRPS scoring rule.
However, that would entail {23,000 preliminary runs }of MCMC to estimate
each }$\widehat{{\normalsize w}}${\normalsize $_{n_{i}}$}. Hence, we propose
a more computationally efficient way of setting each $w_{n_{i}}$ for this
example. The basic idea is to control for the scale of $S_{n_{i}}(P_{%
\boldsymbol{\theta }_{i}}^{n_{i}})$ directly by imposing\ $%
w_{n_{i}}S_{n_{i}}\left( P_{\boldsymbol{\theta ^{+}}}^{n_{i}}\right) =\gamma
\left( n_{i}\right) $, where $\gamma \left( n_{i}\right) $ is a user-defined
deterministic function that serves as a benchmark scale for the term
entering the exponential function in the expression for the FBP posterior,
and $\boldsymbol{\theta }_{i}^{+}$ is a representative parameter value,
taken as the MLE of $\boldsymbol{\theta }_{i}$ in the empirical example. Via
a series of arguments that amount to adopting an approximating Gaussian
model, we set $\gamma \left( n_{i}\right) =-\frac{1}{2}n_{i}d_{\theta _{i}}$%
, where $d_{\theta _{i}}$ is the dimension of $\boldsymbol{\theta }_{i}$.
With this choice of $\gamma (n_{i})$, the scaling constant $w_{n_{i}}$ is
thus set to $w_{n_{i}}=-n_{i}d_{\theta _{i}}/2S_{n_{i}}\left( P_{\boldsymbol{%
\theta _{i}^{+}}}^{n_{i}}\right) $. This choice of $w_{n_{i}}$ converges
asymptotically to a constant $C$ and, as such, Theorem \ref{thm2} still
applies.

\section{Predictive assessment of expected shortfall}

\label{app:ES}Using generic notation for the time being,\textbf{\ }let $%
y_{n+1}$ be a random variable with support {$\mathcal{Y}$ and} with true
predictive\textbf{\ }distribution $F$, conditional on time\textbf{\ }$n$
information. Let $\mathcal{P}^{n}$ denote a class of predictive
distributions for\textbf{\ }$y_{n+1}$, with arbitrary element $P^{n}\in 
\mathcal{P}^{n}$, and let $T_{\alpha }(P^{n})=\left( \text{VaR}_{\alpha
}(P^{n}),\text{ES}_{\alpha }(P^{n})\right) ^{\prime }$, denote a functional
of VaR and ES constructed from\textbf{\ }$P^{n}$, where:%
\begin{equation*}
\text{VaR}_{\alpha }(P^{n}):=\inf \{y\in \mathcal{Y}:P^{n}(y)\geq \alpha \}
\end{equation*}%
and 
\begin{equation*}
\text{ES}_{\alpha }(P^{n}):=\frac{1}{\alpha }\int_{0}^{\alpha }\text{Var}%
_{a}(P^{n})\text{d}a.
\end{equation*}%
Note that this definition of ES$_{\alpha }$ is equivalent to our earlier
definition of\textbf{\ }ES$_{\alpha }$ in Section \ref{Sect:sect433} as a
conditional expectation of $y_{n+1}$ with respect to a given\textbf{\ }$%
P^{n}.$

\cite{ziegel2020robust} propose to measure the accuracy with which ES can be
predicted, across any $P^{n}\in \mathcal{P}^{n}$, using the following
(positively-oriented) scoring function:{\ for $y\in \mathcal{Y}$, }%
\begin{eqnarray}
f_{\eta }\left[ T_{\alpha }(P^{n}),{y}\right] &=&-I\left[ \eta \leq \text{ES}%
_{\alpha }\left( P^{n}\right) \right] \left\{ (1/\alpha )I\left[ {y}\leq 
\text{VaR}_{\alpha }\left( P^{n}\right) \right] \left[ \text{VaR}_{\alpha
}\left( P^{n}\right) -{y}\right] \right.  \notag \\
&&\left. -\left[ \text{VaR}_{\alpha }\left( P^{n}\right) -\eta \right]
\right\} -I\left( \eta \leq {y}\right) \left( {y}-\eta \right) ,
\label{Eq:ESscore}
\end{eqnarray}%
where $\eta $ is a known scalar. It is shown in \cite{fissler2016higher}
that {the scoring function in (\ref{Eq:ESscore})} is \textit{consistent} for 
$T_{\alpha }(P^{n})$, {for any }$\eta \in \mathbb{R}$, {and, hence, is
consistent} for ES. The consistency of $f_{\eta }$ guarantees that, for all $%
P^{n}\in \mathcal{P}^{n}$, 
\begin{equation}
\int_{\mathcal{Y}}f_{\eta }\left[ T_{\alpha }(F),y\right] \text{d}F(y)\geq
\int_{\mathcal{Y}}f_{\eta }\left[ T_{\alpha }(P^{n}),y\right] \text{d}F(y),
\label{Eq:consistency}
\end{equation}%
where it is assumed that $\int_{\mathcal{Y}}f_{\eta }\left[ T_{\alpha
}(P^{n}),y\right] $d$F(y)$ exists. Consistency then implies that if we
measure predictive accuracy according to $f_{\eta }\left[ T_{\alpha
}(P^{n}),y\right] $, the forecaster's most sensible course of action (to
obtain accurate predictions for ES) is to quote the element $P^{n}$ such
that the functional $T_{\alpha }(P^{n})$ is `closest' to the true but
unknown predictive functional,\textbf{\ }$T_{\alpha }(F).$

{Each value of $\eta \in \mathbb{R}$ within (\ref{Eq:ESscore}) yields a
scoring function that is consistent for ES. Collectively, this set of
functionals, one for each value of $\eta \in \mathbb{R}$, can be employed to
assess {the }}ES{{\textbf{\ }predictive} \textit{dominance} of one
predictive over another. To this end, denote the FBP mean predictive
distribution as $P_{w}^{n}$, and the exact Bayes mean predictive
distribution as $P_{\text{EB}}^{n}$. Following \cite{kruger2020generic}, the
predictive \textit{dominance}} of FBP over exact Bayes will be {in evidence}
if, for all $\eta \in \mathbb{R}$, 
\begin{equation}
\int_{\mathcal{Y}}f_{\eta }\left[ T_{\alpha }(P_{w}^{n}),y\right] \text{d}%
F(y)\geq \int_{\mathcal{Y}}f_{\eta }\left[ T_{\alpha }(P_{\text{EB}}^{n}),y%
\right] \text{d}F(y).  \label{Eq:dominance}
\end{equation}%
\cite{ehm2016quantiles} propose the use of Murphy diagrams as a
visualization tool {for revealing }predictive dominance in empirical
applications. {The idea is} to create a plot whose $x$-axis indicates the
value $\eta $ that defines $f_{\eta }$, while the $y$-axis indicates a
sample estimate of the corresponding difference between the predictions made
under the two methods: 
\begin{equation*}
\Delta _{\eta }(P_{{w}}^{n},P_{\text{EB}}^{n})=\int_{\mathcal{Y}}\left\{
f_{\eta }\left[ T_{\alpha }(P_{w}^{n}),y\right] -f_{\eta }\left[ T_{\alpha
}(P_{\text{EB}}^{n}),y\right] \right\} \text{d}F(y).
\end{equation*}%
Dominance will be visually in evidence when $\Delta _{\eta }\geq 0$ for all $%
\eta $.

\end{document}